\documentclass[ aps,
                pra,
                showpacs,
                amssymb,
                nofootinbib,
                superscriptaddress,
                twocolumn,
                longbibliography,noeprint, 10pt]{revtex4-2}

\usepackage[ansinew]{inputenc}
\usepackage{bm,bbm,dsfont}
\usepackage{amsbsy,amsthm,amssymb,amsfonts,amsmath,mathtools}
\usepackage{graphicx}
\usepackage{physics}
\usepackage{xcolor}
\usepackage{nicefrac}
\usepackage{MnSymbol}
\usepackage{enumerate}
\usepackage{float}
\usepackage{comment}
\usepackage{braket}
\usepackage{soul}
\usepackage[colorlinks]{hyperref}

\usepackage[normalem]{ulem} 
\newcommand{\stkout}[1]{\ifmmode\text{\sout{\ensuremath{#1}}}\else\sout{#1}\fi}

\makeatletter
\newcommand\org@hypertarget{}
\let\org@hypertarget\hypertarget
\renewcommand\hypertarget[2]{%
  \Hy@raisedlink{\org@hypertarget{#1}{}}#2%
  }

\makeatother

\hypersetup{
	bookmarksnumbered,
	pdfstartview={FitH},
	citecolor={darkgreen},
	linkcolor={darkred},
	urlcolor={darkblue},
	pdfpagemode={UseOutlines}}
\definecolor{darkgreen}{RGB}{50,190,50}
\definecolor{darkblue}{RGB}{0,0,190}
\definecolor{darkred}{RGB}{238,0,0}

\newcommand{\eqnref}[1]{Eq.~(\ref{#1})}
\newcommand{\eqnsref}[2]{Eqs.~(\ref{#1}) and (\ref{#2})}
\newcommand{\figref}[1]{Fig.~\ref{#1}}

\newcommand{\secref}[1]{Sec.~\ref{#1}}
\newcommand{\appref}[1]{App.~\ref{#1}}

\newcommand{\mc}[1]{\mathcal{#1}}
\newtheorem{theorem}{Theorem}[section]

\def\one{{\mbox{$1 \hspace{-1.0mm}  {\bf l}$}}}						

\renewcommand{\thesection}{\Roman{section}}
\renewcommand{\thesubsection}{\Roman{section}.\Alph{subsection}}
\renewcommand{\thesubsubsection}{\Roman{section}.\Alph{subsection}.\arabic{subsubsection}}
\makeatletter
\renewcommand{\p@subsection}{}
\renewcommand{\p@subsubsection}{}
\newtheorem{corollary}{Corollary}
\makeatother

\begin{document}

\title{Quantum sequential parameter testing}

\author{Simon Morelli}
\email{simon.morelli@tuwien.ac.at}
\affiliation{Technische Universit{\"a}t Wien, Atominstitut \& Vienna Center for Quantum Science and Technology (VCQ),  Stadionallee 2, 1020 Vienna, Austria}

\author{Ricard Ravell Rodr\'iguez}
\email{ricard.ravell@icfo.eu}
\affiliation{ICFO-Institut de Ci\`encies Fot\`oniques, The Barcelona Institute of
Science and Technology, 08860 Castelldefels (Barcelona), Spain}

\author{John Calsamiglia}
\affiliation{F\'isica Te\`orica: Informaci\'o i Fen\`omens Qu\`antics, Departament de F\'isica, Universitat Aut\`onoma de Barcelona, 08193
Bellaterra (Barcelona), Spain}

\author{Ram\'on Mu\~{n}oz-Tapia}
\affiliation{F\'isica Te\`orica: Informaci\'o i Fen\`omens Qu\`antics, Departament de F\'isica, Universitat Aut\`onoma de Barcelona, 08193
Bellaterra (Barcelona), Spain}

\author{Gael Sent\'is}
\affiliation{F\'isica Te\`orica: Informaci\'o i Fen\`omens Qu\`antics, Departament de F\'isica, Universitat Aut\`onoma de Barcelona, 08193
Bellaterra (Barcelona), Spain}

\author{Michalis Skotiniotis}
\email{mskotiniotis@onsager.ugr.es}
\affiliation{Quantum Thermodynamics and Computation Group,
Department of Electromagnetism and Condensed Matter, Universidad de Granada, 18071 Granada, Spain}
\affiliation{Instituto Carlos I de F\'isica Te\'orica y Computacional, Universidad de Granada, 18071 Granada, Spain}

\begin{abstract}
Sequential strategies in hypothesis testing use a variable number of measurement rounds, allowing a decision to be made as soon as the observed data 
provide a prescribed level of error tolerance. Although sequential testing is well established for a discrete set of hypotheses, extending this framework to a 
continuous parameter space poses additional challenges and has remained largely unexplored. In this article, we introduce sequential \textit{parameter 
testing}, a framework for determining an unknown parameter up to a prescribed tolerance by ruling out sufficiently distant competing values with a 
target error tolerance. We clarify its operational distinction from conventional parameter estimation and develop sequential tests for continuous families 
of hypotheses, introducing the \textit{twin-peaks test} as a natural and computationally efficient analog of sequential likelihood-ratio testing. We 
apply the framework to two paradigmatic quantum tasks: testing the phase and the purity of a qubit. For phase testing, we show numerically that adaptive projective 
measurements achieve the same average sample cost as collective covariant measurements with a fixed number of copies. For purity testing, local 
measurements are optimal, and sequential parameter testing yields significant average sample savings over fixed sample-size protocols. Our results 
establish parameter testing as an operationally meaningful framework for resource-efficient certification tasks involving continuous parameters.
\end{abstract}

\maketitle

\section{Introduction}

Deciding accurately between possible hypotheses that explain observed data is the foundation of modern science.
A paradigmatic task is to determine the value of a parameter that governs the behavior of an otherwise known model based on 
empirical data. When the parameter can only assume two or more discrete values, we speak of (simple) \textit{hypothesis testing}, whereas if 
the parameter is continuous, the task falls under the umbrella of \textit{parameter estimation}. In both settings, the observed data are used 
to draw an inference about the underlying model, either by selecting a hypothesis or by reporting an estimate of the parameter. The quality of 
this inference is assessed through a task-dependent figure of merit (FOM). 

Historically, the sample size, i.e., the number of measurement rounds, was treated as fixed in advance. Because each measurement outcome is 
inherently probabilistic, and the resulting error in the selected hypothesis varies from one experimental run to another, one typically considers 
the mean squared error (MSE)---in the local estimation scenario---or average mean squared error (AvMSE) in the global estimation scenario. In 1945, 
Abraham Wald~\cite{wald_sequential_1945} presented an alternative framework for hypothesis testing called \textit{sequential analysis}, in which 
the roles of sample size and error threshold are effectively reversed:  one prescribes a target error, and the sample size becomes a random variable 
whose value is determined by the observed measurement outcomes. An additional benefit of this approach is an average saving in resources, as test runs can 
potentially stop earlier if the data allow for a decision with the target error tolerance~\cite{Wald1948}.

In this article, we propose sequential methods for the estimation of an unknown parameter $\vartheta$, which we call \textit{parameter testing}, whose 
goal is to guarantee a parameter value to within a prescribed tolerance, $\delta$, and error $\epsilon$. Once an estimate, $\hat\vartheta$, 
is reported, parameter values within a $\delta-$neighborhood $B_\delta(\hat\vartheta)$ are regarded as acceptable, whereas values outside 
this region are treated as competing alternatives to be ruled out. Unlike conventional parameter estimation, where the quality of an estimate is typically 
quantified by a distance-based loss function (such as the MSE or AvMSE), parameter testing is designed for scenarios in which the usefulness of an 
estimate is determined by whether it lies within a prescribed tolerance. The objective is therefore not to minimize estimation error \textit{per se}, 
but to determine, whether the unknown parameter belongs to an acceptable region while excluding sufficiently distant alternatives. Examples where such 
an objective appears naturally include image-guided surgery, where a target location needs only be identified within the spatial tolerance required 
for the planned intervention~\cite{Shamir2009,Lin2023}, and engineering metrology and conformity assessment, where manufactured components 
are accepted or rejected according to whether measured characteristics satisfy prescribed tolerances~\cite{ISOIECGuide984}.

We introduce three sequential procedures for parameter testing: the \textit{complement test}, which pits the $\delta$-neighborhood
$B_\delta(\vartheta)$ around each parameter value against its complement; the \textit{concentration test}, which compares the posterior
density at a candidate parameter value with the average posterior density outside its $\delta$-neighborhood; and the
\textit{twin-peaks test}, which compares the posterior density at the maximum-a-posteriori estimate $\hat{\vartheta}$ with the largest
posterior density attained outside $B_\delta(\hat{\vartheta})$. The three protocols are depicted in Fig.~\ref{fig:diff-tests}.  

Among these three tests, the \textit{twin-peaks test} is most closely aligned with the sequential likelihood-ratio test in hypothesis testing: it stops 
once the best estimate is separated, in likelihood-ratio terms, from every parameter value outside $B_\delta(\vartheta)$. The test is simple 
to implement, avoiding conceptual burdens and the computational cost of repeatedly computing posterior distributions, whilst retaining a direct 
comparison with the strongest competing hypothesis.

We implement these sequential parameter-testing methodologies in a quantum setting where the unknown parameter is encoded in a family of 
quantum states $\{\rho_\vartheta\}_{\vartheta\in\Theta}$, with $\Theta$ being a convex and compact space. We look at the paradigmatic 
cases of determining the phase and the Bloch vector length of a qubit. We derive asymptotic stopping times of the \textit{twin-peaks test} and 
find quantum measurement strategies that achieve these times. For the case of phase-testing, we find that adaptive projective measurements 
perform best and yield the same average number of samples as collective measurement strategies using a fixed predetermined number of samples.  
For the case of purity, we find that the stopping time depends on the value of the purity and, assuming that the direction of the Bloch vector is known, 
local projective measurements achieve the asymptotic stopping time. If the direction of the Bloch vector is not known, we show that sequential testing 
procedures that process the quantum states in batches of size $M$ achieve the same stopping time---in the limit $M\to\infty$--- as for the local 
projective measurements with the direction of the Bloch vector assumed known. Unlike the case of phase-testing, our \textit{twin-peaks test} for 
purity testing can exhibit significant gains in the number of samples needed compared to fixed sample-size strategies. Our methodology and results are 
of particular relevance for several key tasks in quantum information where the success/failure of the task exhibits a sharp threshold between useful 
and non-useful resources. Such tasks include entanglement bounds and entanglement distillation thresholds~\cite{Werner1989,Horodecki1995}, the 
quantum bit error rate (QBER) or CHSH threshold for the extractability of a secret key in quantum key 
distribution~\cite{ShorPreskill2000,AcinGisinMasanes2006,PironioAcinMasanes2010}, fault-tolerance thresholds for quantum error-correction~\cite{AharonovBenOr1997,RaussendorfHarringtonGoyal2006} and many more.

The paper is structured as follows. We start with an overview of Bayesian estimation theory, sequential hypothesis 
testing, and quantum hypothesis testing in Sec.~\ref{sec:framework}. We then proceed by discussing the difference between our approach to 
conventional parameter estimation in Sec~\ref{sec:cont-hyp-testing} and possible ways to take the continuous limit of multihypothesis 
testing, before introducing the \textit{twin-peaks test} in Sec.~\ref{sec:twin-peaks_test}. In Sec.~\ref{sec:stopping times}, we give an 
analytical estimate of the stopping times of the \textit{twin-peaks test}, and compare sequential strategies to deterministic ones in 
Sec.~\ref{sec:comparison_fixed_size}, extending the analysis also to the quantum setting. Finally, we test the developed framework with two 
paradigmatic examples---the estimation of an unknown phase and unknown purity of a qubit---in Sec.~\ref{sec:phase} and Sec.~\ref{sec:purity}.
\begin{figure*}[t]
\begin{center}
    \includegraphics[width=\linewidth]{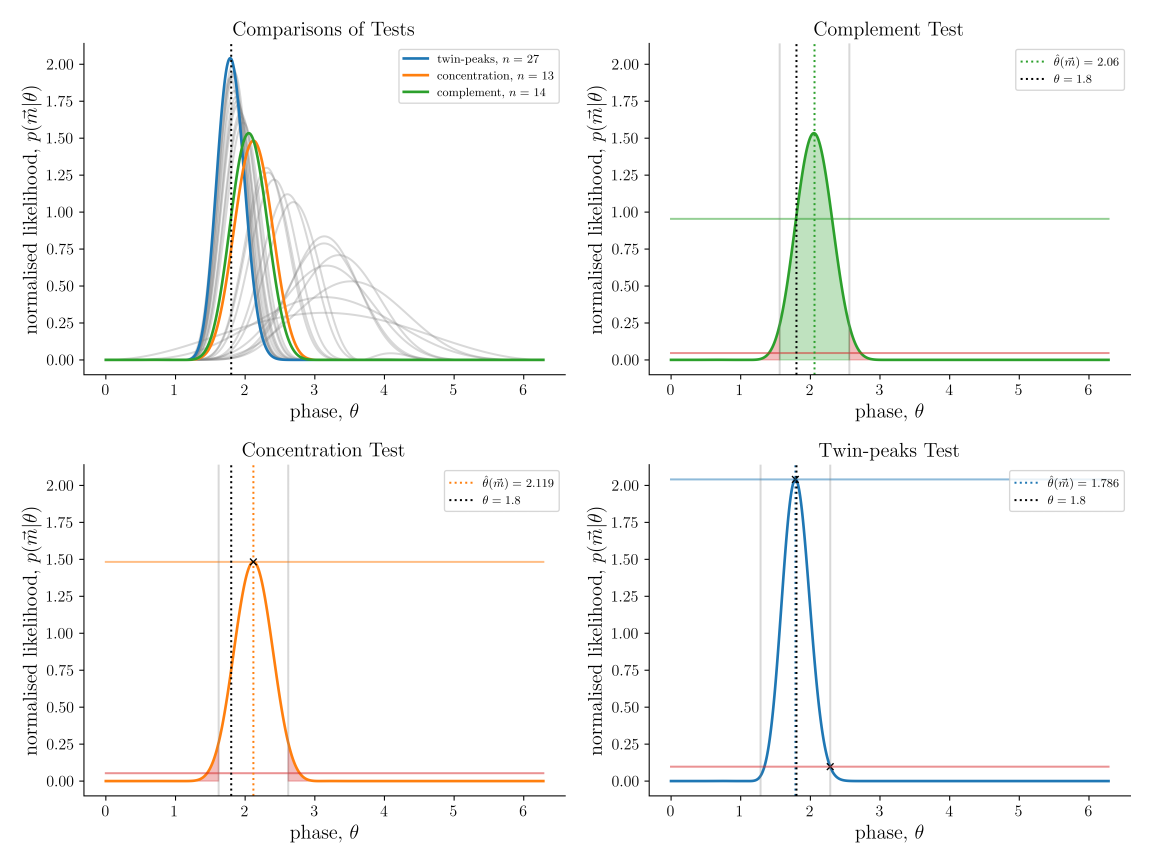}\\
    \caption{The figure shows the three different tests explained in this work. The first image shows the updated normalized likelihood of 
    each round as gray curves, with the three rounds where a test stops highlighted in color. The other three figures show the 
    stopping conditions of the three sequential tests introduced in this work in detail. In green, we see the \textit{complement test}, where 
    the probability of the parameter lying in the interval is compared to the probability that it lies outside said interval. In orange, the 
    \textit{concentration test}, where the maximum likelihood is compared to the average likelihood outside a given interval, and in blue, 
    the \textit{twin-peaks test}, where the maximum likelihood is compared against the highest likelihood outside of the interval. The gray vertical 
    lines indicate the interval, the black dotted line the true value $\vartheta=1.8$, and the colored dotted line the estimator of the corresponding test. 
    The colored horizontal line indicates the value of the test, and in red the compared value. For this example, we have chosen an error 
    $\epsilon=0.05$ and an interval size $2\delta=1$.}
    \label{fig:diff-tests}
\end{center}
\end{figure*}

\section{Framework}\label{sec:framework}
We begin by briefly reviewing the main concepts on which our framework is based. Both hypothesis testing and parameter estimation involve 
making inferences about unknown parameters that index a statistical model, either by selecting plausible values or by deciding between 
competing subsets of the parameter space. These parameters are not directly measurable, but are inferred from observed data via the 
likelihood function---specified by the model---which quantifies how probable the observed data are for different parameter values. The main 
difference between the two frameworks lies in the nature of the inference task. In hypothesis testing, the goal is to decide between 
competing hypotheses, which typically correspond to subsets of the parameter space (such as a null hypothesis versus an alternative), 
whereas in parameter estimation, the goal is to infer the value of the parameter itself. This difference is reflected in commonly used FOMs 
for the two tasks: error probabilities (e.g., type-I or type-II error rates) for the case of hypothesis testing, and 
distance measures in parameter space (such as the MSE) for parameter estimation.


\subsection{Parameter estimation}\label{sec:parameter_estiamtion}
Parameter estimation can be approached from either a frequentist or a Bayesian perspective and may involve either local or global figures of merit.
We begin by briefly reviewing the main concepts of parameter estimation that will serve as reference for the parameter-testing framework introduced
below.

In the frequentist approach, the parameter, $\vartheta$, is treated as fixed but unknown. The observed data, $m$, which may represent either a single 
outcome or a complete measurement record, are distributed according to the conditional probability distribution $p_\vartheta(m) = p(m\mid \vartheta)$. An estimator 
$\hat{\vartheta}(m)$ is a function of the data that assigns an estimate to the unknown parameter. Its performance is evaluated over hypothetical 
repetitions of the experiment at the same parameter value. A standard FOM is the pointwise MSE
	\begin{equation}
    		V(\hat{\vartheta}\mid\vartheta):=\int p(m\mid \vartheta)\bigl(\hat{\vartheta}(m)-\vartheta\bigr)^2\,\dd m \, .
	\label{eq:MSE}
	\end{equation}
For unbiased estimators, this quantity coincides with the variance.

For regular statistical models, the local sensitivity of the outcome distribution to changes in the parameter is quantified by the Fisher
information
	\begin{equation}
    		I(\vartheta):=\int p(m\mid \vartheta)\left[\partial_{\vartheta}\log p(m\mid \vartheta)\right]^2\dd m \, .
	\label{eq:FI}
	\end{equation}
Under the usual regularity and unbiasedness assumptions, the Cram\'er-Rao bound lower bounds the variance by~\cite{rao1992information,cramer1999mathematical},
	\begin{equation}
    		V(\hat{\vartheta}\mid\vartheta)\ge\frac{1}{I(\vartheta)}\, .
	\label{eq:CR_bound}
	\end{equation}
For independent observations the Fisher information is additive, and regular estimators such as the maximum-likelihood estimator can 
asymptotically attain the corresponding bound. Hence, the Fisher information provides a local, pointwise characterization of estimation precision.

For finite data, however, the problem is generally global, since several well-separated parameter values may remain compatible with the observations. 
Bayesian estimation provides a natural framework for this regime: the unknown parameter is assigned a prior distribution, $\pi(\vartheta)$, that captures 
the initial knowledge of the parameter. Upon observing data $m$, our prior knowledge about $\vartheta$ is updated to the posterior distribution according 
to Bayes's rule~\cite{DAgostini2003,Bolstad2009,AmaralTurkmanPaulinoMuller2019}
    	\begin{equation}\label{eq:Bayes}
        	p(\vartheta\mid m)=\frac{p(m\mid\vartheta)\, \pi(\vartheta)}{\int p(m\mid\vartheta) \, \pi(\vartheta)\,\dd \vartheta}=\frac{p(m\mid\vartheta)\,
        	\pi(\vartheta)}{p(m)}\,.
    	\end{equation}
The posterior, therefore, quantifies the remaining uncertainty over the full parameter space. The AvMSE associated with an estimator 
$\hat{\vartheta}(m)$ is
    	\begin{equation}
        	\bar V = \iint p(\vartheta,m)\bigl(\vartheta-\hat{\vartheta}(m)\bigr)^2\,\dd\vartheta\,\dd m\, ,
    	\label{eq:AvMSE}
    	\end{equation}
and the estimator minimizing the AvMSE is the posterior mean,
    	\begin{equation}
        	\hat{\vartheta}(m) = \int p(\vartheta\mid m)\,\vartheta\,\dd\vartheta\, .
    	\label{eq:posteriormean}
    	\end{equation}
For this estimator, the outcome-dependent posterior variance is
    	\begin{equation}
        	V(m) = \int p(\vartheta\mid m)\bigl(\vartheta-\hat{\vartheta}(m)\bigr)^2\,\dd\vartheta\, ,
    	\label{eq:post_var}
    	\end{equation}
which captures the level of certainty in the estimation and can be used to monitor the progress of the estimation procedure.
Averaging over all measurement outcomes returns the AvMSE.


\subsection{Sequential hypothesis testing}\label{sec:sequential}
In hypothesis testing, two or more hypotheses are compared to determine which is the most compatible with the observed data. Sequential strategies are 
characterized by a variable number of test rounds. This has two main advantages. First, once a hypothesis is accepted, it is done so with the desired 
error in the decision (strong error condition), as opposed to the weaker condition to reach this prescribed error on average. Second, a variable sample size 
can lead to substantial savings in resources.

Let us illustrate this with a toy example. Imagine we want to test the two hypotheses that a coin flip results in `heads' with probability 1 
(hypothesis 1) against a probability of 0.9 (null hypothesis). Clearly, as soon as we observe the outcome `tail', we can discard hypothesis 1 and accept 
the null hypothesis. Conditioned that the null hypothesis holds, we will, on average, observe a `tail' once every 10 coin flips. Testing with a fixed 
sample size of 10 flips will, in some runs, result in an outcome `tail' earlier than the 10th flip, whereas we might not see `tail' at all on other 
runs, leading us to accept the wrong hypothesis. Testing with a variable sample size, we can stop when we observe a `tail' for the first time. Only 
if we do not observe `tail' for a long time, we might wrongly accept hypothesis 1.

Wald~\cite{wald_sequential_1945} formalized this idea in the \textit{sequential probability ratio test} (SPRT) formulated in the following way.
Let $H_x$ be the hypothesis that we would like to test and $\Vec{m}_k=(m_1,\dots,m_k)$ be the outcomes of the first $k$ measurements. 
We want to compare hypothesis $H_x$ with the mutually exclusive null hypothesis $H_y$, where we assume that 
the two hypotheses are exhaustive. We require that, whenever the test accepts hypothesis $H_x$, the (posterior) probability of having made 
the correct decision is greater than a desired error, i.e.,
	\begin{equation}
    		P(H_x\mid \Vec{m}_k)\ge 1-\epsilon_x\, .
	\label{eq:posterior_error}
	\end{equation}
Since the two hypotheses are mutually exclusive and collectively exhaustive, this is equivalent, by Bayes' rule, to the stopping condition, 
	\begin{equation}
    		\frac{P(\Vec{m}_k\mid H_x)\pi(H_x)}{P(\Vec{m}_k\mid H_y)\pi(H_y)}=\frac{P(H_x\mid \Vec{m}_k)}{P(H_y\mid \Vec{m}_k)}\ge\frac{1-\epsilon_x}{\epsilon_x}
    		=:A_{\epsilon_x}\, .
	\label{eq:stopping_condition_binary}
	\end{equation}
After each measurement round, we check whether the stopping condition is satisfied for either hypothesis. If neither condition is met, the test continues, 
and a new datum is collected. The test therefore stops only when the accepted hypothesis reaches the required level of error tolerance. This run-by-run 
guarantee is known as a \textit{strong error condition}~\cite{Slussarenko_2017}. By contrast, fixed-sample-size strategies do not generally satisfy the same run-by-run 
stopping condition. Their performance is instead specified through error probabilities averaged over the possible measurement outcomes, either separately under each 
hypothesis or according to a chosen prior over the hypotheses. This is commonly referred to as a \textit{weak error condition}.

The stopping time, $N$, is defined as the minimal $k$, such that one of the stopping conditions of \eqnref{eq:stopping_condition_binary} is met
    \begin{equation}
        N=\inf\left\{k:\frac{P(\Vec{m}_k\mid H_x)}{P(\Vec{m}_k\mid H_y)}\geq A_{\epsilon_x}\frac{\pi(H_y)}{\pi(H_x)}\right\}\, .
    \label{eq:stopping_time}
    \end{equation}
For two simple hypotheses and i.i.d. observations, the log-likelihood ratio has, under $H_x$, mean increment 
    \begin{equation} 
		D\!\left( p_{x}\,\middle\|\,p_{y} \right) = \sum_m P(m\mid H_x) \log \frac{P(m\mid H_x)}{P(m\mid H_y)}, 
	\label{eq:KL_divergence} 
	\end{equation} 
which corresponds to the well-known relative entropy, or Kullback-Leibler (KL) divergence, between the two probability distributions, where $p_x(m)=P(m\mid H_x)$.
Consequently, in the small-error limit and neglecting the lower-order contribution from boundary overshoot, the expected stopping time 
conditioned on $H_x$ satisfies \cite{wald_sequential_1945,baum1994sequential}
	\begin{equation} 
		\mathds{E}(N\mid H_x) = \frac{ \log A_{\epsilon_x} + \log\!\left[ \pi(H_y)/\pi(H_x) \right] }{ D\!\left( p_{x}\,\middle\|\,p_{y} \right) } + 
		\mathcal{O}(1)\, . 
	\label{eq:asymptotic_stopping_time_binary} 
	\end{equation} 
Equivalently, 
	\begin{equation} 
		\mathds{E}(N\mid H_x) \sim \frac{\log(1/\epsilon_x) }{ D\!\left( p_{x}\,\middle\|\,p_{y} \right) } 
	\end{equation} 
as $\epsilon_x\to0$, provided that the prior odds remain fixed.

For multiple mutually exclusive and collectively exhaustive hypotheses, the strong-error condition can be imposed by comparing the posterior probability 
of $H_x$ with the total posterior probability of all competing hypotheses,
	\begin{equation}
    		\frac{P(H_x\mid\Vec{m})}{\sum\limits_{y\neq x}P(H_y\mid\Vec{m})}\geq A_{\epsilon_x}\, ,
	\label{eq:condition_sum}
	\end{equation}
where $A_{\epsilon_x}=(1-\epsilon_x)/\epsilon_x$. Since
	\begin{equation}
    		\sum_{y\neq x}P(H_y\mid\Vec{m})=1-P(H_x\mid\Vec{m}),
	\end{equation}
this condition is equivalent to the \emph{strong-error conditions}
	\begin{equation}
    		P(H_x\mid\Vec{m})\geq 1-\epsilon_x,
	\end{equation}
i.e., it guarantees the prescribed posterior probability of a correct decision~\cite{baum1994sequential,VeeravalliBaum1995}.

An alternative is to compare $H_x$ only with its strongest individual competitor,
	\begin{equation}
    		\frac{P(H_x\mid\Vec{m})}{\max\limits_{y\neq x}P(H_y\mid\Vec{m})}\geq A_x,
	\label{eq:condition_max}
	\end{equation}
where $A_x>1$ is an \textit{evidence threshold}. 
Unlike \eqnref{eq:condition_sum}, this \emph{maximum-competitor condition} does not directly bound 
the total posterior probability of an incorrect decision, since the posterior weights of several competing hypotheses may add up to a substantial value.
For the same numerical threshold, however, the maximum-competitor condition is generally satisfied earlier than the complement condition, because the 
probability of the strongest individual competitor cannot exceed the total probability of all competitors. This potentially reduces the stopping time, 
but the thresholds of the two tests should not be assigned the same error-probability interpretation without an additional calibration~\cite{Armitage50,Lorden77}.
Under standard regularity assumptions, both associated multi hypothesis sequential tests are asymptotically optimal, in the regime of small error probabilities 
or decision risks, with respect to the expected sample size~\cite{Draglia_1999}.


\subsection{Sequential quantum hypothesis testing}\label{sec:quantum}

We now review sequential quantum hypothesis testing between two fixed states.
In quantum hypothesis testing, simple hypotheses correspond to possible states of a quantum system, while the data are obtained from the outcomes of quantum measurements. 
The fixed-sample setting, where the number of copies of the unknown state is specified in advance, has a long history dating back to the pioneering works of 
Refs.~\cite{helstrom1969quantum,HOLEVO1973337,yuen1975optimum,HiaiPetz1991} and has since been extensively 
developed~\cite{OgawaNagaoka2000,Hayashi2002,Audenaert2007,Tsang_2012,Mosonyi_2014}. By contrast, sequential quantum hypothesis testing has 
only recently begun to receive significant 
attention~\cite{vargas2021quantum,Li_2022,Gasbarri_2024,perez2025bounds,rodasalichs2025sequentialanalysiscontinuousspinnoise,zecchin2025quantumsequentialuniversalhypothesis,simpson2026optimalerrorexponentscomposite}. These works demonstrate that the data-efficiency advantages of sequential analysis extend 
to quantum state discrimination, at least for binary simple hypotheses.

A distinctive feature of the quantum setting is that the measurement itself becomes part of the optimization problem.
Under hypotheses $H_0$ and $H_1$, every requested copy is prepared in one of two quantum states, respectively.
After each measurement round, the protocol either accepts one of the hypotheses or requests an additional copy.
The most general sequential quantum strategy need not measure the copies
individually. Conditioned on all previous ``continue'' outcomes, the first $k$
copies may remain in a joint state stored in a quantum memory. At round $k+1$,
a fresh copy is appended to this conditional state and a general quantum
instrument is applied to the enlarged system. The instrument has two
conclusive outcomes, associated with accepting $H_0$ or $H_1$, and a third
outcome that updates the quantum memory and continues the protocol. This
framework includes adaptive operations, weak measurements, quantum memories,
and collective measurements on an increasing number of copies
\cite{vargas2021quantum}.
In quantum metrology, the term \textit{sequential} is sometimes used for protocols where a fixed number of probes are 
interrogated consecutively while preserving quantum coherence throughout the
protocol~\cite{Giovannetti2006,Bavaresco2021,Bavaresco_2022,andre2026strategyoptimizationbayesianquantum}. In this article, \textit{sequential} refers
exclusively to protocols in which the number of copies is determined adaptively from previous measurement outcomes.

Optimizing directly over all such strategies is generally intractable. The
instruments at different rounds are coupled through the conditional quantum
state produced by all previous ``continue'' outcomes, so neither the
measurement nor the dimension of the retained system is fixed. Nevertheless,
universal lower bounds on the mean stopping time can be obtained without
solving this optimization explicitly.

Consider binary discrimination between two states $\rho$ and $\sigma$, corresponding to hypothesis $H_0$ and $H_1$, respectively.
Every sequential quantum strategy satisfying an
evidence threshold $A$ obeys, in the large-threshold regime,
\begin{equation}\label{eq:asymptotic_stopping_time_quantum}
		\mathds{E}(N\mid \rho)\ge\frac{\log A}{D(\rho\|\sigma)}+\mathcal{O}(1)\, ,
	\end{equation}
where
\begin{equation}
		D(\rho\|\sigma)=\tr\!\left[\rho(\log\rho-\log\sigma)\right]\, ,
	\label{eq:quantum_KL}
	\end{equation}
is the quantum relative entropy~\cite{vargas2021quantum}.
The bound of~\eqnref{eq:asymptotic_stopping_time_quantum} is directional: when $H_0$ is true, the relevant rate is
$D(\rho\Vert\sigma)$, whereas under $H_1$ it is
$D(\sigma\Vert\rho)$. Importantly, this bound depends only on the states themselves and applies to the fully general class of sequential quantum instruments described above.

A much simpler strategy repeatedly applies the same single-copy positive operator-valued measure (POVM) $\mathcal M=\{E_m\}_m$. This produces i.i.d. classical distributions 
\begin{equation}
    p_{\mathcal M}(m) = \tr(E_m\rho),
    \qquad
    q_{\mathcal M}(m) =   \tr(E_m\sigma),
\end{equation}
that follow the asymptotic behavior of~\eqnref{eq:asymptotic_stopping_time_binary}.
Optimizing the bound individually, conditioned on one hypothesis, gives the \emph{measured relative entropy}
    \begin{equation}\label{eq:Measured_RE}
		D_{\mathrm{meas}}(\rho\|\sigma):=\sup_{\mathcal M}D\!\left(p_{\mathcal M}\middle\|q_{\mathcal M}\right)\, ,
	\end{equation}
which maximizes the classical relative entropy over all single-copy POVMs~\cite{Berta_2017}.
The two measured relative entropies are the optimal directional evidence rates obtainable by repeating a suitably chosen single-copy measurement.

An intermediate and analytically tractable class of quantum strategies is
obtained by measuring independent blocks of copies. Consider a POVM
$\mathcal M^{(K)}$ acting jointly on $K$ copies and repeat this block
measurement independently. If
$p_{\mathcal M^{(K)}}$ and $q_{\mathcal M^{(K)}}$ denote the two
block-outcome distributions, the expected numbers of consumed copies satisfy
\begin{align}
    \mathds{E}_{\mathcal M^{(K)}}(N\mid H_0)
    &\sim
    \frac{
        K\log A
    }{
        D\!\left(
            p_{\mathcal M^{(K)}}
            \Vert
            q_{\mathcal M^{(K)}}
        \right)
    }\, .
\label{eq:block_quantum_measurement_stopping}
\end{align}
This is not the most general sequential quantum strategy: quantum information
is discarded after every block and only the classical block outcomes are
retained. Its advantage is that the measurement rounds are i.i.d., so the
classical sequential theory applies directly.

Optimizing over the $K$-copy block measurement gives the evidence rate per
copy
\begin{equation}
    \frac{1}{K}
    D_{\mathrm{meas}}
    \left(
        \rho^{\otimes K}
        \Vert
        \sigma^{\otimes K}
    \right)\,  .
\end{equation}
Hiai and Petz showed that this rate converges to the quantum relative entropy,
while Hayashi provided an explicit sequence of collective measurements
attaining the limit~\cite{HiaiPetz1991,Hayashi2002},
\begin{equation}
    \lim_{K\rightarrow\infty}
    \frac{1}{K}
    D_{\mathrm{meas}}
    \left(
        \rho^{\otimes K}
        \Vert
        \sigma^{\otimes K}
    \right)
    = D(\rho\Vert\sigma)\, .
\label{eq:regularized_measured_RE}
\end{equation}
Choosing an increasing block size sufficiently slowly ensures both
convergence of the evidence rate and a diverging number of measured blocks,
thereby attaining
\begin{equation}
    \mathds{E}(N\mid H_0)
    = \frac{\log A}{D(\rho\Vert\sigma)} \left[1+o(1)\right].
\end{equation}

A complication remains: a block measurement optimized for
$D(\rho\Vert\sigma)$ need not be optimal for
$D(\sigma\Vert\rho)$. Since the true hypothesis is unknown, repeating one
fixed measurement need not attain both directional bounds simultaneously.
The same issue arises already at the single-copy level, where a POVM attaining
$D_{\mathrm{meas}}(\rho\Vert\sigma)$ need not attain
$D_{\mathrm{meas}}(\sigma\Vert\rho)$.
Li, Tan, and Tomamichel showed that this directional conflict can be resolved
adaptively~\cite{Li_2022}. Their strategy uses the accumulated data to identify
the hypothesis currently favored by the likelihood. When $H_0$ is favored,
the next measurement is chosen to optimize evidence in the direction
$\rho$ against $\sigma$; when $H_1$ is favored, a measurement optimized for
the reverse direction is used. Under the true hypothesis, the fraction of
rounds in which the wrong directional measurement is selected becomes
asymptotically negligible. Consequently, both single-copy directional rates
can be attained within the same adaptive protocol,
\begin{align}
    \mathds{E}(N\mid H_0)
    &\sim
    \frac{
        \log A
    }{
        D_{\mathrm{meas}}(\rho\Vert\sigma)
    }\, .
\label{eq:adaptive_measured_RE_rates}
\end{align}

The construction extends to collective blocks. At each round, the protocol
selects the $K$-copy block measurement optimized for the hypothesis currently
favored by the data. Taking first the large-threshold limit and then allowing
the block size to increase, the directional rates approach the corresponding
quantum relative entropies. Hence,
\begin{align}
    \mathds{E}(N\mid H_0)
    &\sim
    \frac{
        \log A
    }{
        D(\rho\Vert\sigma)
    }\, ,
\label{eq:adaptive_quantum_RE_rates}
\end{align}
under the sampling constraints and regularity assumptions considered in
Ref.~\cite{Li_2022}. The adaptive block strategy therefore attains the
universal quantum limits without requiring a single measurement to be
simultaneously optimal in both directions.

The support conditions require separate attention. The quantum relative
entropy is finite when
$\operatorname{supp}\rho\subseteq\operatorname{supp}\sigma$ and diverges
otherwise. If $D(\rho\Vert\sigma)=\infty$, there exists a measurement outcome
that occurs with nonzero probability under $\rho$ but never under $\sigma$.
Repeatedly testing for such an outcome can identify $\rho$ without error and
with finite mean stopping time when $H_0$ is true. This statement is
directional and does not imply the same conclusion under $H_1$, since
$D(\sigma\Vert\rho)$ may remain finite. For two distinct pure states, both
directional relative entropies diverge, and sequential unambiguous
discrimination can achieve zero error with finite mean sample size.

The results reviewed above concern two fixed simple hypotheses. For a finite
collection of hypotheses, the binary bounds remain valid pairwise: conditioned
on $H_x$, the stopping time is constrained by the competitor that is hardest
to distinguish from $H_x$. Achieving the corresponding multihypothesis limit
requires a policy capable of selecting measurements that accumulate evidence
efficiently against all relevant competitors. The binary adaptive construction
motivates choosing the measurement according to the hypothesis currently
favored by the data, but it does not by itself establish a general
attainability theorem for multiple or continuously many hypotheses. In the
parameter-testing setting developed below, we therefore use the binary results
as pairwise quantum benchmarks while treating the optimization over the
continuum of competing parameter values explicitly.

Finally, we note that sequential quantum hypothesis testing has recently been extended beyond simple hypotheses to 
settings involving fixed composite 
hypotheses~\cite{zecchin2025quantumsequentialuniversalhypothesis,simpson2026samplecomplexitycompositequantum,simpson2026optimalerrorexponentscomposite}. 
These developments are closely connected to generalized versions of the quantum Stein 
lemma~\cite{hayashi2025generalizedquantumsteinslemma,Lami2025GQSL,lami2025generalisedquantumsanovtheorem,lami2025doublycompositechernoffsteinlemma}.


\section{Sequential parameter testing}\label{sec:cont-hyp-testing}
Binary sequential hypothesis testing discriminates between two simple hypotheses, corresponding to two values of an unknown parameter. In many 
estimation scenarios, however, there is no natural reason to restrict attention to a finite set of parameter values. The question arises how a 
sequential test can be designed to single out the parameter value most compatible with the observed measurement data, and when this data is conclusive 
enough to reach a decision. A natural sequential approach is to employ Bayesian updating and stop once the posterior variance falls below a 
prescribed threshold. Since the posterior variance generally depends on the measurement outcomes, the stopping time is random, allowing the desired 
estimation precision to be certified on a run-by-run basis. 

A stopping rule based on the posterior variance reflects the objectives of parameter estimation. Here, however, we are interested in a different objective, 
more closely aligned with hypothesis testing. In conventional parameter estimation, FOMs such as the MSE penalize estimation errors according to their 
distance from the true parameter value. Instead, we consider settings in which the operational distinction is between acceptable and unacceptable 
parameter values: all values within a prescribed tolerance region are regarded as equivalent for the task, whereas values outside this region are treated as 
competing alternatives, irrespective of their precise distance from it, very much in line with quantum certification procedures~\cite{Pallister2018}. Motivated by this threshold-based objective, we introduce a framework for 
\emph{parameter testing} as a continuous analogue of hypothesis testing.

Within this framework, we develop a sequential testing procedure, the \emph{twin-peaks test}, which is described in the next section. When passing from 
a finite set of hypotheses to a continuum, several subtleties arise. To motivate the \textit{twin-peaks test}, we first examine a number of alternative approaches 
and discuss their limitations.

Let $\vartheta\in\Theta$ be an unknown parameter in a convex and compact parameter space distributed according to the prior probability density function 
$\pi(\vartheta)$. To estimate the parameter, one has access to a series of measurement outcomes $\Vec{m}$ sampled from the set of all possible measurement 
outcomes $\{\Vec{m}\}$. Each outcome occurs with conditional probability distribution $P(\Vec{m}\mid \vartheta)$, and one can calculate the posterior probability 
density function $p(\vartheta\mid \Vec{m})$ by Bayes' rule.

A naive approach might discretize the parameter space into a finite number of compact and convex regions, $\Theta_i$, and associate with each region 
the hypothesis that the parameter lies in it, $H_i:\, \vartheta\in\Theta_i$. If the regions form a partition, that is, every two regions are disjoint, and the regions 
cover the entire parameter space, the set of hypotheses is mutually exclusive and exhaustive. The probability of a hypothesis after 
observing outcomes $\Vec{m}$ is given by the integral of the posterior over the corresponding region $\int_{\Theta_i} p(\beta\mid \vec{m})\dd\beta$.
This approach has two related drawbacks arising from the predetermined partition. First, the reported estimate is restricted to a fixed collection of cells, so the 
location and resolution of the accepted region depend on an arbitrary discretization chosen before observing the data. Second, two nearby parameter values lying on 
opposite sides of a cell boundary are treated as different hypotheses, whereas more distant values within the same cell are treated as equivalent. As a result, if the 
true parameter lies close to a boundary, posterior evidence may remain divided between neighboring cells, substantially delaying or even preventing the test from 
reaching its stopping threshold.

To circumvent these partition artifacts, the methods introduced below allow the center of the tolerance region to vary continuously with the observed data. The 
operational resolution remains fixed by the prescribed tolerance $\delta$, but the accepted region is centered on the parameter value favored by the data and, under 
standard consistency conditions, converges towards the true parameter.

The first possibility is to extend the strong-error condition of \eqnref{eq:condition_sum} to a continuous parameter space. 
For every $\vartheta\in\Theta$, we define the composite neighbourhood hypothesis
	\begin{equation}
    		H_\vartheta:\varphi\in B_\delta(\vartheta)\, ,
	\end{equation}
where $\varphi$ denotes the unknown parameter and $\vartheta$ is the center of the tolerance region. Since the hypotheses $H_\vartheta$ overlap, 
they are not mutually exclusive and cannot be compared through the sum of the probabilities of all remaining hypotheses, as in \eqnref{eq:condition_sum}.
We therefore define the set of $\delta$-separated alternatives as
	\begin{equation}
    		B^{c}_\delta(\vartheta):=\Theta\setminus B_\delta(\vartheta)\, .
	\label{eq:barBall}
	\end{equation}
Each $H_\vartheta$ is then compared with its complementary hypothesis
	\begin{equation}
    		H_{\bar\vartheta}:\varphi\in B^{c}_\delta(\vartheta)\, .
	\end{equation}
This associates a Bayesian strong-error condition with every hypothesis $H_\vartheta$: whenever $H_\vartheta$ is accepted, its posterior probability
must satisfy
\begin{equation}
    P(H_\vartheta\mid\Vec{m})\geq 1-\epsilon\, .
\end{equation}
All these conditions are evaluated in parallel, and the global procedure stops as soon as the condition is satisfied for at least one $\vartheta\in\Theta$. 
We refer to this procedure as the \textit{complement test} (see App.~\ref{app:complement_test} for a detailed construction).

The complement test avoids the artifacts introduced by a fixed partition since the tolerance region is centered according to the observed data rather than on a 
predetermined grid. Its main practical limitation is that $H_\vartheta$ and $H_{\bar\vartheta}$ are composite hypotheses. Their evidence cannot be obtained by 
evaluating the likelihood at a single parameter value. Instead, one must specify a prior or reference measure and integrate the pointwise likelihood over the 
corresponding parameter regions. These integrals must be evaluated for every candidate center and recalculated after each measurement round, making 
the complement test explicitly prior-dependent and computationally demanding. 

To avoid these integrations, one may instead formulate the test directly in terms of pointwise likelihood or posterior-density values. In a continuous
parameter space, however, a point hypothesis $H_\vartheta:\varphi=\vartheta$ corresponds to a set of measure zero. It is therefore not meaningful 
to compare its posterior probability with the total posterior probability of the remaining parameter values, as in the discrete multihypothesis condition of 
\eqnref{eq:condition_sum}~\cite{baum1994sequential,VeeravalliBaum1995}. A possible alternative is to compare the posterior density associated with 
$H_\vartheta$ with the average posterior density over the set of $\delta$-separated alternatives $B^c_\delta(\vartheta)$, defined in ~\eqnref{eq:barBall}.
The resulting ratio compares $p(\vartheta\mid\Vec{m})$ with the average posterior density over parameter values separated from $\vartheta$ by at least 
$\delta$. The test stops once this ratio exceeds a prescribed threshold for some $\vartheta$, and the maximizing value may be reported as the estimate. 
We refer to this procedure as the \textit{concentration test}; it is illustrated in orange in Fig.~\ref{fig:diff-tests} and presented in detail in 
App.~\ref{app:concentration_test}.

To further simplify the test, rather than averaging over all distinguishable alternatives, one may compare the best-fitting parameter value only with its strongest 
competitor. This principle underlies several multihypothesis testing procedures, in which the probability of one hypothesis is compared with the maximum probability 
among the remaining hypotheses~\cite{Armitage50,Lorden77,Draglia_1999}. Building on this idea, we introduce a test that we coin twin-peaks test---shown in blue in 
Fig.~\ref{fig:diff-tests}. We next describe the test in detail, as it underpins the main results of this work.


\section{The twin-peaks test}\label{sec:twin-peaks_test}
We now formulate the continuous analogue of the maximum-competitor condition in \eqnref{eq:condition_max}. In the discrete setting, the posterior
probability of a candidate hypothesis is compared with that of its strongest competitor. In a continuous parameter space, point hypotheses have zero
posterior probability, but their posterior densities can still be compared once a reference measure has been fixed.

To motivate this comparison, consider two sufficiently small regions of equal volume centered at $\vartheta$ and $\vartheta_0$. To leading order, their
posterior probabilities are proportional to $p(\vartheta\mid\Vec{m})$ and $p(\vartheta_0\mid\Vec{m})$, respectively. As the regions shrink, the ratio of
their posterior probabilities converges to the ratio of the corresponding posterior densities. Therefore, it is meaningful to compare point hypotheses through
posterior-density ratios rather than through their vanishing individual probabilities. Such density ratios are understood relative to the reference
measure used to define the posterior density.

By Bayes' rule, the ratio of the posterior densities at two parameter values $\vartheta$ and $\vartheta_0$ is
	\begin{equation}
    		\frac{p(\vartheta\mid\Vec{m}) }{p(\vartheta_0\mid\Vec{m})}=\frac{ p(\Vec{m}\mid\vartheta)\, \pi(\vartheta) }{ p(\Vec{m}\mid\vartheta_0)\,
		\pi(\vartheta_0)}\, .
	\label{eq:probability_ratio}
	\end{equation}
We denote by
	\begin{equation}
    		\hat{\vartheta}(\Vec{m})\in\operatorname*{arg\,max}_{\vartheta\in\Theta}p(\vartheta\mid\Vec{m})
	\label{eq:MAP_estimate}
	\end{equation}
a current maximum-a-posteriori (MAP) estimate. A direct comparison of $\hat{\vartheta}$ with all remaining parameter values would be trivial, since values of $\vartheta_0$ can approach $\hat{\vartheta}$ arbitrarily closely and by continuity the same holds for their likelihood. It is hence natural
to regard all parameter values within $B_\delta(\hat{\vartheta})$ as operationally equivalent and restrict the
comparison to the set of $\delta$-separated alternatives $B^c_\delta(\hat{\vartheta})$, defined in \eqnref{eq:barBall}.

The stopping condition of the \textit{twin-peaks test} is
	\begin{equation}
    		\frac{p(\hat{\vartheta}\mid\Vec{m})}{\displaystyle\max_{\vartheta_0\in B^c_\delta(\hat{\vartheta})}p(\vartheta_0\mid\Vec{m})}
    		=\frac{p(\Vec{m}\mid\hat{\vartheta})\,\pi(\hat{\vartheta})}{\displaystyle\max_{\vartheta_0\in B^c_\delta(\hat{\vartheta})}
        	p(\Vec{m}\mid\vartheta_0)\,\pi(\vartheta_0)}\geq A_\epsilon\, .
	\label{eq:stopping_condition_twin_peaks}
	\end{equation}
The denominator is the largest posterior density among the $\delta$-separated alternatives and therefore
identifies the strongest competitor to the current MAP estimate.

The twin-peaks condition has a direct worst-competitor interpretation. Once \eqnref{eq:stopping_condition_twin_peaks} is satisfied, every parameter
value outside $B_\delta(\hat{\vartheta})$ is suppressed relative to the MAP estimate by at least the factor $A_\epsilon$. Equivalently, all parameter 
values that remain competitive with the MAP estimate at evidence level $A_\epsilon$ lie within the prescribed tolerance region. The test thus provides 
a posterior-density-ratio certification of the desired resolution. The accepted test identifies the region $B_\delta(\hat{\vartheta})$, whose center 
$\hat{\vartheta}$ may additionally be reported as a point estimate.

For a uniform prior, the MAP estimate in \eqnref{eq:MAP_estimate} coincides with the maximum-likelihood estimate (MLE), and the prior factors in
\eqnref{eq:stopping_condition_twin_peaks} cancel. In this case, the twin-peaks test compares the maximum likelihood directly with the largest
likelihood attained over $B^c_\delta(\hat{\vartheta})$. It therefore requires only pointwise likelihood evaluations and optimizations, rather than
the prior-weighted integrations over parameter regions required by the complement and concentration tests. Moreover, by focusing on the strongest
competitor, the test avoids averaging together alternatives that may be supported very differently by the observed data.

For independent samples, the twin-peaks test may be viewed as a continuous analogue of parallel pairwise SPRTs. The current MAP estimate
$\hat{\vartheta}$ is compared with every $\delta$-separated alternative, and the test stops once all the corresponding posterior-density 
ratios exceed the threshold $A_\epsilon$. Since the strongest competitor produces the smallest of these ratios, it is sufficient to evaluate 
the single maximization appearing in \eqnref{eq:stopping_condition_twin_peaks}. This interpretation is only an analogy with ordinary SPRTs: 
both $\hat{\vartheta}$ and its strongest competitor are selected from the accumulated data and may change between measurement rounds. 
Moreover, all pairwise comparisons are correlated because they are constructed from the same measurement record.

We note  that unlike the complement test, the twin-peaks condition does not directly bound the total posterior probability outside
$B_\delta(\hat{\vartheta})$. The threshold $A_\epsilon$ quantifies the evidence against the strongest individual competitor rather than an exact 
finite-sample posterior error probability. For the twin-peaks test, however, $\epsilon$ should be understood as an
asymptotic calibration parameter rather than as an exact finite-sample probability of error (see \appref{app:Asymptotic_equivalence})


\section{Stopping times and asymptotic behavior}\label{sec:stopping times}

The asymptotic stopping-time bound for binary hypothesis testing given in \eqnref{eq:asymptotic_stopping_time_binary} extends naturally to multiple 
competing hypotheses, in which case the stopping time is governed by the least distinguishable alternative~\cite{baum1994sequential,Draglia_1999}. 
At each stage, the \textit{twin-peaks test} compares the hypothesis corresponding to the parameter $\vartheta$ against all hypotheses $\vartheta_0\notin B_\delta(\vartheta)$, 
and we obtain the expected asymptotic stopping time [cf.~Eq.~\eqref{eq:asymptotic_stopping_time_binary}]
    \begin{equation}\label{eq:asymptotic_stopping_time_classic}
        \mathds{E}(N\mid \vartheta) \sim \max\limits_{\vartheta_0: \|\,\vartheta-\vartheta_0\|\,\ge 
        \delta}\frac{\log{A_\epsilon}+\log\frac{\pi(\vartheta_0)}{\pi(\vartheta)}}{D(\vartheta\|\,\vartheta_0)}\, ,
    \end{equation}
where 
	\begin{equation}
		D(\vartheta\|\,\vartheta_0):=\lim_{n\to\infty}\frac{D(P(\Vec{m}_n\mid\vartheta)\|\, P(\Vec{m}_n\mid\vartheta_0))}{n}\,  
	\label{eq:KL_rate}
	\end{equation}		
denotes the asymptotic Kullback--Leibler rate whenever the limit exists. Furthermore, \eqnref{eq:asymptotic_stopping_time_classic} remains valid whenever, 
under the true parameter $\vartheta$, the normalized log-likelihood ratio satisfies
 	\begin{equation}
		\frac{1}{n} \log\frac{P(\Vec{m}_n\mid\vartheta)}{P(\Vec{m}_n\mid\vartheta_0)}\xrightarrow[P_\vartheta]{\mathrm{a.s.}}
		D(\vartheta\|\,\vartheta_0)
	\label{eq:almost_surely}
	\end{equation}
For i.i.d.\ observations, \eqnref{eq:KL_rate} reduces to the ordinary Kullback--Leibler divergence by the strong law of large numbers.
More generally, \eqnref{eq:asymptotic_stopping_time_classic} remains
valid for arbitrary measurement strategies whenever the normalized log-likelihood ratio converges almost surely according to
\eqnref{eq:almost_surely}. In this case, $D(\vartheta\|\,\vartheta_0)$ is interpreted as the corresponding
asymptotic information rate.  The average expected asymptotic stopping time is then 
    \begin{equation}\label{eq:average_asymptotic_stopping_time}
        \mathds{E}(N)=\int_\Theta \pi(\vartheta)\, \mathds{E}(N\mid \vartheta)\, \dd \vartheta\, .
    \end{equation}
    
For a uniform prior, the optimization depends only on the KL divergence. If the KL divergence increases monotonically with parameter-space distance, the least 
distinguishable alternatives lie on the boundary of the tolerance region, i.e., $|\hat{\vartheta}-\vartheta_0|=\delta$. This is locally true for regular 
identifiable models and sufficiently small $\delta$ by the quadratic expansion of the KL divergence---see \eqref{eq:kl_fisher_local} below. Without global 
monotonicity, however, the optimization must be carried out over all $\delta$-separated alternatives $B_\delta^c(\hat{\vartheta})$.

As discussed in \secref{sec:quantum}, quantum experiments may employ adaptive measurement strategies, so that the sequence of observed outcomes is generally 
not i.i.d. The asymptotic stopping-time bound for local binary hypothesis testing given in
\eqnref{eq:adaptive_measured_RE_rates}
extends directly to the 
\textit{twin-peaks test}, provided the induced log-likelihood ratio satisfies the almost-sure convergence condition \eqnref{eq:almost_surely}. In this case, 
the stopping time is governed by the corresponding asymptotic information rate. 

Whether this bound can be saturated depends on the measurement strategy. If the locally optimal measurement depends on the unknown parameter, then the optimal 
information rate is generally not achievable. Conversely, if there exists a parameter-independent measurement attaining the optimal information rate, then an i.i.d.\ 
measurement strategy saturates the local bound of 

\eqnref{eq:adaptive_measured_RE_rates},
and no adaptive local strategy can asymptotically improve the stopping 
time. If, in addition, the measured relative entropy coincides with the quantum relative entropy, then no collective measurement strategy can perform asymptotically 
better. We will encounter precisely such a scenario in \secref{sec:purity}, where we estimate the purity of a qubit.

The local behavior of the stopping time is determined by the curvature of the information rate. In the i.i.d.\ setting, the local curvature of the Kullback--Leibler 
divergence is given by the Fisher information~\cite{AmariShun-ichi2000Moig} (see also \appref{app:Fisher_bounds} for a derivation),
    \begin{equation}\label{eq:kl_fisher_local}
        D(p(\vartheta)\|\,p(\vartheta+\delta))=\frac{1}{2}I(\vartheta)\delta^2+o(\delta^2)\, .
    \end{equation}

When the likelihood has concentrated around the true value, the closest competing 
parameter outside $B_\delta(\vartheta_n)$ lies at a distance $\delta$.  For small $\delta$
    \begin{equation}\label{eq:asymptotic_stopping_time_fisher}
        \mathbb E(N_{\rm TP}\mid \vartheta)\sim\frac{-2\log\epsilon}{I(\vartheta)\delta^2}\ \,
    \end{equation}
where TP stands for twin-peaks. In \appref{app:complement_test}, we derive the same leading small-$\delta$ and small-$\epsilon$ scaling as the strong Bayesian 
stopping rule of the \textit{complement test}.


\section{Comparison to deterministic strategies}\label{sec:comparison_fixed_size}

We now compare the performance of sequential strategies with tests using a deterministic number of samples.
To compare the two approaches, we require both strategies to satisfy the stopping condition of \eqnref{eq:stopping_condition_twin_peaks}.
For sequential strategies the twin-peaks stopping condition is imposed on every trajectory: for every measurement record on which the sequential test stops, the evidence ratio exceeds the prescribed threshold. We shall refer to it as the  \textit{strong stopping condition}.
To benchmark the sequential protocol against a strategy using a predetermined number $N$ of samples, one could require the same condition to hold for every possible measurement record after $N$ rounds. Such a uniform requirement may be excessively stringent and may not be satisfied by any finite sample size. We therefore adopt a weaker fixed-sample benchmark in which a bounded function of the twin-peaks evidence ratio is controlled on average over the possible measurement outcomes. This choice requires some care: although any monotone transformation of the evidence ratio defines the same trajectory-wise stopping rule, different transformations generally lead to different averages.

We choose the transformation that, in the binary
case, coincides with the posterior probability of error and gives the usual
calibration $A_\epsilon=(1-\epsilon)/\epsilon$.
Specifically, we define
    \begin{align}
        	\eta = \left(\frac{p(\hat{\vartheta}\mid \Vec{m}_k)}{\max\limits_{\|\,\hat{\vartheta}-\vartheta_0\|\,\ge \delta} 
        	p(\vartheta_0\mid\Vec{m}_k)}+1\right)^{-1}\, .
    \label{eq:eta}
    \end{align}
For two fixed and exhaustive hypotheses, $\eta$ is exactly the posterior
probability of error associated with accepting the more probable hypothesis.
For the twin-peaks test, it is instead a bounded measure of the residual
ambiguity between the current estimate and its strongest
$\delta$-separated competitor. It does not include the combined posterior
weight of the remaining alternatives and is therefore not the total posterior
probability of an incorrect estimate.

With the calibration $ A_\epsilon = \frac{1-\epsilon}{\epsilon}$
the twin-peaks condition of Eq.~\eqref{eq:stopping_condition_twin_peaks} is exactly equivalent to
$\eta(\Vec{m}_N)\leq\epsilon$.
This provides an  interpretation of the calibration parameter:
$\epsilon$ bounds, on every stopped trajectory, the residual ambiguity in the
strongest pairwise comparison. It should not be interpreted as an exact
finite-sample probability that the true parameter lies outside
$B_\delta(\hat{\vartheta})$. In~\appref{app:Asymptotic_equivalence} we show that this interpretation holds in the asymptotic limit of large sample sizes.

For a fixed-sample strategy, we instead control the same quantity on average.
We define $N_{\mathrm{fss}}^1(\epsilon)$ as the smallest sample size satisfying
\begin{equation}
    \mathbb E(\eta)
    =
    \sum_{\Vec{m}_N}
        p(\Vec{m}_N)\,
        \eta(\Vec{m}_N)
    \leq
    \epsilon\, ,
\label{eq:smallest_size}
\end{equation}
where the average over the unknown parameter is included in the outcome
probability
\begin{equation}
    p(\Vec{m}_N)
    =
    \int_\Theta
        p(\Vec{m}_N\mid\vartheta)\,
        \pi(\vartheta)\,
        \dd\vartheta\, .
\end{equation}
Thus, $N_{\mathrm{fss}}^1(\epsilon)$ controls the residual ambiguity
on average over both the measurement records and the prior distribution of the
parameter. Unlike the sequential condition, it does not require
$\eta(\Vec{m}_N)\leq\epsilon$ for every possible record.

The prior-averaged criterion may conceal parameter values for which the
fixed-sample strategy performs poorly, particularly when those values have
small prior weight. We therefore also consider the stronger,
parameter-uniform criterion
\begin{equation}
    \max_{\vartheta\in\Theta}
    \mathbb E(\eta\mid\vartheta)
    =
    \max_{\vartheta\in\Theta}
    \sum_{\Vec{m}_N}
        p(\Vec{m}_N\mid\vartheta)\,
        \eta(\Vec{m}_N)
    \leq
    \epsilon\, ,
\label{eq:eta_strong}
\end{equation}
and denote by $N_{\mathrm{fss}}^\infty(\epsilon)$ the smallest deterministic
sample size satisfying it. This condition is independent of the prior and
guarantees the prescribed expected residual ambiguity separately for every
true parameter value. It nevertheless remains an outcome-averaged condition
and is therefore weaker than imposing the trajectory-wise twin-peaks condition
for every possible measurement record.

We thus distinguish three levels of control: the sequential test bounds
$\eta$ trajectory by trajectory,
$N_{\mathrm{fss}}^1(\epsilon)$ bounds its prior-averaged value, and
$N_{\mathrm{fss}}^\infty(\epsilon)$ bounds its outcome-averaged value over the full parameter space. Since the latter two sample sizes are
deterministic, they can be compared directly with the mean sequential stopping
time $\mathbb E \left(N_{\mathrm{TP}}(\epsilon)\right)$.

To obtain an analytical estimate of these fixed-sample benchmarks, consider
an i.i.d.\ measurement strategy and define
\begin{equation}
    D_\delta(\vartheta)
    :=
    \inf_{\vartheta_0\in
    B_\delta^c(\vartheta)}
    D\!\left(
        p_\vartheta
        \big\|
        p_{\vartheta_0}
    \right)\, .
\label{eq:delta_distinguishability_rate}
\end{equation}
For i.i.d. samples, 
\eqnref{eq:asymptotic_stopping_time_classic} or, under the regularity assumptions, \eqnref{eq:asymptotic_stopping_time_fisher},
 the twin-peaks log-evidence grows typically as
$N D_\delta(\vartheta)$. If, in addition, the log-evidence is sufficiently
sharply concentrated around this value that atypical records make a negligible
contribution to the expectation of $\eta$, then
\begin{equation}
    \mathbb E(\eta\mid\vartheta)
    \approx
    \left[
        \exp\!\left(
            N D_\delta(\vartheta)
        \right)
        +1
    \right]^{-1}\, .
\label{eq:eta_concentration_approx}
\end{equation}
Consequently,
\begin{align}
    \mathbb E(\eta)
    &\approx
    \int_\Theta
        \pi(\vartheta)
        \left[
            \exp\!\left(
                N D_\delta(\vartheta)
            \right)
            +1
        \right]^{-1}
        \dd\vartheta\, ,
\label{eq:expected_error}
\\
    \max_{\vartheta\in\Theta}
    \mathbb E(\eta\mid\vartheta)
    &\approx
    \max_{\vartheta\in\Theta}
    \left[
        \exp\!\left(
            N D_\delta(\vartheta)
        \right)
        +1
    \right]^{-1}\, .
\label{eq:wc_error}
\end{align}

For regular identifiable models and sufficiently small $\delta$,
one can take the approximation of~\eqnref{eq:kl_fisher_local} 
to obtain
\begin{align}
    \mathbb E(\eta)
    &\approx
    \int_\Theta
        \pi(\vartheta)
        \left[
            \exp\!\left(
                \frac{N}{2}
                I(\vartheta)\delta^2
            \right)
            +1
        \right]^{-1}
        \dd\vartheta\, ,
\label{eq:expected_error_local}
\\
    \max_{\vartheta\in\Theta}
    \mathbb E(\eta\mid\vartheta)
    &\approx
    \left[
        \exp\!\left(
            \frac{N}{2}
            \delta^2
            \min_{\vartheta\in\Theta}I(\vartheta)
        \right)
        +1
    \right]^{-1}\, .
\label{eq:wc_error_local}
\end{align}
Inverting Eq.~\eqref{eq:wc_error_local} gives
\begin{equation}
    N_{\mathrm{fss}}^\infty(\epsilon)
    \approx
    \frac{
        2\log A_\epsilon
    }{
        \delta^2
        \displaystyle
        \min_{\vartheta\in\Theta}I(\vartheta)
    }\, .
\label{eq:N_fss_asymptotic}
\end{equation}

Equation~\eqref{eq:N_fss_asymptotic} clarifies the potential advantage of
sequential stopping. For a fixed parameter value, sequential and fixed-sample
strategies have the same leading dependence on $\delta$ and
$A_\epsilon$. A deterministic strategy satisfying the parameter-uniform
criterion must, however, be designed for the least informative parameter
value. The sequential test instead adapts its stopping time to the evidence
accumulated in each measurement record and may therefore reduce the
prior-averaged sample cost when the information rate $I(\vartheta)$ depends on the unknown parameter. Maximizing the Fisher information over all possible measurements results in the quantum Fisher information (QFI) $\mathcal{I}(\vartheta)$, which governs the asymptotic behavior of the optimal measurement strategy.


\section{Phase estimation}\label{sec:phase}
We now apply the \textit{twin-peaks test} to the estimation of an unknown phase of a single-qubit system. We are given a probe system in the state 
$\ket{\theta}=(\ket{0}+\mathrm{e}^{i\theta}\ket{1})/\sqrt{2}$ and want to determine its phase $\theta$, where all values in the interval 
$[0,2\pi)$ are equally likely. We treat this problem as a state discrimination task over a continuous set of states. This differs from the standard 
scenario in quantum parameter estimation in that we have no control over the probe state itself, but only over the performed measurement.

We begin by considering single-copy measurements. Since the states $\ket{\theta}$ lie in the equatorial plane of the Bloch sphere, any component of a POVM effect 
orthogonal to this plane does not contribute to the outcome probabilities. Such components can be removed while preserving positivity and completeness, and we may 
therefore restrict without loss of generality to POVMs with effects in the equatorial plane. The phase covariance of the state family and the uniform prior further 
motivate focusing on symmetric equatorial POVMs:
	\begin{equation} 
		E_m = \frac{2}{n} \ket{\phi_m}\!\bra{\phi_m}, \qquad m=0,\ldots,n-1\, ,
	\end{equation} 
where 
	\begin{equation}
		\ket{\phi_m} = \frac{1}{\sqrt{2}} (\ket{0} + \mathrm{e}^{\mathrm{i}\phi_m}\ket{1}), \qquad \phi_m = \frac{2\pi m}{n}\, .
	\end{equation}
The probability of observing outcome $m$ for a probe with phase $\theta$ is 
	\begin{equation}
        	p(m\mid \theta)=\frac{1}{n}\left(1+\cos(\theta-\tfrac{2\pi m}{n})\right)\, .
    	\label{eq:prob_outcome}
	\end{equation}
For even $n$, the POVM can be implemented by uniformly selecting one of $n/2$ equatorial projective measurements. In the continuous limit, this yields the covariant~\cite{HOLEVO1973337,holevo2011probabilistic} 
phase POVM 
	\begin{equation}
		\left\{ E_\phi = \frac{\dd\phi}{\pi} \ket{\phi}\!\bra{\phi} \right\}_{\phi\in[0,2\pi)}, \qquad \int_0^{2\pi}\frac{\dd\phi}{\pi} \ket{\phi}\!\bra{\phi}\, = \one\, ,
	\label{eq:covariant_phase_measurement}
	\end{equation}
with likelihood 
	\begin{equation} 
		p(\phi\mid\theta) = \operatorname{tr} \left[ E_\phi\ket{\theta}\!\bra{\theta} \right] = \frac{1}{2\pi} \left[ 1+\cos(\phi-\theta) \right]\, .	
	\label{eq:covariant_phase_likelihood}
	\end{equation}
Operationally, this corresponds to choosing an equatorial measurement axis uniformly at random and labelling its two outcomes by the corresponding antipodal directions.

\begin{figure}[t]
    \centering
    \hspace{-20pt}
    \includegraphics[width=1.05\columnwidth, trim={10pt 0 50pt 50pt}, clip]{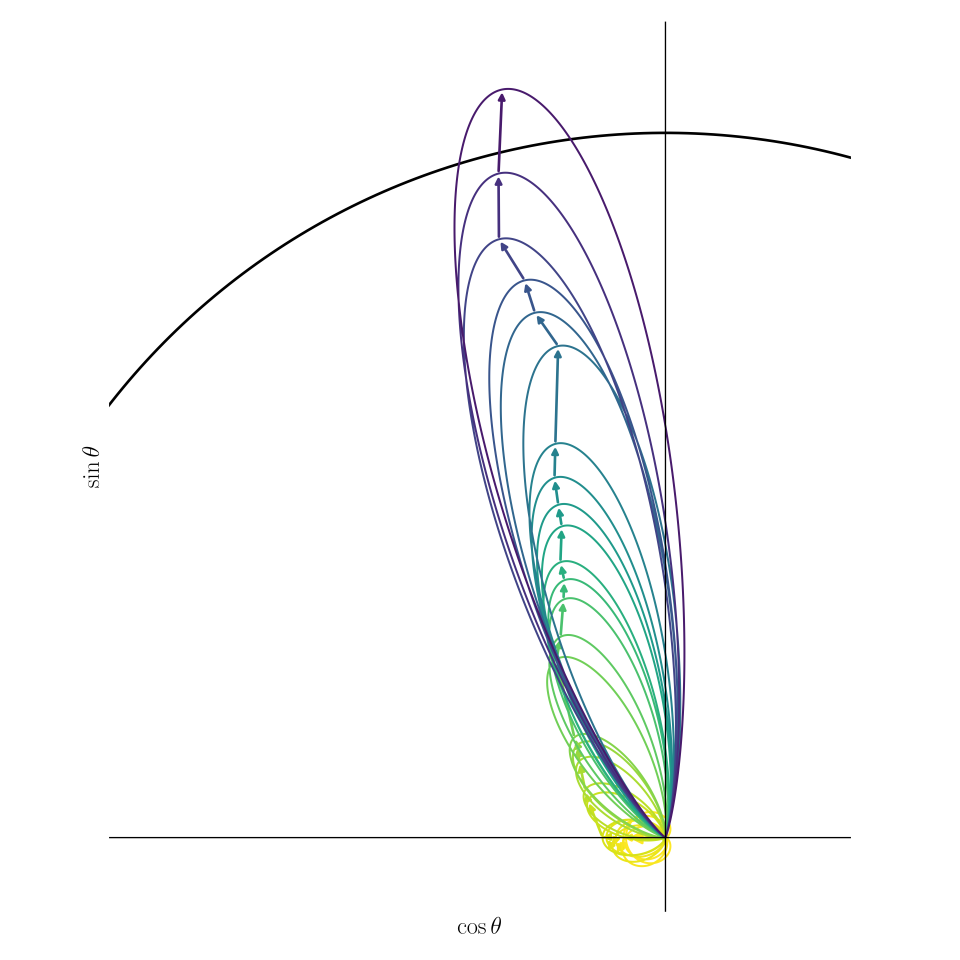}
    \caption{One run of the equatorial phase testing protocol according to the \textit{twin-peaks test}, and choosing a random measurement strategy. 
    The plot is made in phase space, where the horizontal (vertical) axis corresponds to $\cos\theta$ ($\sin\theta$). Each ellipse corresponds to one 
    round, with a color gradient from yellow for early rounds to blue for later ones. In each ellipse, the arrow points towards the most likely phase, 
    $\hat\theta$, of that round. The distance between that point and the origin is the likelihood ratio 
    $p(\vec{m}_k\mid \hat\theta)/p(\vec{m}_k\mid \theta_0)$, where $\theta_0$ is selected as the alternative hypothesis by the \textit{twin-peaks test} given 
    the interval of size $2\delta=1$ around the ML estimate $\hat\theta$. For the rest of the points forming the ellipse, which correspond to all other 
    phases $\theta$, their distance to the origin represents the ratio $p(\vec{m}_k\mid \theta)/p(\vec{m}_k\mid \theta_0)$. Note that here $\theta_0$ is fixed 
    to the same value as before, that is, the most likely phase within the alternative hypothesis set given the interval around $\hat\theta$. The 
    radius of the black arc represents the saturation of the stopping condition Eq.~\eqref{eq:stopping_condition_twin_peaks} for an evidence calibration parameter
    $\epsilon=0.05$. The data comes from the same estimation round shown in Fig.~\ref{fig:diff-tests}.}
    \label{fig:phase_point_test}
\end{figure}
\begin{figure}[t]
    \centering
    \hspace{-15pt}
    \includegraphics[width=1.05\columnwidth]{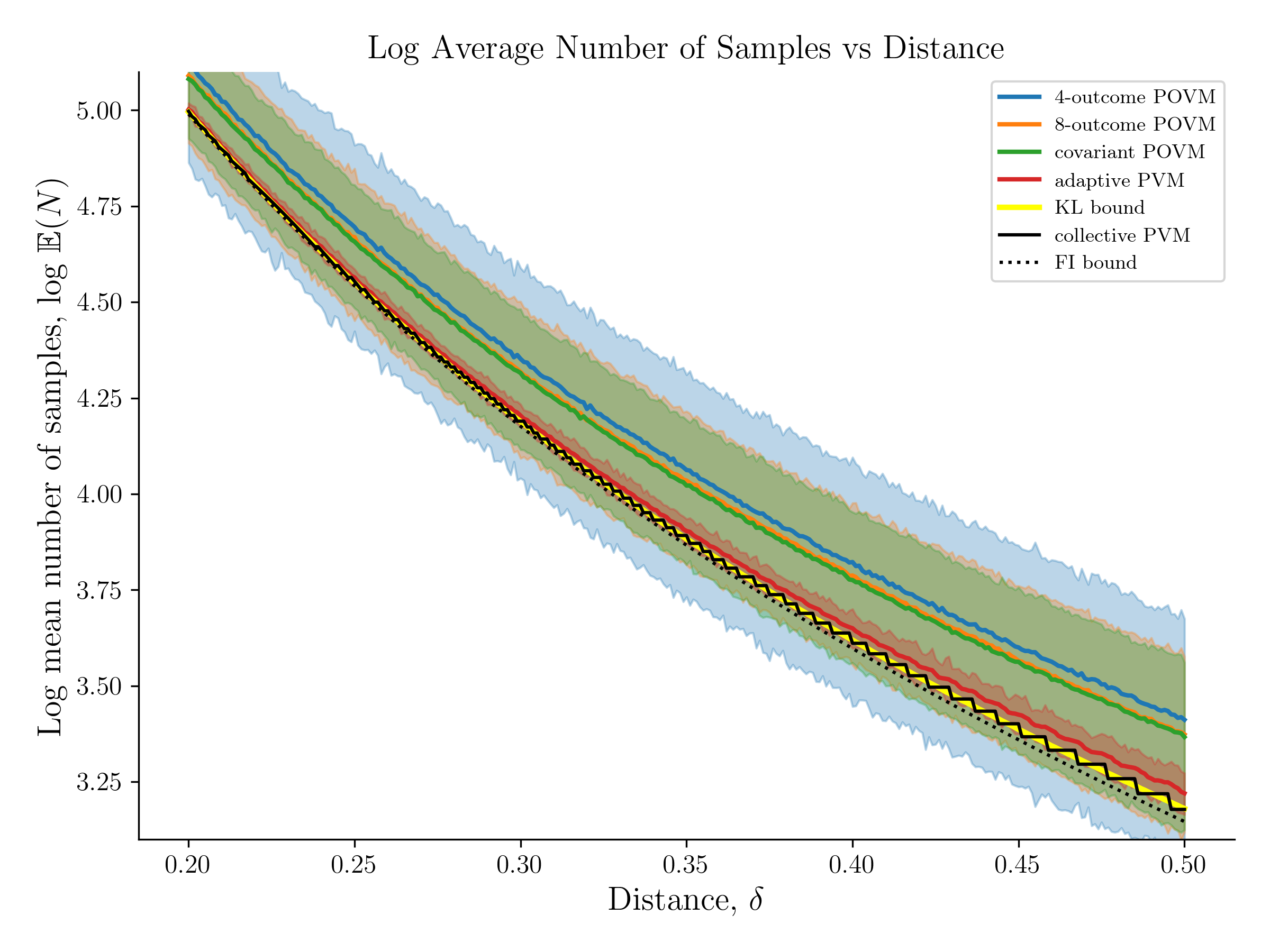}
    \caption{Log-plot of the mean number of rounds until the maximum sequential test stops using the evidence threshold $A_\epsilon=19$, corresponding to the calibration $\epsilon=0.05$. Different colors represent the different strategies explained in the main text; the lines show the mean number of rounds over $10^4$ runs for the different strategies, with their corresponding standard deviation shown as shadow areas. The 4- and 8-outcome POVMs are shown in blue and 
    orange, respectively, the random projective measurement in green, and the adaptive measurement in red. The black line represents the collective 
    strategy with a fixed number of rounds, and the yellow line shows the asymptotic stopping time from the KL divergence for the phase-covariant POVM 
    in Eq.~\eqref{eq:stopping_bound_phase}. The black dotted line shows the approximation for small $\delta$ with the Fisher information of  Eq.~\eqref{eq:asymptotic_stopping_time_fisher}.}
   \label{fig:phase_point_avg_samples}
\end{figure}

A measurement strategy specifies how the measurement performed at each round is selected. The simplest strategies repeatedly apply one of the finite-outcome 
POVMs above, or independently sample from the covariant phase POVM. Their outcomes are i.i.d., which allows their asymptotic stopping times to be related to the KL 
divergence between the corresponding outcome distributions. 

Already in binary sequential quantum hypothesis testing, adaptive measurements are important for simultaneously attaining the optimal directional distinguishability 
rates under both hypotheses~\cite{Li_2022}. The measurements optimized for accumulating evidence in the two directions need not coincide, and the adaptive strategy 
switches between them according to the hypothesis currently favored by the measurement record.

Inspired by this result, we consider a simple greedy adaptive phase-testing strategy. After the first $k$ measurement outcomes, we compute the current maximum-
likelihood estimate $\hat{\theta}(\Vec{m}_k)$. A projective measurement perpendicular to this estimate has maximal local sensitivity to small phase variations, as 
quantified by the Fisher information. We therefore choose the measurement axis in round $k+1$ as 
	\begin{equation}
		\phi_{k+1} = \hat{\theta}(\Vec{m}_k) + \frac{\pi}{2}\, ,
	\label{eq:adaptive_phase_measurement} 
	\end{equation}
and update the estimate after observing the next outcome. At $\theta=\hat{\theta}(\Vec{m}_k)$, the classical Fisher information of this measurement equals the 
single-copy quantum Fisher information. The opposite perpendicular direction defines the same projective measurement with its outcome labels exchanged. After every 
round, we evaluate the twin-peaks condition in \eqnref{eq:stopping_condition_twin_peaks} and stop once it is satisfied. Figure~\ref{fig:phase_point_test} shows one 
realization of this procedure, including the evolution of the likelihood profile and the maximum-likelihood estimate.

For an i.i.d. measurement strategy, the asymptotic stopping time is governed by the least distinguishable $\delta$-separated alternative, as described by 
\eqnref{eq:asymptotic_stopping_time_classic}. For the covariant phase POVM, the KL divergence depends only on the phase difference and can be evaluated exactly 
(see \appref{app:KL_for_phase}), 
	\begin{equation} 
		D\!\left( p(\phi\mid\theta) \big\| p(\phi\mid\theta+\delta) \right) = 2\sin^2\frac{\delta}{2}\, .
	\label{eq:stopping_bound_phase} 
	\end{equation}
The corresponding asymptotic prediction is therefore 
	\begin{equation}
		\mathds{E} \left( N_{\mathrm{TP}}\mid\theta \right) \sim \frac{ \log A }{ 2\sin^2(\delta/2) }\, .
	\label{eq:phase_covariant_stopping}
	\end{equation}
Both the relative entropy and the resulting prediction are independent of $\theta$, as expected from phase covariance. For small $\delta$, 
\eqnref{eq:stopping_bound_phase} reduces to the local Fisher-information expression in \eqnref{eq:asymptotic_stopping_time_fisher}
and 
	\begin{equation}
    		\mathds{E} \left( N_{\mathrm{TP}} \right)\sim\frac{2\log A}{\delta^2}\, .
	\label{eq:phase_covaraint_small_delta_N}
	\end{equation}
Indeed, the classical Fisher information of the covariant measurement is $I(\theta)=1$, which in this example coincides with the single-copy quantum Fisher 
information. This connection concerns the classical outcome distributions generated by the measurement. It should not be interpreted as a local expansion of the 
quantum relative entropy between the pure states, which diverges for every pair of distinct phases. 

Figure~\ref{fig:phase_point_avg_samples} compares the finite-outcome i.i.d. strategies with the random projective and adaptive strategies.  Among the i.i.d. 
measurements, increasing the number of measurement directions improves performance: the 8-outcome POVM outperforms the 4-outcome POVM and approaches the 
random projective strategy. The adaptive greedy protocol performs best among the local schemes. The strategies also differ substantially in their stopping-time 
fluctuations. Finite-outcome i.i.d. measurements, particularly the 4-outcome POVM, may generate extended sequences of weakly informative outcomes and therefore 
exhibit broad, strongly right-skewed stopping-time distributions. These fluctuations decrease as the number of measurement directions increases. The adaptive 
strategy repeatedly measures near the direction of maximal local sensitivity and consequently yields a substantially narrower stopping-time distribution.

The exact KL prediction in \eqnref{eq:phase_covariant_stopping} and its small-$\delta$ Fisher-information approximation from 
\eqnref{eq:asymptotic_stopping_time_fisher} are shown for comparison. The difference between the numerical means and the KL prediction is primarily a finite-
threshold effect. The latter describes the leading asymptotic behavior of a fixed-pair likelihood-ratio process, whereas the twin-peaks protocol continuously updates 
both the estimate and its strongest $\delta$-separated competitor. For the moderate threshold used here ($A_\epsilon=19$), finite-threshold corrections and 
boundary overshoot remain visible.

We finally compare the sequential single-copy strategies with a collective strategy using a predetermined number of copies. The $N$-copy probe state can
be decomposed in the symmetric subspace as
	\begin{equation}
    		\ket{\theta}^{\otimes N}=\frac{1}{2^{N/2}} \sum_{n=0}^{N}\sqrt{\binom{N}{n}}\,  \mathrm{e}^{\mathrm{i}n\theta}   \ket{n}\, 
	\end{equation}
where
    	\begin{equation}
    		\ket{n} := \binom{N}{n}^{-1/2} \sum_{\substack{x\in\{0,1\}^{N}\\|x|=n}} \ket{x}
	\end{equation}
is the normalized symmetric Dicke state with $n$ excitations. The canonical covariant phase POVM on this subspace is
	\begin{equation}
    		\left\{E_\phi^{(N)}=\frac{\dd\phi\ (N+1)}{2\pi}\ket{\phi_N}\!\bra{\phi_N}\right\}_{\phi\in[0,2\pi)}\, ,
	\end{equation}
where
	\begin{equation}
    		\ket{\phi_N}  = \frac{1}{\sqrt{N+1}} \sum_{n=0}^{N} \mathrm{e}^{\mathrm{i}n\phi}  \ket{n}\, ,
	\end{equation}
with
	\begin{equation}
    		\int_0^{2\pi}\frac{\dd\phi\ (N+1)}{2\pi}\ket{\phi_N}\!\bra{\phi_N}\,=\one_{\mathrm{sym}}\, .
	\end{equation}
i.e. it is a valid POVM in the span of the input states.

The corresponding likelihood is
	\begin{equation}
    		p_N(\phi\mid\theta)=\frac{1}{2\pi\,2^N}\left|\sum_{n=0}^{N}\sqrt{\binom{N}{n}}\,\mathrm{e}^{\mathrm{i}n(\theta-\phi)}
    		\right|^2\, .
	\label{eq:collective_covariant_likelihood}
	\end{equation}
\eqnref{eq:collective_covariant_likelihood} depends on the unknown phase and the measurement outcome only through their difference. Changing $\phi$
therefore translates the complete likelihood profile without changing its shape. For a uniform prior, the maximum-likelihood estimate and its strongest
$\delta$-separated competitor are translated by the same amount, so the twin-peaks ratio, i.e. the LHS of \eqnref{eq:stopping_condition_twin_peaks}, is 
independent of the observed outcome. Consequently, for every threshold value $A_\epsilon$ and resolution $\delta$  there is a deterministic minimum 
number of copies $N_{\mathrm{cov}}(A_\epsilon,\delta)$ for which the twin-peaks condition is satisfied. This number is obtained by evaluating
\eqnref{eq:stopping_condition_twin_peaks} using the likelihood in \eqnref{eq:collective_covariant_likelihood} and is shown by the black line
in Fig.~\ref{fig:phase_point_avg_samples}. As the twin-peaks ratio is independent of the measurement outcome for the collective covariant strategy, 
the required number of copies is deterministic and the stopping-time fluctuations vanish. In the small-$\delta$ regime, it is straightforward to verify that 
	\begin{equation} 
		N_{\mathrm{cov}}(A_\epsilon,\delta) \sim \frac{2\log A_\epsilon}{\delta^2}\,.
	\label{eq:collective_small_delta_N} 
	\end{equation}
This coincides with the local Fisher-information prediction of \eqnref{eq:phase_covaraint_small_delta_N}, although the two approaches may differ through finite-threshold and other sub-leading corrections.
 
One may interpolate between the single-copy and fully collective strategies by applying the covariant measurement independently to blocks of $M$ copies and
evaluating the twin-peaks condition after each block. Such a protocol is sequential at the block level, with total copy cost $N=M k$, where $k$ is the
number of measured blocks. Unlike the single collective measurement on all $N$ copies, the accumulated likelihood depends on the relative outcomes of the
different blocks, and the stopping time is therefore random. We do not optimize the block size here.

In App.~\ref{app:phase_precision}, we compare the same strategies using conventional estimation figures of merit. In App.~\ref{app:entangled_probes}, we consider the 
broader metrological setting in which the probe states themselves may also be optimized, including the use of entangled probes.


\section{Purity Testing}\label{sec:purity}

We now apply the \textit{twin-peaks test} to the problem of estimating the purity of a qubit mixed state. 
Such states arise, for example, from the action of a depolarizing channel,
    \begin{equation}
    	\mathcal E_r(\ketbra{\psi})\ =\ r\,\ketbra{\psi}+ \frac{1-r}{2}\one\ =\ \frac{\one +\mathbf{r}\cdot\boldsymbol{\sigma}}{2}\, ,
	\label{eq:Bloch_decomp}
	\end{equation}
where $\boldsymbol{\sigma}$ is the vector of Pauli operators and $\mathbf{r}\in\mathbb{R}^3$, \mbox{$0\leq r=\|\mathbf{r}\|\leq1$}, the Bloch vector 
whose magnitude is directly related to the purity of the state $\rho=\mathcal E_r(\ketbra{\psi})$ via $r^2=2\tr\rho^2-1$. Henceforth, we shall 
associate the word purity with $r$.

Our local measurement strategy consists of performing a two-outcome PVM along some direction $\hat{\mathbf m}\in\mathbb{R}^3$, 
with $\|\hat{\bf m}\|\,=1$.  The corresponding PVM elements can be written in Bloch form as 
    \begin{equation}\label{eq:projectors}
        \Pi_{\pm\hat{\mathbf{m}}}=\frac{\one\pm\hat{\mathbf{m}}\cdot\boldsymbol{\sigma}}{2}\, ,
    \end{equation}
and the associated probability distribution for the measurement outcomes is given by
    \begin{equation}\label{eq:prob_dist}
        p(\pm\mid r) = \frac{1\pm\, \hat{\mathbf{m}}\cdot{\mathbf{r}}}{2}\, ,
    \end{equation}
where $\pm$ signify the $m=\pm1$ outcomes of the PVM $\Pi_{\pm\hat{\mathbf{m}}}$. Assuming that the direction of $\mathbf{r}$ is known, the optimal PVM 
is $\hat{\mathbf{m}}=\mathbf{r}/r$ and the problem reduces to estimating the parameter $p=(1+r)/2$ of a classical random variable distributed according 
to the Bernoulli distribution.  Later in this section, we devise a direction-agnostic, collective measurement strategy for estimating the purity and 
show that it asymptotically yields the same average number of samples as the local measurement strategy described above.

\begin{figure}[t]
    \centering
    \includegraphics[width=1\linewidth]{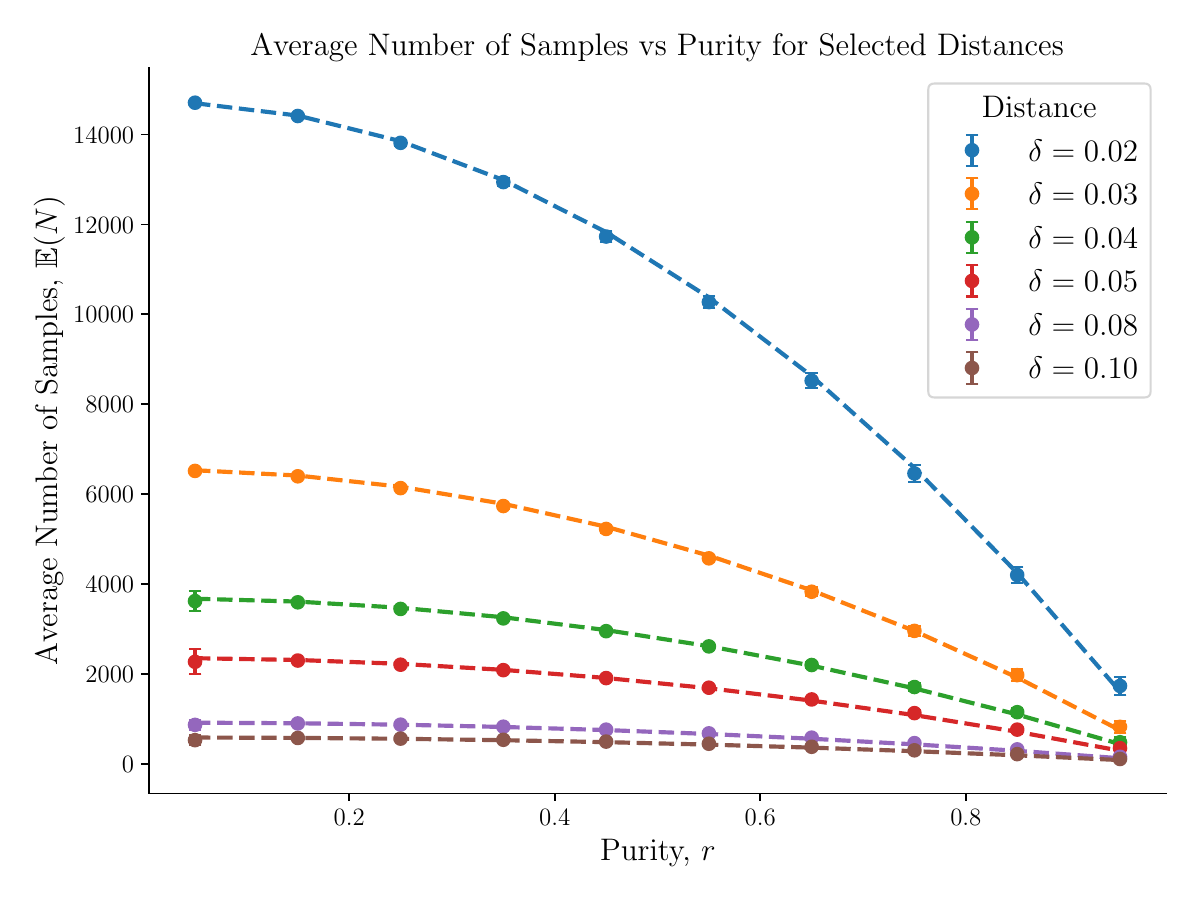}
    \caption{Average number of samples required by the \textit{twin-peaks test} for estimating the purity $r$ of a qubit mixed state $\rho$ using the evidence threshold $A_\epsilon=19$, corresponding to the calibration $\epsilon=0.05$. The average number of samples is over 100 tests. The error bars correspond to the standard deviation. The dotted lines correspond to the average number of samples required in the asymptotic limit (\eqnref{eq:asymptotic_stopping_time_purity}). In order for the numerical simulations not to run indefinitely, we set a threshold for the number of samples to 15000.}
    \label{fig:purity}
\end{figure}

Observe that since the states in \eqnref{eq:Bloch_decomp} commute for all values of $r$ neither adaptive nor collective measurements 
can improve the corresponding asymptotic distinguishability rate.  After $n$ measurements, let $n_+$ and $n_-=n-n_+$ denote the numbers of
outcomes $+1$ and $-1$ respectively. The maximum-likelihood estimate is
	\begin{equation}
    		\hat r_n= \max\left\{ 0, \frac{n_+-n_-}{n}\right\}\, .
	\label{eq:purity_MLE}
	\end{equation}
The likelihood is concave and maximized at $\hat r_n$. Its largest value outside the tolerance interval is therefore attained at
$\hat r_n-\delta$ whenever $\hat r_n\geq\delta$, and at $\hat r_n+\delta$ otherwise. The test stops once the likelihood ratio between
$\hat r_n$ and this strongest competitor reaches the evidence threshold $A_\epsilon$. For an interior true value $r>\delta$, the estimate and 
its strongest competitor converge to $r$ and $r-\delta$, respectively.  The \textit{twin-peaks} test stopping-time expression in 
\eqnref{eq:asymptotic_stopping_time_classic} then gives the large-threshold limit
	\begin{equation}\label{eq:asymptotic_stopping_time_purity}
        	\mathds{E}(N\mid r)\sim\frac{2\log{A_\epsilon}}{(1+r) \log \left(\frac{1+r}{1+r-\delta}\right)+(1-r) \log \left(\frac{1-r}{1-r+\delta}\right)}\, .
    	\end{equation}
For $r<\delta$, the corresponding expression is obtained by replacing the competitor $r-\delta$ with $r+\delta$. The relation is asymptotic rather than
exact because the stopped likelihood ratio generally overshoots $\log A_\epsilon$ and, during the initial rounds, both the estimate and its
strongest competitor remain data dependent. For small $\delta$ and fixed $r$ in the interior of the parameter space, it holds
	\begin{equation}\label{eq:asymptotic_stopping_time_purity_no_overshoot}
        	\mathds{E}(N\mid r) \sim \frac{2(1-r^2)\log{A_\epsilon}}{\delta^2}\, .
   	 \end{equation}
This agrees with the general Fisher-information expression in \eqnref{eq:asymptotic_stopping_time_fisher}, since $  I(r) = 1/(1-r^2)$.
For this commuting family, the classical Fisher information extracted by the measurement in the direction of the Bloch vector coincides with the  
quantum Fisher information.

\figref{fig:purity} compares the numerical mean stopping times with the KL prediction of \eqnref{eq:asymptotic_stopping_time_purity} as a function of  
the purity, $r$, and for various distances.  Specifically, the domain $0\leq r\leq 1$ is divided into 10 equally spaced intervals, and we randomly select a 
value from each interval.  For each such value, we perform 100 tests, and we set the calibration at $\epsilon=0.05$.  In order to guarantee that the 
tests stop we set an upper threshold on the number of samples equal to 15000.  The mean stopping time decreases with $r$, reflecting the increase in 
Fisher information towards the pure-state boundary. In this example, the required numbers of measurements are considerably larger than those used in 
the phase-testing example. As the likelihood depends only on the two outcome counts $(n_+,n_-)$, however, the simulations remain efficient even for these 
longer runs. The larger stopping times place the protocol more deeply in the asymptotic regime, which is consistent with the good agreement observed between 
the numerical and analytical curves.

The numerical stopping times also exhibit comparatively small fluctuations. This can be understood from the concentration of the maximum-likelihood
estimate,
	\begin{equation}
    		\operatorname{Var}(\hat r_n\mid r)=\frac{1-r^2}{n}+o\!\left(\frac{1}{n}\right)\,.
	\label{eq:purity_MLE_variance}
	\end{equation}
For small $\delta$, the twin-peaks statistic accumulated after $n$ rounds is
approximately
	\begin{equation}
    		\frac{n\delta^2}{2\left(1-\hat r_n^2\right)}\, .
	\end{equation}
The threshold-crossing condition therefore gives
	\begin{equation}
    		N_{\mathrm{TP}}\simeq\frac{2\left(1-\hat r_{N_{\mathrm{TP}}}^{\,2}\right)\log A_\epsilon}{\delta^2}\, .
	\end{equation}
Linearizing this expression around the true value yields
	\begin{equation}
    		N_{\mathrm{TP}}-\mathds{E}(N_{\mathrm{TP}}\mid r)\simeq-\frac{4r\log A_\epsilon}{\delta^2}
    		\left(\hat r_{N_{\mathrm{TP}}}-r\right)\, .
	\end{equation}
Evaluating \eqnref{eq:purity_MLE_variance} at the leading mean stopping time in \eqnref{eq:asymptotic_stopping_time_purity_no_overshoot} then gives
	\begin{equation}
 		\operatorname{Var}(N_{\mathrm{TP}}\mid r)\sim\frac{8r^2\log A_\epsilon}{\delta^2}\, ,
	\label{eq:purity_variance_small_delta}
	\end{equation}
and hence
	\begin{equation}
    		\Delta N_{\mathrm{TP}}=\sqrt{ \operatorname{Var}(N_{\mathrm{TP}}\mid r)}\sim\frac{2r}{\delta}\sqrt{ 2\log A_\epsilon}\, .
	\label{eq:purity_std_small_delta}
	\end{equation}
Thus, while the mean stopping time grows as $\delta^{-2}$, its standard deviation grows only as $\delta^{-1}$. The relative fluctuations consequently
vanish in the small-$\delta$ limit,
	\begin{equation}
    		\frac{\Delta N_{\mathrm{TP}}}{\mathds{E}(N_{\mathrm{TP}}\mid r)}\sim\frac{\sqrt{2}\,r\,\delta}{ (1-r^2)\sqrt{\log A_\epsilon}   }\, .
	\label{eq:purity_relative_fluctuations}
	\end{equation}
This explains the narrow stopping-time distributions observed in Fig.~\ref{fig:purity}. The fluctuation calculation is a leading asymptotic approximation: it neglects the 
initial transient, the overshoot of the stopping boundary, and higher-order corrections. At $r=0$, the linear term vanishes and the dominant fluctuations require a 
higher-order expansion.  In \figref{fig:purity} this can be seen by the absence of error bars, or unrealistically small error bars, for the cases of small $\delta$ 
and small $r$, since as for these cases, the sequential \textit{twin-peaks} test does not stop before the 15000 threshold 
for the majority of the 100 runs. 

These finite-sample effects are also visible in Fig.~\ref{fig:log_samples}, which compares the numerical mean stopping times with the relative-entropy prediction in 
\eqnref{eq:asymptotic_stopping_time_purity} as a function of the tolerance $\delta$. The agreement improves as $\delta$ decreases, when the larger stopping times 
allow the estimate and its strongest competitor to stabilize. For larger $\delta$, the test typically stops after comparatively few measurements, making finite-threshold corrections, boundary overshoot, and the discreteness of the likelihood-ratio process more relevant. These effects are particularly pronounced near the 
endpoints of the parameter space. Near $r=0$, the constrained maximum-likelihood estimate frequently reaches the boundary, and the strongest competitor may switch 
between the two sides of the tolerance region. Near $r=1$, only the lower competitor is admissible, and the rarity of outcome $1$ produces a strongly discrete 
likelihood-ratio process. The relative-entropy curve should therefore be interpreted as a large-threshold prediction rather than an exact finite-sample expression.
\begin{figure}[t]
    \centering
    \includegraphics[width=1\linewidth]{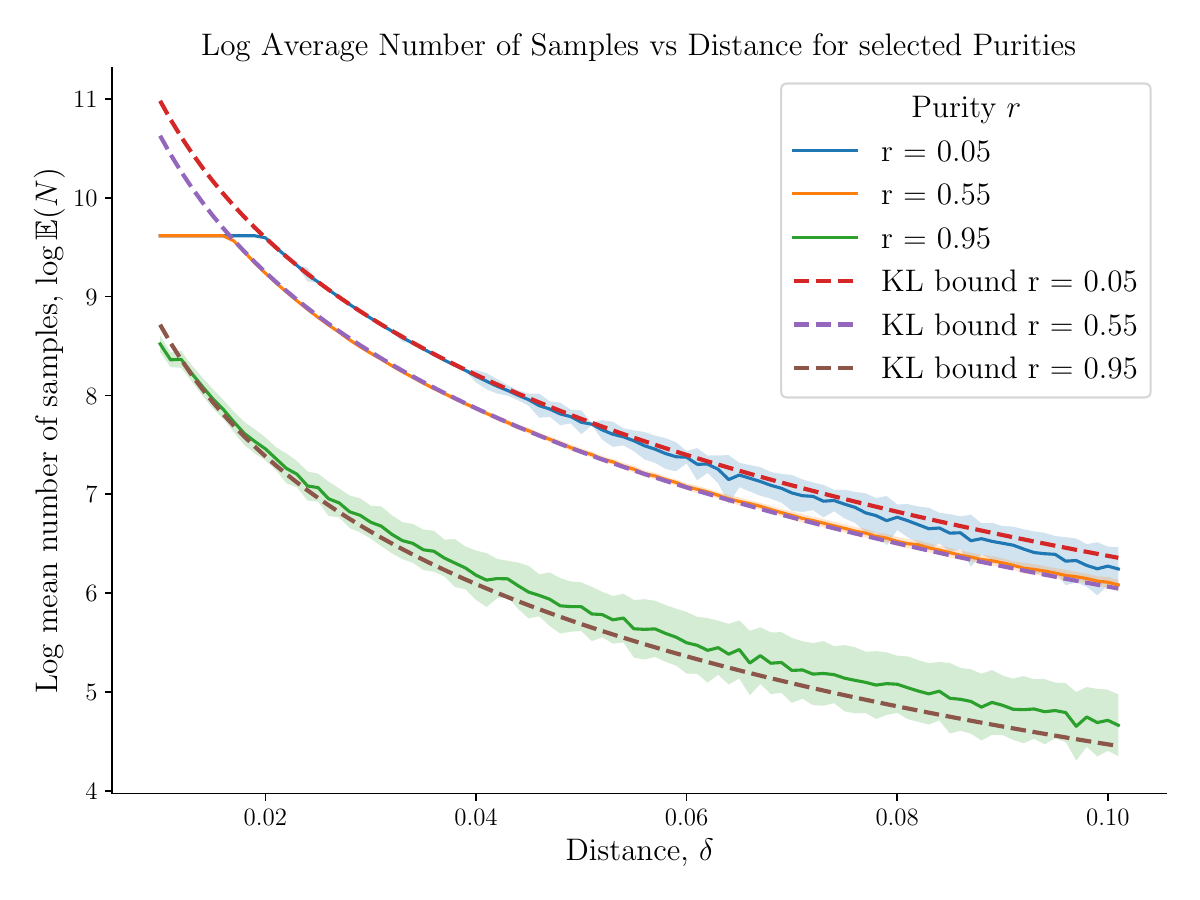}
    \caption{Logarithm of the average number of samples required by the \textit{twin-peaks test} for estimating the purity $r$ of a qubit mixed state $\rho$ as a function of the distance $\delta$. The calibration is $\epsilon=0.05$. The average number of samples is computed over 100 tests, with the shaded region corresponding to the standard deviation.  The dashed lines correspond to the asymptotic bound of \eqnref{eq:asymptotic_stopping_time_purity}. 
    In order for the numerical simulations not to run indefinitely, we set a threshold for the number of samples to 15000; the flat orange and blue lines show this user-defined threshold value for small values of $\delta$.}
    \label{fig:log_samples}
\end{figure}

Let us turn to the case of sequentially estimating the purity when the direction of the Bloch vector is unknown.
To that end, consider the following sequential strategy: at each round, $k$, we are given $M$ copies of the state $\rho$ 
and perform the weak Schur measurement, i.e., the PVM $\{\Pi_J=\sum_{m=-J}^J\ket{J,m}_{\hat{n}}\bra{J,m}\, |\, 
0(\tfrac{1}{2})\leq J\leq \tfrac{M}{2}\}$.  
Here $\ket{J,m}_{\hat{n}}$ are the simultaneous eigenstates of the total angular 
momentum operator ${\bf J}^2$, and its projection along some direction $\hat{n}$, $\bf{J}_{\hat{n}}$, not necessarily the 
direction of the Bloch vector $\hat{\bf r}$.  The conditional probability distribution for this measurement is given by
    \begin{align}\label{eq:Schur_sampling}
        p_M&(J\mid r)=\tr(\Pi_J \rho^{\otimes M})\nonumber\\
        &=\left(\frac{1-r^2}{4}\right)^{\frac{M}{2}}\binom{M}{\frac{M}{2}-J}\frac{2J+1}{\frac{M}{2}+J+1}\frac{R^{J+1}-R^{-J}}{R-1}\, ,
    \end{align}
where $R=(1+r)/(1-r)$. The average stopping time, $\mathds{E}(N\mid r)$, is given by 
\eqnref{eq:asymptotic_stopping_time_purity} with $p(m\mid  r)$ replaced by $p(J\mid  r)$ of \eqnref{eq:Schur_sampling}.     
In the limit of large $M$, the probability distribution of \eqnref{eq:Schur_sampling} can be approximated by a Gaussian
distribution of mean $\mu(r)=\tfrac{Mr}{2}$, and variance $\sigma^2(r)=\tfrac{M (1-r^2)}{4}$~\cite{Guta2006}. 
Taking the asymptotic limit so that the overshooting of our test statistic tends to zero, one finds
    \begin{align}\label{eq:collective_avg_samples}
        \mathds{E}(N^{(M)}\mid r)\sim \frac{2\log{A_\epsilon}}{M \delta^2 I(r)}\, \Rightarrow\, 
        \mathds{E}(N\mid r)\sim \frac{2\log{A_\epsilon}}{\delta^2 I(r)}\, ,
    \end{align}
exactly equaling the asymptotic time of our local sequential strategy, where the direction of $\hat{r}$ is known. In \appref{app:purity_collective}
we show how the total average number of samples $\mathbb{E}(N^{(M)}\mid r)$, diminishes with the number $M$ of states measured collectively in each round, with 
the largest savings occurring in the case of highly mixed states.

Hence, in the frequentist regime where one wishes to locate the purity to within a very small interval in parameter space, if the direction of the 
Bloch vector is not known, collective measurement strategies perform no better than a strategy that first sacrifices a \textit{vanishing fraction} of 
the samples to estimate the direction, $\hat{\bf r}_{\mathrm{est}}$, of the Bloch vector and then performs a local sequential test using the PVM 
$\Pi_{\hat{\bf r}_{\mathrm{est}}}$.  This is in line with Ref.~\cite{PhysRevLett.95.110504}, where it 
was shown that such a strategy using, for example, $\sqrt{n}$ copies to estimate the direction and then estimate the purity along this direction with the remaining 
systems yields the optimal asymptotic MSE.

Finally, it is instructive to compare the performance of the \textit{twin-peaks test} with a test using a deterministic number of samples.
Since no collective measurement strategy employing multiple copies can outperform the bound given in \eqnref{eq:asymptotic_stopping_time_purity}, we 
take the r.h.s. as a lower bound to the sample size of the optimal deterministic strategy. Also, the sequential strategy matches this bound asymptotically, and 
we can calculate the expected number of samples over values of $r$ according to \eqnref{eq:average_asymptotic_stopping_time}, shown as a green line in 
Fig.~\ref{fig:fixed_vs_seq}. The expected error and worst error are then given by \eqnsref{eq:expected_error}{eq:wc_error}, and we use this to 
calculate the number of samples that lead to the desired expected error probability. In Fig.~\ref{fig:fixed_vs_seq}, these are shown as a dotted green line and a 
dashed black line, respectively and we can see that the sequential strategy is the one that requires fewer samples on average.

Since we also consider a scenario where the direction of the Bloch vector is not known, a uniform distribution in parameter space could be considered 
inapt. Two natural choices are a uniform distribution in the Hilbert-Schmidt metric, which results in the prior $\pi(r)=3r^2$, and a uniform distribution in 
Bures metric, which results in the prior $\pi(r)=4 r^2/(\pi\sqrt{1-r^2})$. For these priors, we can calculate the corresponding expected number of samples 
for the sequential strategy, shown as red and yellow lines in Fig.~\ref{fig:fixed_vs_seq}, and the corresponding fixed sample-size strategies that 
achieve the same expected error, shown as dotted lines of the corresponding color in Fig.~\ref{fig:fixed_vs_seq}. Both the Hilbert-Schmidt and the 
Bures metric have a vanishing prior for $r=0$. As a consequence, there is no upper bound on the worst error case, as the ratio between prior 
probabilities in \eqnref{eq:asymptotic_stopping_time_classic} diverges.

Nonetheless, for any positive lower bound on the prior probability (or value of $r$), the ratio of the priors is a constant, and the asymptotic behavior is dominated by small $\epsilon$. Using \eqnref{eq:asymptotic_stopping_time_fisher}, the expected number of samples can be calculated as
    \begin{equation}
        \mathbb E(N)\sim\frac{-2\log\epsilon}{\delta^2}\int_0^1 \pi(r)(1-r^2)\,\mathrm{d}r\, .
    \end{equation}
For a flat prior, the \textit{twin-peaks test} has an expected stopping time $\mathbb E(N)\sim-4/3\log(\epsilon)\delta^{-2}$, for the Hilbert-Schmidt 
uniformly distributed purity, we get $\mathbb E(N)\sim-4/5\log(\epsilon)\delta^{-2}$, and finally for the Bures metric uniformly distributed purity we 
get $\mathbb E(N)\sim-1/2\log(\epsilon)\delta^{-2}$. Because the purity is unknown, a fixed-sample approach must use a sample size sufficient to cover the worst-case scenario. A sequential test, by contrast, learns the purity as data are collected and adaptively determines the number of samples required.  The fixed-sample 
benchmark is controlled by $r=0$, which results in a number of samples that scales as $N^\infty_{\text{fss}}(d,\epsilon)\sim- 2\log(\epsilon)\delta^{-2}$. Therefore, 
sequential strategies need a constant factor of the resources, depending on the prior distribution. In the considered cases, these factors are $2/3$ for a flat 
distribution, $2/5$ for the Hilbert-Schmidt uniform samples, and $1/4$ in the case of Bures uniform samples. 

\begin{figure}[t]
    \centering
    \includegraphics[width=1\linewidth]{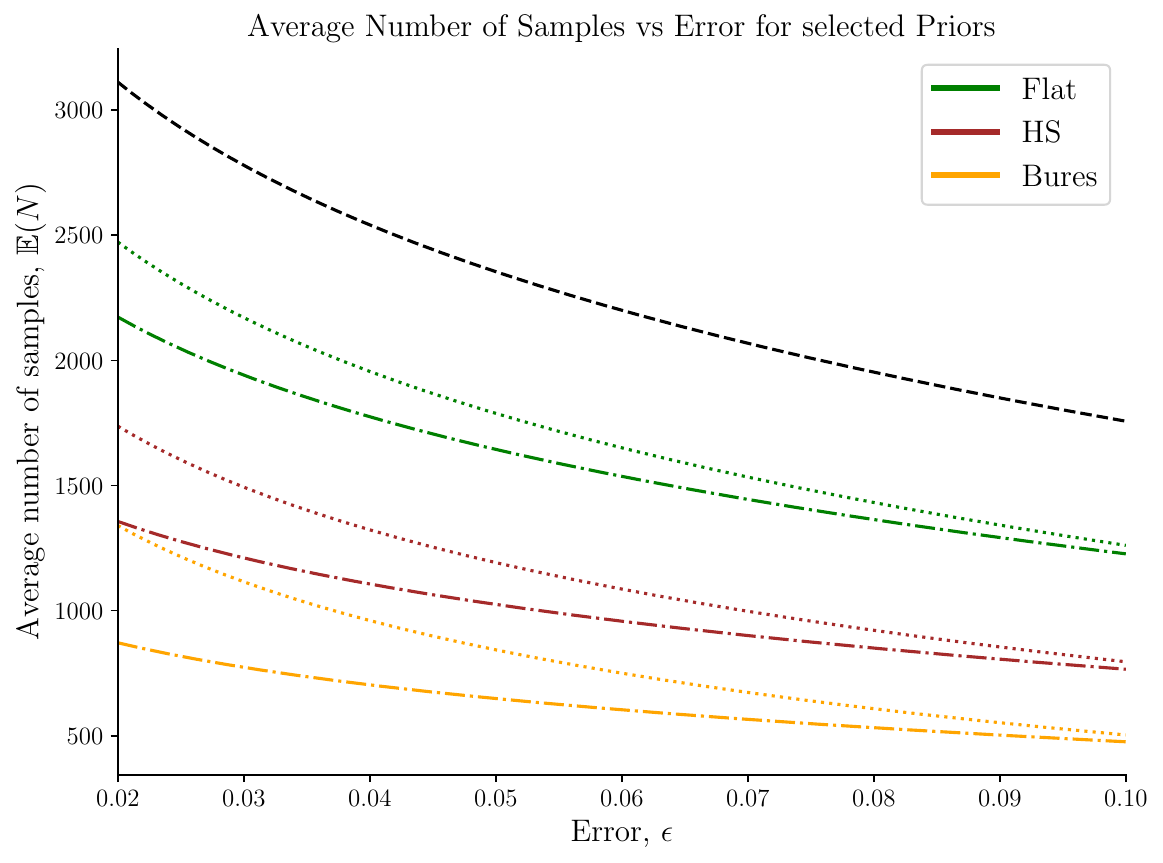}
    \caption{Average number of samples $\mathbb{E}(N)$ for a sequential strategy in dashed-dotted lines vs. the error of $\epsilon$ for a minimal distance between 
    hypotheses $\delta=0.05$. Different colors indicate different priors: the flat prior is shown in green, the Hilbert-Schmidt uniform prior in red, and the 
    Bures uniform prior in orange. The dotted lines show the number of samples for fixed-sample-size strategies such that the expected error matches the 
    error of the corresponding sequential strategy. The black dashed line shows the number of samples for the worst-case error for a flat prior.}
    \label{fig:fixed_vs_seq}
\end{figure}

\section{Conclusions}

In this work, we introduced \textit{parameter testing}, a sequential framework for determining an unknown continuous parameter up to a prescribed tolerance. Once an estimate $\hat{\vartheta}$ is reported, parameter values within a neighborhood $B_\delta(\hat{\vartheta})$ are regarded as 
acceptable, while sufficiently distant values are treated as competing alternatives that must be ruled out with a target error tolerance. In this sense, 
parameter testing retains the likelihood-ratio logic of hypothesis testing while extending it to a continuous family of hypotheses.

The distinction from conventional parameter estimation lies primarily in the operational figure of merit. Standard estimation typically quantifies the 
quality of an estimate through a loss function that depends on its distance from the true parameter, such as the mean squared error. Parameter testing 
is instead concerned with separating an acceptable region from competing parameter values outside a prescribed tolerance. The two approaches are 
therefore complementary: distance-based estimation is natural when the cost of an inaccurate estimate varies smoothly with the estimation error, 
whereas parameter testing is better suited to tasks with a threshold between acceptable and unacceptable parameter values.

Building on sequential hypothesis testing, we developed three sequential procedures for parameter testing and analyzed the non-trivial challenges that arise when passing from a discrete to a continuous set of hypotheses. In particular, we introduced the \textit{complement}, \textit{concentration}, and \textit{twin-peaks} tests. Among them, the \textit{twin-peaks test} provides a particularly natural continuous analog of the sequential likelihood-ratio testing: it compares the most likely current estimate with the strongest competing parameter value outside the prescribed tolerance region. At the same 
time, it is computationally simple and, in the asymptotic regime, leads to the same stopping condition as the more explicitly region-based complement 
test.

We applied this framework to two paradigmatic quantum parameter-testing tasks: determining the phase and the purity of a qubit. These examples illustrate 
two qualitatively different situations. For phase testing, the statistical model is symmetric in the unknown phase, but the optimal measurement 
generally depends on the parameter itself. We compared i.i.d., adaptive, and collective measurement strategies, as well as their performance with 
asymptotic predictions based on the Fisher information and the KL divergence. We found that adaptive projective measurements provide the best sequential 
strategy and achieve the same average sample cost as collective covariant measurements performed with a predetermined number of copies.

Purity testing, by contrast, is effectively classical when the direction of the Bloch vector is known. The optimal projective measurement is 
independent of the unknown purity, and the resulting KL divergence coincides with the quantum relative entropy. Consequently, neither adaptive nor 
collective strategies can improve upon local i.i.d. measurements. In this case, the stopping time depends strongly on the true purity, allowing 
sequential parameter testing to exploit the information contained in each experimental run and to yield significant average sample savings over fixed 
sample-size protocols. We further showed that, when the Bloch vector direction is unknown, sequential procedures acting on finite batches of copies can 
asymptotically recover the same performance.

Overall, our results establish sequential parameter testing as an operationally meaningful framework for continuous-parameter problems in which the 
objective is to certify a parameter value up to a prescribed tolerance rather than to minimize a distance-based estimation error. By combining 
threshold-based certification with data-dependent stopping, the framework provides a natural route towards resource-efficient protocols for quantum
information tasks in which the usefulness of a state or device is determined by whether an underlying parameter lies within an acceptable region.
\section{Acknowledgments}
SM is thankful to Mohammad Mehboudi, Paolo Abiuso, Jes\'us Rubio and Luis A. Correa for insightful discussions.
SM is a recipient of an APART-MINT Fellowship of the Austrian Academy of Sciences at the Atominstitut of the Technische Universit\"at Wien.
RRR acknowledges the Government of Spain (Severo Ochoa CEX2019-000910-S and 
FUNQIP), Fundaci\'o Cellex, Fundaci\'o Mir-Puig, Generalitat de Catalunya (CERCA program).
JCC and RMT acknowledge financial support  from  MICIN grant 
PID2022-141283NB-I00 funded by MCIN/AEI/10.13039/501100011033.
GS acknowledges funding from the Programa Talent UAB - Banco de Santander.
MS acknowledges project PID2024-162141OB-I00, funded by MI-CIU/AEI/10.13039/501100011033/ FEDER, UE, Ayuda Ram\'on y Cajal 2021 (RYC2021-032032-I,
MICIU/AEI/10.13039/501100011033, ESF+) as well as Project FEDER C-EXP-256-UGR23 Consejer\'ia de Universidad, Investigaci\'on e Innovaci\'on y UE 
Programa FEDER Andaluc\'ia 2021-2027.

\bibliography{bibfile.bib}

@article{Pallister2018,
  title = {Optimal Verification of Entangled States with Local Measurements},
  author = {Pallister, Sam and Linden, Noah and Montanaro, Ashley},
  journal = {Phys. Rev. Lett.},
  volume = {120},
  issue = {17},
  pages = {170502},
  numpages = {5},
  year = {2018},
  month = {Apr},
  publisher = {American Physical Society},
  doi = {10.1103/PhysRevLett.120.170502},
  url = {https://link.aps.org/doi/10.1103/PhysRevLett.120.170502}
}

@book{cramer1999mathematical,
    author = {Cram{\'e}r, Harald},
    title = {Mathematical methods of statistics},
    series = {Princeton Landmarks in Mathematics},
    volume = {26},
    year = {1999},
    publisher = {Princeton University Press},
    doi = {10.1515/9781400883868},
    url = {https://doi.org/10.1515/9781400883868}
}

@incollection{rao1992information,
    author = {Rao, C. Radhakrishna},
    title = {Information and the accuracy attainable in the estimation of statistical parameters},
    booktitle = {Breakthroughs in Statistics: Foundations and basic theory},
    pages = {235--247},
    year = {1992},
    publisher = {Springer},
    doi = {10.1007/978-1-4612-0919-4_16}
}

@article{vargas2021quantum,
    author = {Mart{\'\i}nez-Vargas, Esteban and Hirche, Christoph and Sent{\'\i}s, Gael and Skotiniotis, Michalis and Carrizo, Marta and Mu{\~n}oz-Tapia, Ramon and Calsamiglia, John},
    title = {Quantum sequential hypothesis testing},
    journal = {Phys. Rev. Lett.},
    volume = {126},
    number = {18},
    pages = {180502},
    year = {2021},
    doi = {10.1103/PhysRevLett.126.180502}
}

@article{baum1994sequential,
    author = {Baum, Carl W. and Veeravalli, Venugopal V.},
    title = {A sequential procedure for multihypothesis testing},
    journal = {IEEE Trans. Inf. Theory},
    volume = {40},
    number = {6},
    pages = {1994--2007},
    year = {1994},
    doi = {10.1109/18.340472}
}

@article{wald_sequential_1945,
    author = {Wald, A.},
    title = {Sequential {Tests} of {Statistical} {Hypotheses}},
    journal = {Ann. Math. Stat.},
    volume = {16},
    number = {2},
    pages = {117--186},
    year = {1945},
    doi = {10.1214/aoms/1177731118}
}

@article{Holevo1984,
    author = {Holevo, Alexander S.},
    title = {Covariant measurements and imprimitivity systems},
    journal = {Lect. Notes Math.},
    volume = {1055},
    pages = {153},
    year = {1984},
    doi = {10.1007/BFb0071720}
}

@article{perez2025bounds,
    author = {P{\'{e}}rez-Guijarro, Jordi and Pag{\`{e}}s-Zamora, Alba and Fonollosa, Javier R.},
    title = {Bounds in {S}equential {U}nambiguous {D}iscrimination of {M}ultiple {P}ure {Q}uantum {S}tates},
    journal = {Quantum},
    volume = {9},
    pages = {1919},
    year = {2025},
    doi = {10.22331/q-2025-11-20-1919}
}

@article{Bartlett_2007,
    author = {Bartlett, Stephen D. and Rudolph, Terry and Spekkens, Robert W.},
    title = {Reference frames, superselection rules, and quantum information},
    journal = {Rev. Mod. Phys.},
    volume = {79},
    number = {2},
    pages = {555--609},
    year = {2007},
    doi = {10.1103/RevModPhys.79.555}
}

@misc{zecchin2025quantumsequentialuniversalhypothesis,
    author = {Zecchin, Matteo and Simeone, Osvaldo and Ramdas, Aaditya},
    title = {Quantum Sequential Universal Hypothesis Testing},
    year = {2025},
    eprint = {2508.21594},
    archivePrefix = {arXiv}

}

@article{Draglia_1999,
    author = {Draglia, V. P. and Tartakovsky, A. G. and Veeravalli, V. V.},
    title = {Multihypothesis sequential probability ratio tests .I. Asymptotic optimality},
    journal = {IEEE Trans. Inf. Theory},
    volume = {45},
    number = {7},
    pages = {2448--2461},
    year = {1999},
    doi = {10.1109/18.796383}
}

@article{Slussarenko_2017,
    author = {Slussarenko, Sergei and Weston, Morgan M. and Li, Jun-Gang and Campbell, Nicholas and Wiseman, Howard M. and Pryde, Geoff J.},
    title = {Quantum State Discrimination Using the Minimum Average Number of Copies},
    journal = {Phys. Rev. Lett.},
    volume = {118},
    number = {3},
    pages = {030502},
    year = {2017},
    doi = {10.1103/PhysRevLett.118.030502}
}

@article{VeeravalliBaum1995,
    author = {Veeravalli, V. V. and Baum, C. W.},
    title = {Asymptotic efficiency of a sequential multihypothesis test},
    journal = {IEEE Trans. Inf. Theory},
    volume = {41},
    number = {6},
    pages = {1994--1997},
    year = {1995},
    doi = {10.1109/18.476323}
}

@article{Lorden77,
    author = {Lorden, Gary},
    title = {Nearly-Optimal Sequential Tests for Finitely Many Parameter Values},
    journal = {Ann. Statist.},
    volume = {5},
    number = {1},
    pages = {1--21},
    year = {1977},
    doi = {10.1214/aos/1176343737},
    url = {https://doi.org/10.1214/aos/1176343737}
}

@article{Armitage50,
    author = {Armitage, P.},
    title = {Sequential Analysis with More than Two Alternative Hypotheses, and its Relation to Discriminant Function Analysis},
    journal = {J. R. Stat. Soc. Ser. B},
    volume = {12},
    number = {1},
    pages = {137--144},
    year = {1950},
    doi = {10.1111/j.2517-6161.1950.tb00050.x},
    url = {https://doi.org/10.1111/j.2517-6161.1950.tb00050.x}
}

@article{Li_2022,
    author = {Li, Yonglong and Tan, Vincent Y. F. and Tomamichel, Marco},
    title = {Optimal Adaptive Strategies for Sequential Quantum Hypothesis Testing},
    journal = {Commun. Math. Phys.},
    volume = {392},
    number = {3},
    pages = {993--1027},
    year = {2022},
    doi = {10.1007/s00220-022-04362-5}
}

@book{DAgostini2003,
    author = {D'Agostini, Giulio},
    title = {Bayesian Reasoning in Data Analysis: A Critical Introduction},
    publisher = {World Scientific},
    year = {2003},
    doi = {10.1142/5262}
}

@book{Bolstad2009,
    author = {Bolstad, William M.},
    title = {Understanding Computational Bayesian Statistics},
    publisher = {John Wiley \& Sons},
    year = {2009},
    doi = {10.1002/9780470567371}
}

@book{AmaralTurkmanPaulinoMuller2019,
    author = {Amaral Turkman, M. Ant\'onia and Paulino, Carlos Daniel and M\"uller, Peter},
    title = {Computational Bayesian Statistics: An Introduction},
    publisher = {Cambridge University Press},
    year = {2019},
    doi = {10.1017/9781108646185}
}

@article{ShorPreskill2000,
    author = {Peter W. Shor and John Preskill},
    title = {Simple Proof of Security of the {BB}84 Quantum Key Distribution Protocol},
    journal = {Phys. Rev. Lett.},
    volume = {85},
    number = {2},
    pages = {441--444},
    year = {2000},
    doi = {10.1103/PhysRevLett.85.441}
}

@article{AharonovBenOr1997,
    author = {Dorit Aharonov and Michael Ben-Or},
    title = {Fault-Tolerant Quantum Computation with Constant Error Rate},
    journal = {SIAM J. Comput.},
    volume = {38},
    number = {4},
    pages = {1207--1282},
    year = {2008},
    doi = {10.1137/S0097539799359385}
}

@article{RaussendorfHarringtonGoyal2006,
    author = {Robert Raussendorf and Jim Harrington and Kovid Goyal},
    title = {A Fault-Tolerant One-Way Quantum Computer},
    journal = {Ann. Phys.},
    volume = {321},
    number = {9},
    pages = {2242--2270},
    year = {2006},
    doi = {10.1016/j.aop.2006.01.012}
}

@article{AcinGisinMasanes2006,
    author = {Antonio Ac{\'i}n and Nicolas Gisin and Lluis Masanes},
    title = {From {B}ell's Theorem to Secure Quantum Key Distribution},
    journal = {Phys. Rev. Lett.},
    volume = {97},
    number = {12},
    pages = {120405},
    year = {2006},
    doi = {10.1103/PhysRevLett.97.120405}
}

@article{PironioAcinMasanes2010,
    author = {Stefano Pironio and Antonio Ac{\'i}n and Serge Massar and Valerio Scarani and Ronald de Wolf and Nicolas Gisin},
    title = {Random Numbers Certified by {B}ell's Theorem},
    journal = {Nature},
    volume = {464},
    pages = {1021--1024},
    year = {2010},
    doi = {10.1038/nature09008}
}

@article{Werner1989,
    author = {Reinhard F. Werner},
    title = {{Quantum States with Einstein-Podolsky-Rosen Correlations Admitting a Hidden-Variable Model}},
    journal = {Phys. Rev. A},
    volume = {40},
    number = {8},
    pages = {4277--4281},
    year = {1989},
    doi = {10.1103/PhysRevA.40.4277}
}

@article{Horodecki1995,
    author = {Ryszard Horodecki and Pawe{\l} Horodecki and Micha{\l} Horodecki},
    title = {Violating {B}ell Inequality by Mixed Spin-1/2 States},
    journal = {Phys. Lett. A},
    volume = {200},
    number = {5-6},
    pages = {340--344},
    year = {1995},
    doi = {10.1016/0375-9601(95)00214-N}
}

@article{HiaiPetz1991,
    author = {Hiai, Fumio and Petz, D{\'e}nes},
    title = {The proper formula for relative entropy and its asymptotics in quantum probability},
    journal = {Commun. Math. Phys.},
    volume = {143},
    number = {1},
    pages = {99--114},
    year = {1991},
    doi = {10.1007/BF02100287}
}

@article{OgawaNagaoka2000,
    author = {Ogawa, Tomohiro and Nagaoka, Hiroshi},
    title = {Strong converse and {Stein's} lemma in quantum hypothesis testing},
    journal = {IEEE Trans. Inf. Theory},
    volume = {46},
    number = {7},
    pages = {2428--2433},
    year = {2000},
    doi = {10.1109/18.887855}
}

@book{holevo2011probabilistic,
    author = {Holevo, A. S.},
    doi = {10.1007/978-88-7642-378-9},
    isbn = {9788876423789},
    publisher = {Edizioni della Normale, Springer Basel},
    title = {{Probabilistic and Statistical Aspects of Quantum Theory}},
    url = {https://doi.org/10.1007/978-88-7642-378-9},
    year = {2011},
}

@article{Hayashi2002,
    author = {Hayashi, Masahito},
    title = {Optimal sequence of quantum measurements in the sense of {Stein's lemma}},
    journal = {J. Phys. A: Math. Gen.},
    volume = {35},
    number = {50},
    pages = {10759--10773},
    year = {2002},
    doi = {10.1088/0305-4470/35/50/307}
}

@article{Audenaert2007,
    author = {Audenaert, K. M. R. and Calsamiglia, J. and Mu\~noz-Tapia, R. and Bagan, E. and Masanes, Ll. and Ac{\'i}n, A. and Verstraete, F.},
    title = {Discriminating {States}: {The} {Quantum} {Chernoff} {Bound}},
    journal = {Phys. Rev. Lett.},
    volume = {98},
    number = {16},
    pages = {160501},
    year = {2007},
    doi = {10.1103/PhysRevLett.98.160501}
}

@article{Lami2025GQSL,
    author = {Lami, Ludovico},
    title = {A Solution of the Generalised Quantum {Stein's Lemma}},
    journal = {IEEE Trans. Inf. Theory},
    volume = {71},
    number = {6},
    pages = {4454--4484},
    year = {2025},
    doi = {10.1109/TIT.2025.3543610}
}

@article{Berta_2017,
    author = {Berta, Mario and Fawzi, Omar and Tomamichel, Marco},
    title = {On variational expressions for quantum relative entropies},
    journal = {Lett. Math. Phys.},
    volume = {107},
    number = {12},
    pages = {2239--2265},
    year = {2017},
    doi = {10.1007/s11005-017-0990-7}
}

@book{AmariShun-ichi2000Moig,
  author    = {Amari, Shunichi and Nagaoka, Hiroshi},
  title     = {Methods of Information Geometry},
  series    = {Translations of Mathematical Monographs},
  volume    = {191},
  publisher = {American Mathematical Society},
  address   = {Providence, RI},
  year      = {2000},
  isbn      = {0821805312},
  doi       = {10.1090/mmono/191},
  url       = {https://doi.org/10.1090/mmono/191}
}

@article{Wald1948,
author = {A. Wald and J. Wolfowitz},
title = {{Optimum Character of the Sequential Probability Ratio Test}},
volume = {19},
journal = {Ann. Math. Stat.},
number = {3},
publisher = {Institute of Mathematical Statistics},
pages = {326 -- 339},
year = {1948},
doi = {10.1214/aoms/1177730197},
URL = {https://doi.org/10.1214/aoms/1177730197}
}

@misc{rodasalichs2025sequentialanalysiscontinuousspinnoise,
    author={Elisabet Roda-Salichs and Giulio Gasbarri and Antoni Alou and Michalis Skotiniotis and Aleksandra Sierant and Diana M\'endez-Avalos and Morgan W. Mitchell and John Calsamiglia},
  title = {Sequential analysis in a continuous spin-noise quantum sensor},
  year = {2025},
  eprint = {2509.16177},
archivePrefix={arXiv}
}

@misc{andre2026strategyoptimizationbayesianquantum,
      title={{Strategy optimization for Bayesian quantum parameter estimation with finite copies: Adaptive greedy, parallel, sequential, and general strategies}}, 
      author={Erik L. Andr\'e and Jessica Bavaresco and Mohammad Mehboudi},
      year={2026},
      eprint={2602.09655},
      archivePrefix={arXiv}
}

@misc{simpson2026samplecomplexitycompositequantum,
      title={Sample Complexity of Composite Quantum Hypothesis Testing}, 
      author={Jacob Paul Simpson and Efstratios Palias and Sharu Theresa Jose},
      year={2026},
      eprint={2601.08588},
      archivePrefix={arXiv}
}

@misc{lami2025generalisedquantumsanovtheorem,
      title={Generalised quantum {S}anov theorem revisited}, 
      author={Ludovico Lami},
      year={2025},
      eprint={2510.06340},
      archivePrefix={arXiv}
}

@misc{lami2025doublycompositechernoffsteinlemma,
      title={A doubly composite {C}hernoff-{S}tein lemma and its applications}, 
      author={Ludovico Lami},
      year={2025},
      eprint={2510.06342},
      archivePrefix={arXiv}
}

@article{Mosonyi_2014,
   title={Quantum Hypothesis Testing and the Operational Interpretation of the Quantum {R}\'enyi Relative Entropies},
   volume={334},
   ISSN={1432-0916},
   url={http://dx.doi.org/10.1007/s00220-014-2248-x},
   DOI={10.1007/s00220-014-2248-x},
   number={3},
   journal={Communications in Mathematical Physics},
   publisher={Springer Science and Business Media LLC},
   author={Mosonyi, Mil\'an and Ogawa, Tomohiro},
   year={2014},
   month=Dec, pages={1617--1648} }

@article{HOLEVO1973337,
title = {Statistical decision theory for quantum systems},
journal = {J. Multivar. Anal.},
volume = {3},
number = {4},
pages = {337-394},
year = {1973},
issn = {0047-259X},
doi = {https://doi.org/10.1016/0047-259X(73)90028-6},
url = {https://www.sciencedirect.com/science/article/pii/0047259X73900286},
author = {A.S Holevo}
}

@misc{hayashi2025generalizedquantumsteinslemma,
      title={Generalized Quantum {S}tein's Lemma and Second Law of Quantum Resource Theories}, 
      author={Masahito Hayashi and Hayata Yamasaki},
      year={2025},
      eprint={2408.02722},
      archivePrefix={arXiv}
}

@book{dembo1998large,
  author    = {Amir Dembo and Ofer Zeitouni},
  title     = {Large Deviations Techniques and Applications},
  series    = {Applications of Mathematics},
  volume    = {38},
  edition   = {2},
  publisher = {Springer},
  address   = {New York},
  year      = {1998},
  isbn      = {978-0-387-98406-7},
  doi       = {10.1007/978-1-4612-5320-4},
  url       = {https://doi.org/10.1007/978-1-4612-5320-4}
}

@article{PhysRevLett.95.110504,
  title = {Purity Estimation with Separable Measurements},
  author = {Bagan, E. and Ballester, M. A. and Mu\~noz-Tapia, R. and Romero-Isart, O.},
  journal = {Phys. Rev. Lett.},
  volume = {95},
  issue = {11},
  pages = {110504},
  numpages = {4},
  year = {2005},
  month = {Sep},
  publisher = {American Physical Society},
  doi = {10.1103/PhysRevLett.95.110504},
  url = {https://link.aps.org/doi/10.1103/PhysRevLett.95.110504}
}

@article{Gasbarri_2024,
   title={Sequential hypothesis testing for continuously-monitored quantum systems},
   volume={8},
   ISSN={2521-327X},
   url={http://dx.doi.org/10.22331/q-2024-03-20-1289},
   DOI={10.22331/q-2024-03-20-1289},
   journal={Quantum},
   publisher={Verein zur Forderung des Open Access Publizierens in den Quantenwissenschaften},
   author={Gasbarri, Giulio and Bilkis, Matias and Roda-Salichs, Elisabet and Calsamiglia, John},
   year={2024},
   month=Mar, pages={1289} }

@misc{simpson2026optimalerrorexponentscomposite,
      title={Optimal Error Exponents for Composite Sequential Quantum Hypothesis Testing}, 
      author={Jacob Paul Simpson and Efstratios Palias and Sharu Theresa Jose},
      year={2026},
      eprint={2605.04915},
      archivePrefix={arXiv}
}

@article{Shamir2009,
  author  = {Shamir, Reuben R. and Joskowicz, Leo and Spektor, Sergey and Shoshan, Yigal},
  title   = {Localization and registration accuracy in image guided neurosurgery: a clinical study},
  journal = {International Journal of Computer Assisted Radiology and Surgery},
  volume  = {4},
  pages   = {45--52},
  year    = {2009},
  doi     = {10.1007/s11548-008-0268-8}
}

@article{Lin2023,
  author  = {Lin, Zhefan and Lei, Chen and Yang, Liangjing},
  title   = {Modern Image-Guided Surgery: A Narrative Review of Medical Image Processing and Visualization},
  journal = {Sensors},
  volume  = {23},
  number  = {24},
  pages   = {9872},
  year    = {2023},
  doi     = {10.3390/s23249872}
}

@article{Tsang_2012,
   title={Continuous Quantum Hypothesis Testing},
   volume={108},
   ISSN={1079-7114},
   url={http://dx.doi.org/10.1103/PhysRevLett.108.170502},
   DOI={10.1103/physrevlett.108.170502},
   number={17},
   journal={Physical Review Letters},
   publisher={American Physical Society (APS)},
   author={Tsang, Mankei},
   year={2012},
   month=Apr }

@article{helstrom1969quantum,
  title={Quantum detection and estimation theory},
  author={Helstrom, Carl W},
  journal={Journal of statistical physics},
  volume={1},
  number={2},
  pages={231--252},
  year={1969},
  publisher={Springer}
}

@article{yuen1975optimum,
  title     = {Optimum testing of multiple hypotheses in quantum detection theory},
  author    = {Yuen, Horace and Kennedy, Robert and Lax, Melvin},
  journal   = {IEEE Transactions on Information Theory},
  volume    = {21},
  number    = {2},
  pages     = {125--134},
  year      = {1975},
  publisher = {IEEE},
  doi       = {10.1109/TIT.1975.1055351},
  url       = {https://doi.org/10.1109/TIT.1975.1055351}
}

@article{Bavaresco2021,
  title = {Strict Hierarchy between Parallel, Sequential, and Indefinite-Causal-Order Strategies for Channel Discrimination},
  author = {Bavaresco, Jessica and Murao, Mio and Quintino, Marco T\'ulio},
  journal = {Phys. Rev. Lett.},
  volume = {127},
  issue = {20},
  pages = {200504},
  numpages = {7},
  year = {2021},
  month = {Nov},
  publisher = {American Physical Society},
  doi = {10.1103/PhysRevLett.127.200504},
  url = {https://link.aps.org/doi/10.1103/PhysRevLett.127.200504}
}

@article{Bavaresco_2022,
   title={Unitary channel discrimination beyond group structures: Advantages of sequential and indefinite-causal-order strategies},
   volume={63},
   ISSN={1089-7658},
   url={http://dx.doi.org/10.1063/5.0075919},
   number={4},
   journal={Journal of Mathematical Physics},
   publisher={AIP Publishing},
   author={Bavaresco, Jessica and Murao, Mio and Quintino, Marco Túlio},
   year={2022},
   month=Apr }

@article{Giovannetti2006,
  title = {Quantum Metrology},
  author = {Giovannetti, Vittorio and Lloyd, Seth and Maccone, Lorenzo},
  journal = {Phys. Rev. Lett.},
  volume = {96},
  issue = {1},
  pages = {010401},
  numpages = {4},
  year = {2006},
  month = {Jan},
  publisher = {American Physical Society},
  doi = {10.1103/PhysRevLett.96.010401},
  url = {https://link.aps.org/doi/10.1103/PhysRevLett.96.010401}
}

@article{Buhrman2001,
  title = {Quantum Fingerprinting},
  author = {Buhrman, Harry and Cleve, Richard and Watrous, John and de Wolf, Ronald},
  journal = {Phys. Rev. Lett.},
  volume = {87},
  issue = {16},
  pages = {167902},
  numpages = {4},
  year = {2001},
  month = {Sep},
  publisher = {American Physical Society},
  doi = {10.1103/PhysRevLett.87.167902},
  url = {https://link.aps.org/doi/10.1103/PhysRevLett.87.167902}
}

@article{Ekert2002,
  title = {Direct Estimations of Linear and Nonlinear Functionals of a Quantum State},
  author = {Ekert, Artur K. and Alves, Carolina Moura and Oi, Daniel K. L. and Horodecki, Micha\l{} and Horodecki, Pawe\l{} and Kwek, L. C.},
  journal = {Phys. Rev. Lett.},
  volume = {88},
  issue = {21},
  pages = {217901},
  numpages = {4},
  year = {2002},
  month = {May},
  publisher = {American Physical Society},
  doi = {10.1103/PhysRevLett.88.217901},
  url = {https://link.aps.org/doi/10.1103/PhysRevLett.88.217901}
}

@article{Lyons2018,
author = {Ashley Lyons  and George C. Knee  and Eliot Bolduc  and Thomas Roger  and Jonathan Leach  and Erik M. Gauger  and Daniele Faccio },
title = {Attosecond-resolution {H}ong-{O}u-{M}andel interferometry},
journal = {Science Advances},
volume = {4},
number = {5},
pages = {eaap9416},
year = {2018},
doi = {10.1126/sciadv.aap9416},
URL = {https://www.science.org/doi/abs/10.1126/sciadv.aap9416},
eprint = {https://www.science.org/doi/pdf/10.1126/sciadv.aap9416},
}

@article{Aguilar2020,
  title = {Robust Interferometric Sensing Using Two-Photon Interference},
  author = {Aguilar, G. H. and Piera, R. S. and Saldanha, P. L. and Filho, R. L. de Matos and Walborn, S. P.},
  journal = {Phys. Rev. Appl.},
  volume = {14},
  issue = {2},
  pages = {024028},
  numpages = {10},
  year = {2020},
  month = {Aug},
  publisher = {American Physical Society},
  doi = {10.1103/PhysRevApplied.14.024028},
  url = {https://link.aps.org/doi/10.1103/PhysRevApplied.14.024028}
}

@article{Parniak2018,
  title = {Beating the {R}ayleigh Limit Using Two-Photon Interference},
  author = {Parniak, Micha\l{} and Bor\'owka, Sebastian and Boroszko, Kajetan and Wasilewski, Wojciech and Banaszek, Konrad and Demkowicz-Dobrza\ifmmode \acute{n}\else \'{n}\fi{}ski, Rafa\l{}},
  journal = {Phys. Rev. Lett.},
  volume = {121},
  issue = {25},
  pages = {250503},
  numpages = {6},
  year = {2018},
  month = {Dec},
  publisher = {American Physical Society},
  doi = {10.1103/PhysRevLett.121.250503},
  url = {https://link.aps.org/doi/10.1103/PhysRevLett.121.250503}
}

@article{Ndagano2022,
  title = {Quantum Microscopy Based on {{Hong}}--{{Ou}}--{{Mandel}} Interference},
  author = {Ndagano, Bienvenu and Defienne, Hugo and Branford, Dominic and Shah, Yash D. and Lyons, Ashley and Westerberg, Niclas and Gauger, Erik M. and Faccio, Daniele},
  year = 2022,
  month = may,
  journal = {Nat. Photon.},
  volume = {16},
  number = {5},
  pages = {384--389},
  publisher = {Nature Publishing Group},
  issn = {1749-4893},
  doi = {10.1038/s41566-022-00980-6}
}

@article{Ansari2021,
  title = {Achieving the Ultimate Quantum Timing Resolution},
  author = {Ansari, Vahid and Brecht, Benjamin and Gil-Lopez, Jano and Donohue, John M. and \ifmmode \check{R}\else \v{R}\fi{}eh\'a\ifmmode \check{c}\else \v{c}\fi{}ek, Jaroslav and Hradil, Zden\ifmmode \check{e}\else \v{e}\fi{}k and S\'anchez-Soto, Luis L. and Silberhorn, Christine},
  journal = {PRX Quantum},
  volume = {2},
  issue = {1},
  pages = {010301},
  numpages = {7},
  year = {2021},
  month = {Jan},
  publisher = {American Physical Society},
  doi = {10.1103/PRXQuantum.2.010301},
  url = {https://link.aps.org/doi/10.1103/PRXQuantum.2.010301}
}

@article{Mazelanik2022,
  title = {Optical-Domain Spectral Super-Resolution via a Quantum-Memory-Based Time-Frequency Processor},
  author = {Mazelanik, Mateusz and Leszczy{\'n}ski, Adam and Parniak, Micha{\l}},
  year = 2022,
  month = feb,
  journal = {Nat. Comm.},
  volume = {13},
  number = {1},
  pages = {691},
  issn = {2041-1723},
  doi = {10.1038/s41467-022-28066-5}
}

@article{HOM1987,
  title = {Measurement of subpicosecond time intervals between two photons by interference},
  author = {Hong, C. K. and Ou, Z. Y. and Mandel, L.},
  journal = {Phys. Rev. Lett.},
  volume = {59},
  issue = {18},
  pages = {2044--2046},
  numpages = {0},
  year = {1987},
  month = {Nov},
  publisher = {American Physical Society},
  doi = {10.1103/PhysRevLett.59.2044},
  url = {https://link.aps.org/doi/10.1103/PhysRevLett.59.2044}
}

@article{Barenco1997,
author = {Barenco, Adriano and Berthiaume, Andr\'{e} and Deutsch, David and Ekert, Artur and Jozsa, Richard and Macchiavello, Chiara},
title = {Stabilization of Quantum Computations by Symmetrization},
journal = {SIAM Journal on Computing},
volume = {26},
number = {5},
pages = {1541-1557},
year = {1997},
doi = {10.1137/S0097539796302452},
URL = {https://doi.org/10.1137/S0097539796302452},
eprint = {https://doi.org/10.1137/S0097539796302452}
}

@article{Guta2006,
  title = {Local asymptotic normality for qubit states},
  author={Gu{\c{t}}{\u{a}}, M{\u{a}}d{\u{a}}lin and Kahn, Jonas},
  journal = {Phys. Rev. A},
  volume = {73},
  issue = {5},
  pages = {052108},
  numpages = {15},
  year = {2006},
  month = {May},
  publisher = {American Physical Society},
  doi = {10.1103/PhysRevA.73.052108},
  url = {https://link.aps.org/doi/10.1103/PhysRevA.73.052108}
}

@misc{ISOIECGuide984,
  author       = {{International Organization for Standardization}},
  title        = {{ISO/IEC Guide 98-4:2012. Uncertainty of measurement -- Part 4: Role of measurement uncertainty in conformity assessment}},
  year         = {2012},
  publisher    = {International Organization for Standardization},
  url          = {https://www.iso.org/standard/50465.html},
  note         = {Accessed: 2026-07-09}
}

\onecolumngrid
\appendix

\renewcommand{\thesection}{\Alph{section}}
\renewcommand{\thesubsection}{\Alph{section}.\arabic{subsection}}
\renewcommand{\thesubsubsection}{\Alph{section}.\arabic{subsection}.\Roman{subsubsection}}


\section{Description of Tests}\
Let $\Theta\subset\mathbb{R}$ be a parameter space.  For any $\delta>0$ define the sets 
	\begin{equation}
       	 \begin{split}
            		B_\delta(\vartheta) &:= \left\{\varphi \in \Theta \;\middle|\,\; |\,\varphi-\vartheta|\, \le \delta\right\}\\
            		B^{c}_\delta(\vartheta) &:= \Theta \backslash B_\delta(\vartheta)\, .
        	\end{split}
    	\label{app:complement_sets}
	\end{equation}
For every $\vartheta\in\Theta$ we define the composite hypothesis that the parameter lies within a $\delta$-neighborhood of $\vartheta$,   $H_\vartheta\, :\, \varphi \in B_\delta(\vartheta)$.  
Let $\Vec{m}\in M^{\mathbb{N}}$ denote the sequence of observations; note that in principle this sequence may be countably infinite.  Then for $n\in\mathbb{N}$, 
$M^{n}$ is a random variable distributed according to 
	\begin{equation}
        	p(\Vec{m}_n\mid\vartheta)=\prod_{k=1}^n p(m_k\mid\Vec{m}_{k-1},\vartheta)
    	\label{app:prob_observations}
    	\end{equation}
where $\Vec{m}_n=(m_1,\ldots , m_n)$ denotes the vector of $n$ observations.  Note that we do not assume that the observations are 
independent nor that they are identically distributed.  

As the hypotheses are not mutually exclusive, it makes little sense to compare one hypothesis against the union of all other hypotheses. 
Instead, we compare each hypothesis against its corresponding complementary hypothesis, $H_{\bar{\vartheta}}\, :\, \theta \in B^{c}_\delta(\vartheta)$.  Specifically for each 
$\vartheta\in\Theta$, we associate Sequential Probability Ratio Tests (SPRTs) comparing $H_\vartheta$, with $H_{\bar{\vartheta}}$, and stop as soon as one of these tests reaches the stopping boundary.
Because our hypotheses are composite, the construction of the likelihood function is not unique. Below we provide three operationally motivated likelihood functions 
that give rise to three sequential tests.  \secref{app:complement_test} introduces the \textit{complement test}, \secref{app:concentration_test} the 
\textit{concentration test}, and \secref{app:twin-peaks} the \textit{twin-peaks} test.  Finally, \secref{app:Asymptotic_equivalence} proves that under suitable regularity conditions, all three tests 
are asymptotically equivalent at the exponential scale, implying identical error exponents and stopping-time scaling.

\subsection{Complement test}\label{app:complement_test}
For every $\vartheta\in\Theta$, the \textit{complement test} considers the two complementary composite hypotheses
\begin{equation}
    H_\vartheta:
    \varphi\in B_\delta(\vartheta),
    \qquad
    H_{\bar{\vartheta}}:
    \varphi\notin B_\delta(\vartheta).
\end{equation}
Let $\pi(\varphi)\,\mathrm{d}\varphi$ be a prior probability measure on
$\Theta$, and denote the prior probabilities of these hypotheses by
\begin{equation}
    P(H_\vartheta)
    =
    \int_{B_\delta(\vartheta)}
    \pi(\varphi)\,\mathrm{d}\varphi,
    \qquad
    P(H_{\bar{\vartheta}})
    =
    1-P(H_\vartheta).
\end{equation}
The marginal likelihood under $H_\vartheta$ is
\begin{equation}
    p(\Vec{m}_n\mid H_\vartheta)
    =
    \frac{
        \displaystyle
        \int_{B_\delta(\vartheta)}
        p(\Vec{m}_n\mid\varphi)\,
        \pi(\varphi)\,
        \mathrm{d}\varphi
    }{
        P(H_\vartheta)
    },
\end{equation}
and analogously for $H_{\bar{\vartheta}}$.

For each $\vartheta$, the likelihood ratio is therefore
\begin{equation}
    \Lambda_n^{(\mathrm{comp})}(\vartheta)
    :=
    \frac{
        p(\Vec{m}_n\mid H_\vartheta)
    }{
        p(\Vec{m}_n\mid H_{\bar{\vartheta}})
    }.
\end{equation}
By Bayes' rule, the Bayesian strong-error condition
\begin{equation}
    P(H_\vartheta\mid\Vec{m}_n)
    \geq
    1-\epsilon
\end{equation}
is equivalent to
\begin{equation}
    \Lambda_n^{(\mathrm{comp})}(\vartheta)
    \frac{
        P(H_\vartheta)
    }{
        P(H_{\bar{\vartheta}})
    }
    \geq
    \frac{1-\epsilon}{\epsilon}.
\label{app:complement_individual_condition}
\end{equation}
All hypotheses $H_\vartheta$ are evaluated in parallel using the same
measurement record. The global procedure is one-sided: it stops only when
Eq.~\eqref{app:complement_individual_condition} is satisfied in favor of at
least one $H_\vartheta$. No acceptance condition is imposed on
$H_{\bar{\vartheta}}$; evidence against a particular $H_\vartheta$ only
disfavors the corresponding region $B_\delta(\vartheta)$ and cannot cause the
global test to stop.

The stopping time is consequently
\begin{equation}
    N
    = \inf\left\{  n:
        \max_{\vartheta\in\Theta}
        \left[
            \Lambda_n^{(\mathrm{comp})}(\vartheta)
            \frac{P(H_\vartheta) }{P(H_{\bar{\vartheta}})}
        \right]
        \geq
        \frac{1-\epsilon}{\epsilon}
    \right\}.
\end{equation}
The test stops as soon as the strong-error condition is satisfied for at least one $H_\vartheta$. 
The procedure can  be viewed as an infinite collection of correlated
parallel random walkers, each with a single absorbing boundary, together with
a global ``first one past the post'' stopping rule.
The hypothesis attaining the maximum is accepted, thereby identifying a tolerance region $B_\delta(\hat{\vartheta})$ 
with posterior probability at least $1-\epsilon$; its center $\hat{\vartheta}$ may additionally be reported as a point estimate.

Equivalently, using the joint probabilities of the data and the composite hypotheses, the individual stopping condition can be written more compactly as
 \begin{equation} 
 \frac{ P(\Vec{m}_n,H_\vartheta) }{ P(\Vec{m}_n,H_{\bar{\vartheta}}) } 
 \geq \frac{1-\epsilon}{\epsilon}. 
 \end{equation} 
 Here, the dependence on the prior is implicit in the joint probabilities; in particular,
  \begin{equation}
   P(\Vec{m}_n,H_\vartheta) = \int_{B_\delta(\vartheta)} p(\Vec{m}_n\mid\varphi)\, \pi(\varphi)\, \mathrm{d}\varphi .
  \end{equation}
In practice, evaluating this condition can be computationally demanding, because the
 prior-weighted likelihood must be integrated over a different parameter region for every candidate center $\vartheta$
  and recomputed after each measurement round.

\subsection{Concentration test}\label{app:concentration_test}

Alternatively, one may compare the posterior-weighted likelihood at the most likely parameter value against the average posterior-weighted likelihood over all parameter values that are distinguishable 
at a prescribed resolution $\delta$. More precisely, define
	\begin{equation}
    		\hat{\vartheta}:=\arg\max_{\vartheta\in\Theta}p(\Vec m_n\mid\vartheta)\pi(\vartheta)\, 
	\label{app:argmax_concentration}
	\end{equation}
and let $B_\delta(\hat{\vartheta})$ denote its $\delta$-neighborhood.  The posterior-weighted average likelihood outside this neighborhood is 
	\begin{equation}
		\bar{L}_n :=\frac{1}{\pi(\Theta\backslash B_{\delta}(\hat{\vartheta}))} \int_{\Theta\backslash B_{\delta}(\hat{\vartheta})} p(\Vec{m}_n\mid\varphi)\pi(\varphi)\mathrm{d}\varphi\, ,
	\label{app:average_likelihood}
	\end{equation}
for some measure $\pi:\Theta\to \mathbb{R}$. The \textit{concentration test} is defined using the following likelihood ratio statistic
	\begin{equation}
		\Lambda_n^{(\mathrm{CT})}(\hat{\vartheta})= \frac{p(\Vec{m}_n\mid\hat{\vartheta})}{\bar{L}_n}\,,
	\label{app:concentration_likelihood}
	\end{equation}
with the stopping condition
	\begin{equation}
		\Lambda_n^{(\mathrm{CT})}(\hat{\vartheta})\ge A_\epsilon \, .
	\label{app:concentration_stopping}
	\end{equation}	
Here, $A_\epsilon$ is a user-defined threshold that depends on the desired error threshold $\epsilon$ used in the complement test above. 	

Just as for the \textit{complement test}, the concentration test has a clear operational interpretation within the Bayesian framework.  Applying Bayes' theorem \eqnref{app:concentration_likelihood} 
can be rewritten as 
	\begin{equation}
		\frac{p(\hat{\vartheta}\mid\Vec{m}_n)\pi(\Theta\backslash B_{\delta}(\hat{\vartheta}))}{p(\Theta\backslash B_\delta(\hat{\vartheta})\mid \Vec{m}_n)}=\Lambda_n^{(\mathrm{CT})}(\hat{\vartheta})\pi(\hat{\vartheta}) \,.
	\label{app:concentration_posterior}
	\end{equation}
Hence, the statistic of the concentration test measures how strongly the posterior probability density function concentrates within a single resolution cell. For a uniform measure, i.e., 
$\pi(\vartheta)=\lvert\Theta\rvert^{-1}$, the test measures how much larger the posterior probability density function at $\hat{\vartheta}$ is than the average posterior density. 
The test terminates once the posterior density at the maximum-likelihood estimate exceeds the average posterior density over all distinguishable competing parameter values by a factor $A_\epsilon$.
	
\eqnref{app:concentration_posterior} may appear ill-defined as point probabilities vanish for continuous parameter spaces, i.e., $\pi(\{\hat{\vartheta}\})=0$.  However, this issue can be readily remedied by discretizing the 
parameter space and taking the appropriate continuous limit. Specifically, for $k\in\mathbb{N}$ let us partition the parameter space $\Theta$ into compact, convex disjoint sets 
$\{\Theta_j\, |\, \, \pi(\Theta_j)= \pi(\Theta)/k\}_{j=1}^k$.  Let $j^\star$ denote the cell containing the maximizer, $\hat{\vartheta}$, and define the collection of cells intersecting its $\delta$-neighborhood by 
	\begin{equation}
		\mathcal N_\delta(j^\star)=\left\{i:\Theta_i\cap B_\delta(\hat{\vartheta})\neq\varnothing\right\}\, .
	\end{equation}
The posterior-weighted likelihood of each cell is
	\begin{equation}
		p(\Vec{m}_n\mid j):=\int_{\Theta_j} p(\Vec{m}_n\mid\vartheta)\pi(\vartheta)\mathrm{d}\vartheta\, ,
	\label{app:concentration_discreet}
	\end{equation}
and the corresponding likelihood ratio as 
	\begin{equation}
		\Lambda_n^{(\mathrm{CT})}(j^\star):= \frac{p(\Vec{m}_n\mid j^\star)}{\displaystyle\frac1{k-|\,\mathcal N_\delta(j^\star)|\,}\sum_{i\notin\mathcal N_\delta(j^\star)}p(\Vec{m}_n\mid i)}
	\label{app:concentration_discreet_likelihood}
	\end{equation} 	
As each $\Theta_j$ is compact and $p(\Vec{m}_n\mid \vartheta)\pi(\vartheta)$ is continuous in $\vartheta$ the integral in \eqnref{app:concentration_discreet} can be approximated by 
	\begin{equation}
		p(\Vec{m}_n\mid j)=p(\Vec{m}_n\mid \vartheta_j)\,\pi(\Theta_j)\;+\; o(\pi(\Theta_j))\, ,
	\label{app:integral_mean_value}
 	\end{equation}
for some $\vartheta_j\in\Theta_j$ where the error vanishes in the limit of a fine partition. As the partition is refined, the union of cells in $\mathcal N_\delta(j^\star)$
converges to $B_\delta(\hat{\vartheta})$ in measure. Consequently, the discrete likelihood ratio converges, in the Riemann sense, to the continuous statistic \eqnref{app:concentration_likelihood}.

\subsection{The twin-peaks tests}\label{app:twin-peaks}

A third choice of test is the \textit{twin-peaks test}, which compares the globally most likely parameter value, $\hat{\vartheta}\in\Theta$, with the most likely parameter value outside a prescribed 
$\delta$-neighborhood of $\hat{\vartheta}$.  More concretely, we define 
	\begin{equation}
		\hat{\vartheta}:=\arg\max\limits_{\vartheta\in\Theta} \bigl\{p(\Vec{m}_n\mid\vartheta)\pi(\vartheta)\bigr\}\, ,
	\label{app:argmax}
	\end{equation} 
and define the $\delta$-sized neighborhood, $B_\delta(\hat{\vartheta})$, around  $\hat{\vartheta}$.  The statistic of the twin-peaks test is defined as 
	\begin{equation}       \Lambda^{\mathrm{TP}}_n(\hat{\vartheta}):=\frac{p(\Vec{m}_n\mid\hat{\vartheta})\,\pi(\hat{\vartheta})}{\sup\limits_{\vartheta\in \Theta\backslash B_\delta(\hat{\vartheta})}\, p(\Vec{m}_n\mid\vartheta)\,
            \pi(\vartheta)}\, ,
    \label{app:twin_peaks_ratio}
    \end{equation}
with the stopping condition 
	\begin{equation}
		\Lambda^{(\mathrm{TP})}_n(\hat{\vartheta})\ge A_\epsilon\, 
	\label{app:twin-peaks_stopping}
	\end{equation}
for some user-defined threshold $A_\epsilon$ that depends on the desired error-threshold $\epsilon$.

Operationally, the twin-peaks test determines whether the maximum-likelihood estimate is statistically isolated from all competing parameter values outside a prescribed resolution $\delta$. The neighborhood 
$B_\delta(\hat{\vartheta})$ specifies the desired estimation resolution, so that all parameter values within this neighborhood are regarded as operationally equivalent. The test terminates once every parameter value outside this neighborhood has a substantially smaller posterior-weighted likelihood, as quantified by the threshold $A_\epsilon$.  In contrast to the \textit{concentration test}, the 
\textit{twin-peaks test} measures whether the dominant likelihood peak is isolated in the likelihood landscape.  

\subsection{Asymptotic equivalence between all three tests}\label{app:Asymptotic_equivalence}

The \textit{complement test}, \textit{concentration test}, and \textit{twin-peaks test} have distinct test statistics as well as stopping rules.  In this section, we 
show that all three tests are asymptotically equivalent at the exponential scale, meaning that they possess the same error exponent and stopping-time scaling.
To establish this equivalence, write
    \begin{equation}
        p(\vec{m}_n\mid\varphi)=e^{n \ell_n(\varphi)}
    \end{equation}
where we have defined 
    \begin{equation}
        \ell_n(\varphi):=\frac{1}{n}\log p(\vec{m}_n\mid\varphi)\, .
    \end{equation}
        
The key mathematical ingredient is the following version of Laplace's principle~\cite[Theorem 4.3.1]{dembo1998large}.  
\begin{theorem}[Laplace principle]\label{thm:Laplace}
Let $A\subset\Theta$ be a measurable set satisfying $0<\pi(A)<\infty$, and define
	\begin{equation}
    		I_n(A):=\int_A e^{n\ell_n(\varphi)}\pi(\varphi)\,\mathrm{d}\varphi\, .
	\end{equation}
Assume that 
	\begin{equation}
    		\ell_n\xrightarrow[n\to\infty]{\mathrm{a.s.,\,uniformly\ on\ compact\ sets}}\ell,
	\label{app:almost_uniform_convergence}
	\end{equation}
and each $\ell_n$ is uniformly continuous on compact subsets of $\Theta$. Then
	\begin{equation}
    		\frac1n\log I_n(A)=\sup_{\varphi\in A}\ell(\varphi)+o(1),
	\label{app:Laplace_identity}
	\end{equation}
almost surely.
\end{theorem}

Note that the assumption of almost uniform convergence, \eqnref{app:almost_uniform_convergence}, is satisfied by i.i.d. models, ergodic Markov processes, as well as 
adaptive experiments with stabilized likelihood increments.

\begin{proof}
The proof proceeds by establishing that both the limit superior and limit inferior of $\frac1n\log I_n(A)$ converge.  For the limit superior, observe that for any measurable set $A\subset\Theta$, 
we have that point-wise $\ell_n(\varphi)\le \sup\limits_{\omega \in A} \ell_n(\omega), \, \forall\varphi\in A$ and 
    \begin{equation}
        I_n(A):=\int_A e^{n \ell_n(\varphi)} \pi(\varphi)\mathrm{d}\varphi\le \exp\left(n \sup\limits_{\omega \in A} \ell_n(\omega)\right) \pi(A)\, .
    \end{equation}
It follows that 
    \begin{equation}
        \frac{1}{n} \log I_n(A)\leq \sup\limits_{\omega \in A} \ell_n(\omega) +\frac{1}{n} \log \pi(A)\, ,
    \end{equation}
and since $0  < \pi(A)<\infty$ 
    \begin{equation}
            \limsup\limits_{n\to\infty} \frac{1}{n} \log I_n(A)\le \limsup\limits_{n\to\infty}\,  \sup\limits_{\omega\in A} \ell_n(\omega)=\sup\limits_{\omega\in A} 
            \ell(\omega)\, .
    \end{equation}

To establish a lower bound, note that for every $\epsilon>0$ there exists $\varphi\in A$ such that 
    \begin{equation}
        \ell_n(\varphi)\ge \sup\limits_{\omega\in A} \ell_n(\omega)-\epsilon\, .
    \end{equation}
By uniform continuity of $\ell_n(\varphi)$ on compact sets of $\Theta$, there exists a compact neighborhood $B_\epsilon\subset A$ around $\varphi$ such that
$\ell_n(\omega)\ge \ell_n(\varphi)-\epsilon, \, \forall\, \omega\in B_\epsilon$. Hence 
    \begin{equation}
        \inf\limits_{\omega\in B_\epsilon}\ell_n(\omega)\geq \ell_n(\varphi)-\epsilon\geq \sup\limits_{\omega\in A} \ell_n(\omega)-2\epsilon\, ,
    \end{equation}
and 
    \begin{equation}
        \frac{1}{n}\log I_n(A)\ge \sup\limits_{\omega\in A} \ell_n(\omega)-2\epsilon + \frac{1}{n}\log \pi(B_\epsilon)\,.
    \end{equation}
As $0<\pi(B_\epsilon)<\infty$ it follows that
    \begin{equation}
        \liminf_{n\to\infty} \frac{1}{n}\log I_n(A)\geq \liminf_{n\to\infty}\, \sup\limits_{\omega\in A} \ell_n(\omega)-2\epsilon=\sup\limits_{\omega\in A} \ell(\omega)\, ,  
    \end{equation}
where we have let $\epsilon\to 0$. As both the limit superior and limit inferior of $\log I_n(A)/n$ coincide it follows that 
    \begin{equation}
        \frac{1}{n} \log I_n(A) = \sup\limits_{\omega\in A} \ell(\omega) + o(1)\, .
    \end{equation}
\end{proof}

The exponential equivalence among the three tests follows almost immediately.

\begin{corollary}[Exponential equivalence of the likelihood ratios]
Under the assumptions of Theorem~\ref{thm:Laplace}, the likelihood ratios of the complement, concentration, and twin-peaks tests satisfy
	\begin{equation}
		\begin{split}
			\frac1n\log\Lambda_n^{(\mathrm{IL})}(\vartheta)&=\sup\limits_{\varphi\in B_\delta(\vartheta)}\ell(\varphi)-\sup_{\varphi\in B_\delta^c(\vartheta)}\ell(\varphi)+o(1)\\
    			\frac1n\log\Lambda_n^{(\mathrm{CT})}(\hat{\vartheta})&=\ell(\hat{\vartheta})-\sup_{\varphi\in\Theta\backslash B_\delta(\hat\vartheta)}\ell(\varphi)+o(1),\\
    			\frac1n\log\Lambda_n^{(\mathrm{TP})}(\hat{\vartheta})&=\ell(\hat\vartheta)-\sup_{\varphi\in\Theta\backslash B_\delta(\hat\vartheta)}\ell(\varphi)+o(1).
		\end{split}
	\end{equation}
Consequently, the concentration and twin-peaks tests are exponentially equivalent.  Moreover, evaluating the complement test at the maximizing parameter value $\hat{\vartheta}$ yields the same 
exponential rate.
\end{corollary}

\begin{proof}
For the complement test, apply Theorem~\ref{thm:Laplace} separately to the numerator and denominator of
	\begin{equation}
		L_n^{(\mathrm{IL})}(H_\vartheta)=\int_{B_\delta(\vartheta)}e^{n\ell_n(\varphi)}\pi(\varphi)\,d\varphi,
	\end{equation}
and subtract the resulting asymptotic expansions.

For the \textit{concentration test}, the denominator is the integrated likelihood over $\Theta\backslash B_\delta(\hat\vartheta)$, while the numerator consists of the point likelihood evaluated at the maximizer. Hence
	\begin{equation}
		\frac1n\log\Lambda_n^{(\mathrm{CT})}=\ell(\hat\vartheta)-\sup_{\varphi\in\Theta\backslash B_\delta(\hat\vartheta)}\ell(\varphi)+o(1).
	\end{equation}
For the \textit{twin-peaks test}, $\log\pi(\hat\vartheta)=o(n)$, and therefore the prior contributes only a subexponential correction,
	\begin{equation}
		\frac1n\log\bigl(p(\Vec m_n\mid \hat\vartheta)\pi(\hat\vartheta)\bigr)=\ell(\hat\vartheta)+o(1),
	\end{equation}
while the denominator is treated identically. This yields the same exponential limit.
\end{proof}

Since the \textit{concentration} and \textit{twin-peaks tests} possess the same exponential rate as the complement test, they inherit the same asymptotic stopping 
threshold. Consequently, throughout this work, we set $A_\epsilon= \log(1-\epsilon)-\log\epsilon$ with the understanding that, asymptotically, this threshold enforces 
the same strong-error criterion as the \textit{complement test}.  We now turn to the corresponding stopping-time asymptotics. Assuming that the observations are 
generated by a true parameter $\vartheta^*\in\Theta$, and that the maximum-likelihood estimator is strongly consistent, i.e., $\hat{\vartheta}
\xrightarrow{\mathrm{a.s.}}\vartheta^*$, the likelihood ratios of the \textit{concentration} and \textit{twin-peaks tests} satisfy
	\begin{equation}
    		\lim_{n\to\infty}\frac{1}{n}\log\Lambda_n^{(\mathrm{CT})}=\lim_{n\to\infty}\frac{1}{n}\log\Lambda_n^{(\mathrm{TP})}=
    		\inf_{\varphi\in\Theta\setminus B_\delta(\vartheta^*)}D(\vartheta^*\|\,\varphi)\, ,
	\end{equation}
almost surely, where
	\begin{equation}
    		D(\vartheta^*\|\,\varphi):=\lim_{n\to\infty}\frac{1}{n}D_{\mathrm{KL}}\!\left(p(\Vec{m}_n\mid\vartheta^*)\,\middle\|\,p(\Vec{m}_n\mid\varphi)\right)
	\end{equation}
denotes the relative entropy rate. For the \textit{complement test} evaluated at $\vartheta=\vartheta^*$, and under the same consistency assumption, the identical 
asymptotic limit is obtained.  As the stopping criterion is $\log\Lambda_n=A_\epsilon$, it follows that the stopping time of each test satisfies
	\begin{equation}
    		\frac{N}{A_\epsilon}\xrightarrow{\mathrm{a.s.}}\frac{1}{\inf\limits_{\varphi\in\Theta\setminus B_\delta(\vartheta^*)}D(\vartheta^*\|\,\varphi)}
	\end{equation}
for $\epsilon\rightarrow0$. Under the additional first-passage and integrability conditions required to exchange the large-threshold limit with the expectation value, 
the mean stopping time satisfies
	\begin{equation}
    		\mathbb{E}(N\mid\vartheta^*)=\frac{\log(1-\epsilon)-\log\epsilon}{\inf\limits_{\varphi\in\Theta\setminus B_\delta(\vartheta^*)}
    		D(\vartheta^*\|\,\varphi)}+o\!\left(\log\epsilon^{-1}\right).
	\label{app:asymptotic_stopping_time}
	\end{equation}
\subsection{Fisher information as the local curvature of the KL divergence}\label{app:Fisher_bounds}

In this appendix, we look at the behavior of the asymptotic stopping time, \eqnref{eq:asymptotic_stopping_time_classic}, in the limit of indistinguishable hypotheses.
In this regime, the KL divergence between two competing hypotheses $\vartheta$ and $\vartheta+\delta$ admits a quadratic expansion whose leading term is determined by the Fisher information of the 
conditional probability distribution $p(\Vec{m}_n\mid\vartheta)$.  Equivalently, the Fisher information is the Riemannian metric induced by the KL divergence on the statistical 
manifold~\cite{AmariShun-ichi2000Moig}.

Expanding $p(m\mid \vartheta\pm \delta)$ to second order
    \begin{equation}
        p(\vec{m}_n\mid\vartheta\pm \delta)=p(\Vec{m}_n\mid\vartheta)\pm p'(\Vec{m}_n\mid\vartheta)\, \delta+ \frac{1}{2}p''(\Vec{m}_n\mid\vartheta)\,  \delta^2 +\mc{O}( \delta^3) \,,
    \label{app:series_expand_theta}
    \end{equation}
where $p'(\Vec{m}_n\mid\vartheta)= \frac{\mathrm{d}}{\mathrm{d}\vartheta} p(\Vec{m}_n\mid\vartheta)$. It follows that 
	\begin{equation}
    		\log\left(\frac{p(\Vec{m}_n\mid\vartheta)}{p(\Vec{m}_n\mid\vartheta+\delta)}\right)= - \frac{p'(\Vec{m}_n\mid\vartheta)}{p(\Vec{m}_n\mid\vartheta)}\, \delta-\frac{1}{2}\left(\frac{p''(\Vec{m}_n\mid\vartheta)}
		{p(\Vec{m}_n\mid\vartheta)}-\left(\frac{p'(\Vec{m}_n\mid\vartheta)}{p(\Vec{m}_n\mid\vartheta)}\right)^2\right)\,  \delta^2 +\mc{O}( \delta^3)\, 
	\end{equation}
and hence the KL divergence becomes
	\begin{align}\nonumber
    		D(p(\Vec{m}_n\mid\vartheta)\|\,p(\Vec{m}_n\mid\vartheta+\delta))&=\sum_{\Vec{m}_n}\, p(\Vec{m}_n\mid\vartheta)\,  \log\left(\frac{p(\Vec{m}_n\mid\vartheta)}{p(\Vec{m}_n\mid\vartheta+\delta)}\right)\\
    		&=-\delta\, \sum_{\Vec{m}_n} p'(\Vec{m}_n\mid\vartheta) -\frac{ \delta^2}{2}\sum_{\Vec{m}_n} p''(\Vec{m}_n\mid\vartheta) 
    		+\frac{ \delta^2}{2}\sum_{\Vec{m}_n}\frac{p'(\Vec{m}_n\mid\vartheta)^2}{p(\Vec{m}_n\mid\vartheta)}\, .
	\end{align}
As $\sum_{\Vec{m}_n} p(\Vec{m}_n\mid\vartheta)=1$ remains constant, it follows that $\sum_{\Vec{m}_n} p'(\Vec{m}_n\mid\vartheta)=\sum_{\Vec{m}_n} p''(m\mid \vartheta)=0$. Recognizing that
	\begin{equation}
    		I\left(p(\Vec{m}_n\mid\vartheta)\right)=\sum_{\Vec{m}_n}\frac{p'(\Vec{m}_n\mid\vartheta)^2}{p(\Vec{m}_n\mid\vartheta)}
	\end{equation}
is the Fisher information of the conditional probability distribution $p(\Vec{m}_n\mid\vartheta)$, we define the Fisher information rate  by
	\begin{equation}
		I(\vartheta):=\lim_{n\to\infty} \frac{1}{n} I\left(p(\Vec{m}_n\mid\vartheta)\right)\, .
	\end{equation}
Substituting the second order expansion $D(\vartheta\|\,\vartheta+\delta)=\frac{1}{2}I(\vartheta)\delta^2+o(\delta^2)$ into the asymptotic stopping-time expression, \eqnref{app:asymptotic_stopping_time}, gives
	\begin{equation}
    		\mathds{E}(N\mid\vartheta^*)\sim \frac{2(\log(1-\epsilon)-\log(\epsilon))}{ \delta^2 I(\vartheta^*)}\, .
	\end{equation}
\subsection{Bayesian strong-error stopping times.}\label{app:bayesian_stopping}

In this appendix, we derive the asymptotic stopping time associated with a Bayesian strong-error criterion. Rather than comparing likelihoods, the Bayesian formulation monitors the posterior probability 
assigned to a neighborhood of the current estimate of the parameter and terminates once this probability exceeds a prescribed error level. As in the frequentist setting, our objective is to determine 
the asymptotic dependence of the stopping time on the error $\epsilon$ and the desired resolution $\delta$.

Let $\pi(\vartheta)$ denote a prior density that is continuous and strictly positive in a neighborhood of the true parameter, and let $p(\vartheta\mid\Vec{m}_n)$ denote the corresponding posterior distribution 
after $n$ observations.  The Bayesian strong-error stopping time associated with accuracy $\delta$ and error tolerance $\epsilon$ is
	\begin{equation}
		\tau_{\rm B}=\inf\left\{n: P(\vartheta\in B_\delta(\hat\vartheta_n)\mid \Vec{m}_n)\ge1-\epsilon\right\}\, ,
	\label{app:eq:bayesian_stopping}
	\end{equation}
where $B_\delta(\hat{\vartheta}_n)$ denotes the $\delta$-neighborhood of the current maximum-a-posteriori estimate $\hat{\vartheta}_n$.  The stopping criterion of \eqnref{app:eq:bayesian_stopping} is the Bayesian analog of the strong-error criterion considered throughout this work.

Under the standard regularity conditions of Bayesian asymptotic theory, the posterior distribution satisfies the Bernstein--von Mises theorem and is therefore asymptotically Gaussian,
	\begin{equation}
		p(\vartheta\mid\Vec{m}_n)\approx\mathcal{N}\left(\hat{\vartheta}_n,\frac{1}{nI(\vartheta)}\right),
	\label{app:Bernstein_vonMises}
	\end{equation}
where $I(\vartheta)$ denotes the Fisher information rate. Consequently, the posterior probability contained within the $\delta$-neighborhood of the estimator is
	\begin{equation}
		P\!\left(\vartheta\in B_\delta(\hat{\vartheta}_n)\mid\Vec{m}_n\right)\approx\operatorname{erf}\!\left(\delta\sqrt{\frac{nI(\vartheta)}{2}}\right).
	\label{app:posterior_erf}
	\end{equation}
Let $\zeta_{1-\epsilon/2}$ denote the $(1-\epsilon/2)$ quantile of the standard normal distribution,
	\begin{equation}
		\Phi(\zeta_{1-\epsilon/2})=1-\frac{\epsilon}{2},
	\label{app:epsilon_quantile1}
	\end{equation}
or equivalently,
	\begin{equation}
		\zeta_{1-\epsilon/2}=\sqrt{2}\,\operatorname{erf}^{-1}(1-\epsilon).
	\label{app:epsilon_quantile2}
	\end{equation}
Substituting \eqnref{app:posterior_erf} into the stopping criterion \eqnref{app:eq:bayesian_stopping} gives
	\begin{equation}
		\tau_{\mathrm{B}}(\vartheta)\sim\frac{\zeta_{1-\epsilon/2}^{\,2}}{\delta^2I(\vartheta)}.
	\label{app:tau_bayes}
\end{equation}
This condition is the natural continuous-parameter analog of a strong error condition. It is important that the condition is imposed on the posterior
probability of a set, not on the value of the posterior density at a point.

Using the asymptotic expansion of the Gaussian quantile,
	\begin{equation}
		\zeta_{1-\epsilon/2}^{\,2}\sim2\log\frac{1}{\epsilon},\qquad\epsilon\rightarrow0\, ,
	\end{equation}
the stopping time becomes
	\begin{equation}
		\tau_{\mathrm{B}}(\vartheta)\sim\frac{-2\log\epsilon}{\delta^2I(\vartheta)}\, .
	\label{app:tau_bayes_asymptotic}
	\end{equation}
Taking the expectation with respect to the prior distribution therefore yields
	\begin{equation}
		\mathbb{E}(\tau_{\mathrm{B}}\mid\pi): = \int \mathbb{E}(\tau_{\mathrm{B}}\mid\vartheta)\pi(\vartheta)\, \mathrm{d}\vartheta \sim\frac{-2\log(\epsilon)}{\delta^2}\int_{\Theta}\frac{\pi(\vartheta)}{I(\vartheta)}\,
		\mathrm{d}\vartheta.
	\label{app:bayes_avg_tau}
	\end{equation}

To quantify the advantage of Bayesian stopping, compare \eqnref{app:bayes_avg_tau} with the minimax fixed-sample benchmark
	\begin{equation}
		N^*(\delta,\epsilon)\sim\frac{2\log(1/\epsilon)}{\delta^2\inf_{\vartheta\in\Theta}I(\vartheta)},
	\label{app:minimax}
	\end{equation}
which is designed to guarantee the prescribed accuracy uniformly over the entire parameter space. Their ratio is given by
	\begin{equation}
		\frac{\mathbb{E}(\tau_{\mathrm{B}}\mid\pi)}{N^*(\delta,\epsilon)}=\inf_{\vartheta\in\Theta}I(\vartheta)\int_{\Theta}\frac{\pi(\vartheta)}{I(\vartheta)}\,\mathrm{d}\vartheta\le1\, ,
	\end{equation}
since $I(\vartheta)^{-1}\leq (\inf\limits_{\omega\in\Theta}I(\omega))^{-1}$ for every $\vartheta\in\Theta$.

Hence, Bayesian stopping requires, on average, no more observations than a deterministic fixed-sample strategy designed for the least informative parameter value. 
The improvement arises from adapting the stopping time to the local Fisher information through the posterior distribution, while preserving the same asymptotic dependence on the resolution 
$\delta$ and the error $\epsilon$. 

\section{Phase estimation}\label{app:phase}

In this appendix, we apply our \textit{twin-peaks test}  to the problem of quantum phase estimation. In \secref{app:KL_for_phase}, we evaluate the Kullback--Leibler divergence and Fisher information for 
several phase-estimation strategies, obtaining explicit predictions for the asymptotic stopping times of the \textit{twin-peaks test}. In \secref{app:phase_precision} we numerically investigate the performance of the 
\text{twin-peaks test} under various measurement strategies. Finally, in \secref{app:entangled_probes}, we numerically compare the performance of the \textit{twin-peaks test} using entangled probe states.

\subsection{KL divergence for the covariant phase measurement}\label{app:KL_for_phase}

We consider a pure equatorial qubit state, $\ket{\theta} = \frac{1}{\sqrt{2}}(\ket{0}+e^{i\theta}\ket{1})$, and evaluate the classical Kullback-Leibler divergence between the 
measurement outcome distributions associated with parameters $\theta$ and $\theta+\delta$ under the canonical phase-covariant POVM. 
For a single copy, the latter has density
	\begin{equation}
    		E(\phi) = \frac{\dd\phi}{2\pi}\ket{\phi}\!\bra{\phi}\, ,
	\label{app:covariant_POVM}
	\end{equation}
and the corresponding probability density is
	\begin{equation}
    		p(\phi\mid\theta) = \frac{1}{2\pi}\left|\,\langle{\phi}|\,\theta\rangle\right|\,^2 =\frac{1+\cos(\phi-\theta)}{2\pi}\, .
	\label{app:prob_covariant_POVM}
	\end{equation}
The Kullback-Leibler divergence is then
	\begin{align}\nonumber
   		D(p(\phi\mid\theta)\|\, p(\phi\mid\theta+\delta)) &= \int_0^{2\pi}p(\phi\mid\theta)\log \frac{p(\phi\mid\theta)}{p(\phi\mid\theta+\delta)}\,\dd\phi\\
   		&= \frac{1}{2\pi}   \int_0^{2\pi}(1+\cos t)\left(\log(1+\cos t)-\log(1+\cos(t+\delta))\right)\,\dd t
	\label{app:KL_covariantPOVM}
	\end{align}
where we performed the change of variable $t = \phi-\theta$. 

As the Fourier expansion of $\log(1+\cos t)$
	\begin{equation}
    		\log(1+\cos t) = -\log 2 + 2\sum_{k=1}^\infty\frac{(-1)^{k+1}}{k}\cos(kt)\, 
	\label{app:log_FT_final1}
	\end{equation}
contains only the zeroth and first Fourier modes, only the constant term and the first harmonic contribute to the integral of \eqnref{app:log_FT_final1}. Therefore,
	\begin{align}\nonumber
    		\frac{1}{2\pi}\int_0^{2\pi}(1+\cos t)\log(1+\cos t)\,dt&=\frac{1}{2\pi}\int_0^{2\pi}(1+\cos t)(-\log2+2\cos t)\dd t \nonumber\\ \nonumber
    		&=-\log2+\frac{2}{2\pi}\int_0^{2\pi}(1+\cos t)\cos t\,\dd t\\
    		&= -\log2+1\, .
		\label{app:KL_covariantPOVM_integral1}
		\end{align}
Similarly, using
	\begin{equation}
    		\log(1+\cos(t+\delta))=-\log2+2\sum_{k=1}^\infty\frac{(-1)^{k+1}}{k}\cos(k(t+\delta))\, ,
	\label{app:log_FT_final2}
	\end{equation}
we find
	\begin{equation}
    		\frac{1}{2\pi}\int_0^{2\pi}(1+\cos t)\log(1+\cos(t+\delta))\,\dd t = -\log2+\cos d.
	\label{app:KL_covariantPOVM_integral2}
	\end{equation}
Subtracting the two contributions gives
	\begin{equation}
    		D(p(\phi\mid \theta)\|\, p(\phi\mid \theta+\delta)) = 1-\cos\delta= 2\sin^2\frac{\delta}{2}.
	\label{app:KL_covariantPOVM_final}
	\end{equation}

\subsection{Evaluating the twin-peaks test}\label{app:phase_precision}

In Fig.~\ref{fig:phase_point_avg_samples}, we compare the average number of samples between different measurement strategies for the estimation of an unknown phase $\theta$ with the \textit{twin-peaks test}. We 
see that the 4-outcome POVM (blue), equivalent to randomly sampling between two measurement directions perpendicular in the equatorial Bloch plane, performs worse in the number of samples. Already, the 8-outcome POVM (orange), corresponding to 4 evenly spaced measurements on the equatorial plane, performs visibly better, similar to a random sampling in the measurement direction (green). All these measurements 
are i.i.d. strategies. In contrast, the highly non-i.i.d. adaptive strategy (red) performs significantly better. 

We also compare the performance of the test in terms of precision in the estimate. In Fig.~\ref{fig:precision_test_phase}, we compare the mean squared error (left image) and the frequency of estimates at a distance 
larger than $\delta$ from the true value (right image). While the 8-outcome POVM and the random PVM perform very similarly for both figures of merit, the 4-outcome POVM performs worse. Interestingly, the adaptive 
strategy performs significantly worse with respect to the MSE. This illustrates that optimizing the MSE is not equivalent to optimizing the stopping criterion of the \textit{twin-peaks test}. The latter is 
determined by the probability that the estimate lies farther than $\delta$ from the true value, under which the adaptive strategy performs much better than the 4-outcome POVM and only slightly worse than the 8-
outcome POVM and the random PVM, while requiring fewer samples.

\begin{figure}[t]
    \centering
    \hspace{-25pt}
    \includegraphics[width=0.45\columnwidth]{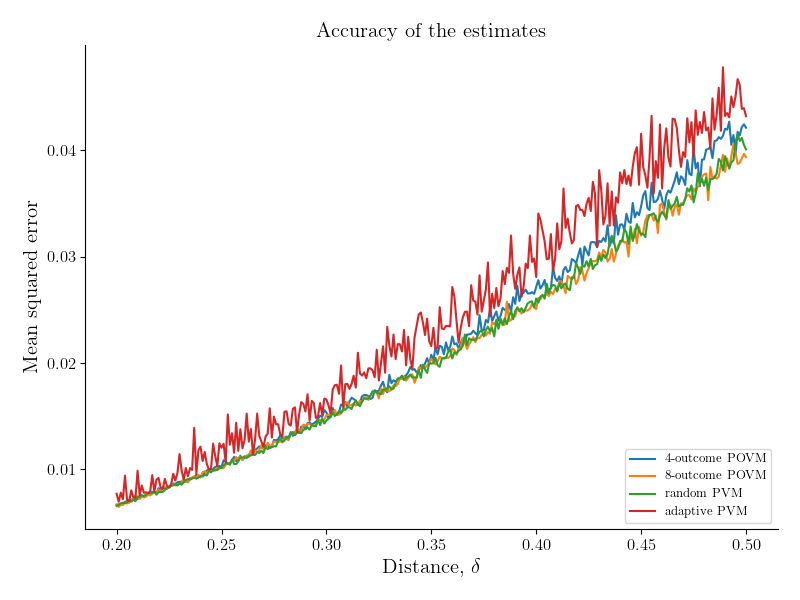}
    \includegraphics[width=0.45\columnwidth]{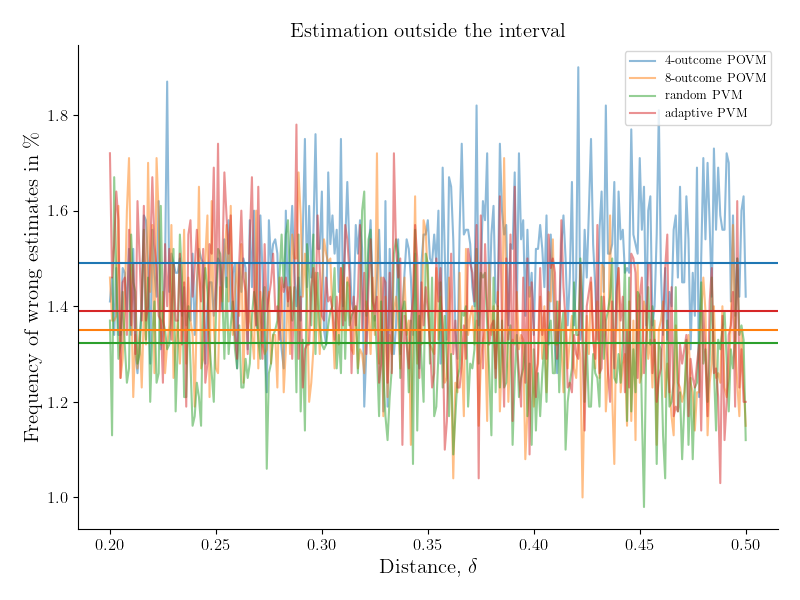}
    \caption{Precision of the different estimation strategies explained in the main text. The left image shows the mean squared error vs. the distance $\delta$, and the right image shows the frequency of estimates 
    outside the interval of distance $2\delta$ vs. the distance $\delta$ and their mean over all values of $\delta$ as horizontal lines in the same color. We use an error of $\epsilon=5\%$.}
    \label{fig:precision_test_phase}
\end{figure}

\subsection{Entangled strategies}\label{app:entangled_probes}

So far, we have treated the phase estimation problem as a state discrimination task between a continuous family of states. In the more general 
framework of quantum parameter estimation, however, the probe state itself becomes part of the estimation strategy. Specifically, one prepares a probe, lets it undergo a parameter-dependent transformation, and then 
performs a measurement on the resulting state. The freedom to optimize the probe state introduces additional resources, most notably entanglement, which can significantly enhance estimation performance.

As a paradigmatic example, consider the estimation of the parameter $\theta$ describing the unitary transformation
	\begin{equation}
		U(\theta)=\ketbra{0}{0}+e^{\mathrm{i} \theta}\ketbra{1}{1}.
	\label{app:unitaries}
	\end{equation}
If only separable probe states are available, the optimal choice is the equatorial state, $\ket{+}=(\ket{0}+\ket{1})/\sqrt{2}$, which, after the action of $U(\theta)$, produces the one-parameter family
$\ket{\theta}=(|\,0\rangle+\mathrm{e}^{i\theta}|\,1\rangle)/\sqrt{2}$ considered in the previous sections. Allowing entangled probes, however, enlarges the set of achievable likelihood functions and can 
substantially improve estimation performance. In the remainder of this section, we compare several entangled probe states using the same \textit{twin-peaks test}, stopping criterion, and FOM 
introduced previously. Throughout, we assume a uniform prior, $\pi(\theta)=1/(2\pi)$, and employ the collective covariant phase POVM.
 
If the FOM is chosen to be the Holevo variance~\cite{Holevo1984}
	\begin{align}
    		V_{\text{Ho}}=|\,\langle\mathrm{e}^{i\theta}\rangle|\,^{-2}-1\, ,
	\label{app:Holevovariance}
	\end{align}
then the optimal $N$-qubit state is the so-called sine-state
	\begin{align}\label{eq:sine_state}
    		\ket{\psi_{\text{sine}}}=\sum\limits_{n=0}^N \sqrt{\frac{2}{N+2}}\sin{\frac{(n+1)\pi}{N+2}}\ket{n},
	\end{align}
where $\ket{n}=\ket{1}^{\otimes n}\ket{0}^{\otimes N-n}$.  In contrast, if the FOM is chosen to be the maximum likelihood, then the optimal state is  the flat state 
	\begin{align}
    		\ket{\psi_{\text{flat}}}=\sum\limits_{n=0}^N \frac{1}{\sqrt{N+1}}\ket{n}\, .
	\label{eq:flat_state}
	\end{align}
For both probe states, the optimal measurement remains the covariant phase POVM, which corresponds to a phase-randomized measurement in the Fourier basis~\cite{Holevo1984,Bartlett_2007}.

The corresponding likelihood functions for obtaining the outcome $\phi$ are
	\begin{align}\nonumber
    		|\,\braket{\phi\mid \psi_{\text{sine}}(\theta)}|\,^2&=\frac{2}{(N+1)(N+2)}\left|\,  \sum _{n=0}^{N} \sin{\frac{(n+1)\pi}{N+2}} e^{i n(\theta-\phi)} \right|\,^2\\
    		|\,\braket{\phi\mid \psi_{\text{flat}}(\theta)}|\,^2&=\frac{1}{(N+1)^2}\left|\,  \sum _{n=0}^{N}  e^{i n(\theta-\phi)} \right|\,^2\, .
	\label{app:entangled_probs}
	\end{align}
Viewed as functions of $\theta$, both likelihoods depend on the measurement outcome $\phi$ only through the shift $\theta\mapsto\theta-\phi$. Consequently, the shape of the likelihood function is 
independent of the measurement outcome, implying that the stopping condition \eqnref{eq:stopping_condition_twin_peaks} is either satisfied or not satisfied irrespective of the observed value of $\phi$. 
We may therefore determine directly the minimum number of probe qubits required to satisfy the stopping criterion.

The above two probe states are not \textit{a priori} optimal for the \textit{twin-peaks} test, since its test statistic differs from both the Holevo variance and the maximum likelihood criterion. 
To determine the optimal probe state, we therefore consider the general $N$-qubit state $\ket{\psi}=\sum\limits_{n=0}^N \psi_n \ket{n}$ for which the likelihood of obtaining the outcome $\phi$ is
	\begin{align}
    		|\,\braket{\phi\mid \psi(\theta)}|\,^2&= \frac{1}{N+1}\left|\,  \sum _{n=0}^{N} \psi_n e^{i n(\theta-\phi)} \right|\,^2\, .
	\label{app:generalstate_likelihood}
	\end{align}
As before, the measurement outcome merely translates the likelihood function without altering its shape. Consequently, we may set $\phi=0$ without loss of generality.  Furthermore, for $\psi_n$ non-negative,  
the likelihood of \eqnref{app:generalstate_likelihood} is maximized at $\theta=0$. The optimal probe state for the \textit{twin-peaks test} is therefore obtained by numerically solving
	\begin{align}
    		\max\limits_{\psi_n}\min\limits_{\theta\in(d,2\pi-\delta)}\frac{\left|\, \sum _{n=0}^{N} \psi_n \right|\,^2}{\left|\, \sum _{n=0}^{N} \psi_n e^{i \theta  n} \right|\,^2}.
	\end{align}
The resulting optimal probe, together with the sine and flat states discussed above, is compared in Fig.~\ref{fig:entangled_strategies}.

\begin{figure}[t]
    \centering
    \hspace{-25pt}
    \includegraphics[width=0.45\columnwidth]{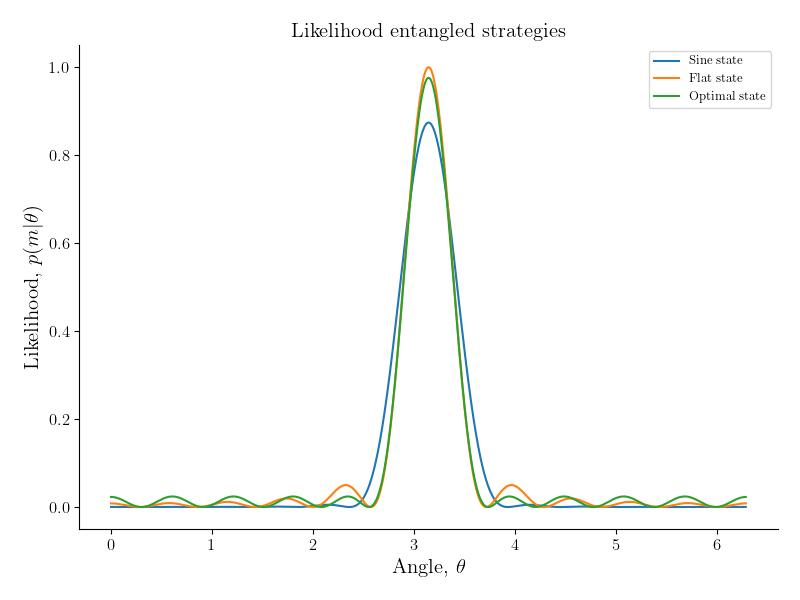}
    \includegraphics[width=0.45\columnwidth]{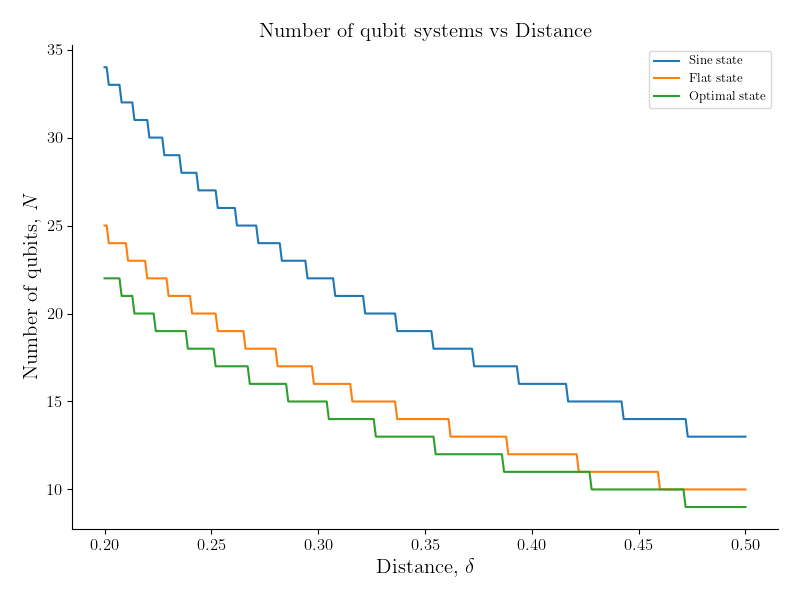}
    \caption{The performance of entangled strategies. In blue the sine state from Eq.~\eqref{eq:sine_state}, in orange the flat state from Eq.~\eqref{eq:flat_state} and in green the numerically optimized state 
    for outcome $\phi=\pi$. The left image shows the likelihood $p(m\mid \theta)$ for the three discussed entangled states for $N=10$, where the green line shows the state optimized for $d=0.5$. The right image 
    shows the number of qubits for the three entangled strategies to fulfil the stopping condition of the sequential strategy with respect to the distance $\delta$ and $\epsilon=0.05$.}
    \label{fig:entangled_strategies}
\end{figure}
\section{Purity estimation using collective measurement strategies}\label{app:purity_collective}

In this appendix, we consider collective measurement strategies for testing the purity of a qubit mixed state.  In \secref{sec:purity} we considered our \textit{twin-peaks test} employing a local PVM for 
estimating the purity. However, in doing so, we assumed that the direction of the Bloch vector is known, begging the question of whether the purity can be tested without 
prior knowledge of the Bloch vector direction.  In fixed-sample protocols, one typically sacrifices a fraction of the samples to first estimate the direction and 
then the remaining samples to estimate the purity by performing a PVM along the estimate, $\hat{\bf r}$, of the direction.  Assuming that the fraction of samples used for estimating the direction scales sub-linearly with $n$, the sample cost for estimating the direction tends to zero in the limit of a large number of samples. The collective measurement approach developed below removes the need for this preliminary estimation step.

An alternative route is to perform collective measurements on multiple copies of the state.  The most ubiquitous example is the two-copy SWAP test~\cite{Barenco1997}---which can be implemented optically 
via the Hong-Ou-Mandel effect~\cite{HOM1987}---and which finds numerous applications from quantum fingerprinting~\cite{Buhrman2001}, 
estimation of non-linear functionals of a quantum state~\cite{Ekert2002}, through to super-resolving spatial~\cite{Parniak2018, Aguilar2020, Ndagano2022}, temporal~\cite{Lyons2018, Ansari2021}, and 
spectral~\cite{Mazelanik2022} imaging.  Let $\ket{\psi}, \ket{\phi}\in\mathcal{H}$; the SWAP operator is defined as 
	\begin{equation}
		\mathrm{SWAP} \ket{\psi}\otimes\ket{\phi}=\ket{\phi}\otimes\ket{\psi}\, .
	\label{app:SWAP}
	\end{equation}
and has spectral decomposition 
	\begin{equation}
		\mathrm{SWAP}=\lambda_S \,\Pi_S +\lambda_A\, \Pi_A\, ,
	\label{app:spectralSWAP}
	\end{equation}
where $\Pi_A=\ketbra{\Psi^-}{\Psi^-}$ is the projector onto the anti-symmetric subspace spanned by $\ket{\Psi^-}=\frac{1}{\sqrt{2}}(\ket{01}-\ket{10})$, with eigenvalue $\lambda_A=-1$, and $\Pi_S=\one-\Pi_A$ the projector 
onto the symmetric subspace with eigenvalue $\lambda_S=1$.

Noting that for a qubit density matrix $\rho$ it holds
	\begin{equation}
		\tr\left(\mathrm{SWAP}(\rho\otimes\rho)\right)=\tr\rho^2\, ,
	\label{app:expectedSWAP}
	\end{equation}   
repeating the SWAP measurement on $N$ i.i.d. pairs of qubits in the state $\rho\otimes\rho$ yields the estimator 
	\begin{equation}
		\hat{\xi}=\frac{1}{N}\sum_{i=1}^N\lambda_i,
	\label{app:estimator}
	\end{equation}
which provides an estimate of $\rho^2$ with corresponding variance
	\begin{equation}
		\sigma(\hat{\xi})=\frac{1-\tr\rho^2}{N}\, .
	\label{app:variance}
	\end{equation}
Using the relation 
	\begin{equation}
		r = \sqrt{2\xi-1}
	\end{equation}
one obtains an estimator of the Bloch vector length by standard error propagation.

Now consider a sequential version of the SWAP test using our \textit{twin-peaks} test.  Since the SWAP test is nothing more than the weak Schur-sampling measurement of \eqnref{eq:Schur_sampling} 
with $M=2$,  the conditional probabilities for obtaining the symmetric ($J=1$) or anti-symmetric ($J=0$) outcome are given by 
	\begin{equation}
		\begin{split}
			p_2(J=0\mid r) &= \frac{1-r^2}{4}\\
			p_2(J=1\mid r) &=\frac{3+r^2}{4}\, .
		\end{split}
	\label{app:probsN2}
	\end{equation}
The average sample size $\mathbb{E}(k\mid r)$ of a sequential SWAP test, in the limit of no overshooting of the boundaries, reads
	\begin{equation}
		\mathbb{E}(k\mid r) = \frac{-2\log\epsilon}{d^2 I_2(r)}
	\label{app:expected_sample_size2}
	\end{equation}	
where $I_2(r)$ is the Fisher information of the conditional probability distribution of \eqnref{app:probsN2}. For the case at hand it reads
\begin{equation}
    I_2(r)= \frac{4 r^2}{(1-r^2)(3+r^2)}\ \,
    \label{app:FisherN2_rmt}
\end{equation}
and can be expressed as 
	\begin{equation}
		I_2(r) = 2 \left(\mathcal{I}(\rho)-\frac{3-r^2}{(1-r^2)(3+r^2)}\right)\, ,
	\label{app:FisherN2new}
	\end{equation} 
where 	
	\begin{equation}
		\mathcal{I}(\rho)=\frac{1}{1-r^2}
	\label{eq:QFI}
	\end{equation}
is the QFI information obtained by performing a single-qubit PVM along the Bloch vector direction.  The factor of 2 in \eqnref{app:FisherN2new} is due to the 
fact that our sequential test employs pairs of qubit states. Plugging \eqnref{app:FisherN2new} into \eqnref{app:expected_sample_size2} results in the \textit{total}
average number of samples
	\begin{equation}
		\mathbb{E}(N^{(2)}\mid r) = \frac{-2\log\epsilon}{d^2 \left(\mathcal{I}(\rho)-\frac{3-r^2}{(1-r^2)(3+r^2)}\right)}\, ,
	\label{app:AvsamplesN2}
	\end{equation}			
where the superscript 2 on the left-hand side is there to remind us that this is the average number of samples when we employ a sequential test with block size equal to 
2.

The same analysis extends directly to weak Schur measurements acting on blocks of $M$ copies. In this case, the Fisher information of the
measurement outcomes takes the form
	\begin{equation}
		I_M(r)  = M(\mathcal{I}(\rho)- g_M(r))\, ,
	\label{app:FisherNM}
	\end{equation} 
where $g_M(r)$ is a ratio of polynomial functions of degree at most $2M$.  Consequently, the \textit{total} average number of samples employed by such a collective \textit{twin-peaks} test 
is given by 
	\begin{equation}
		\mathbb{E}(N^{(M)}\mid r) = \frac{-2\log\epsilon}{d^2 \left(\mathcal{I}(\rho)-g_M(r)\right)}\, .
	\label{app:AvsamplesNM}
	\end{equation}
In \figref{fig:ratio_samples}, we compare the ratio between the total average number of samples used by a collective \textit{twin-peaks} test of block size $M=2,4,8$.  The total average sample size of a collective 
\textit{twin-peaks test} of block size $M-1$ requires more average samples than one of block size $M$.  Moreover, it is known that as the block size becomes asymptotically large, 
the Fisher information rate of the Schur measurement approaches the QFI~\cite{Guta2006}, i.e., 
	\begin{equation}
		\lim_{M\to\infty} \frac{I_M(r)}{M}=\mathcal{I}(\rho)\, .
	\label{app:Fisherasymptotics}
	\end{equation}		
Hence, a collective \textit{twin-peaks test} for testing the purity of qubit states without knowing their Bloch vector direction requires \textit{the same} number of samples on average 
as the \textit{twin-peaks} test that assumes the Bloch vector direction is known, i.e., 
	\begin{equation}
		\lim_{M\to\infty}\mathbb{E}(N^{(M)}\mid r)= \mathbb{E}(N\mid r).
	\label{app:asymptotic_avg_samples}
	\end{equation}
\begin{figure}[ht]
    \centering
    \hspace{-15pt}
    \includegraphics[width=0.45\columnwidth]{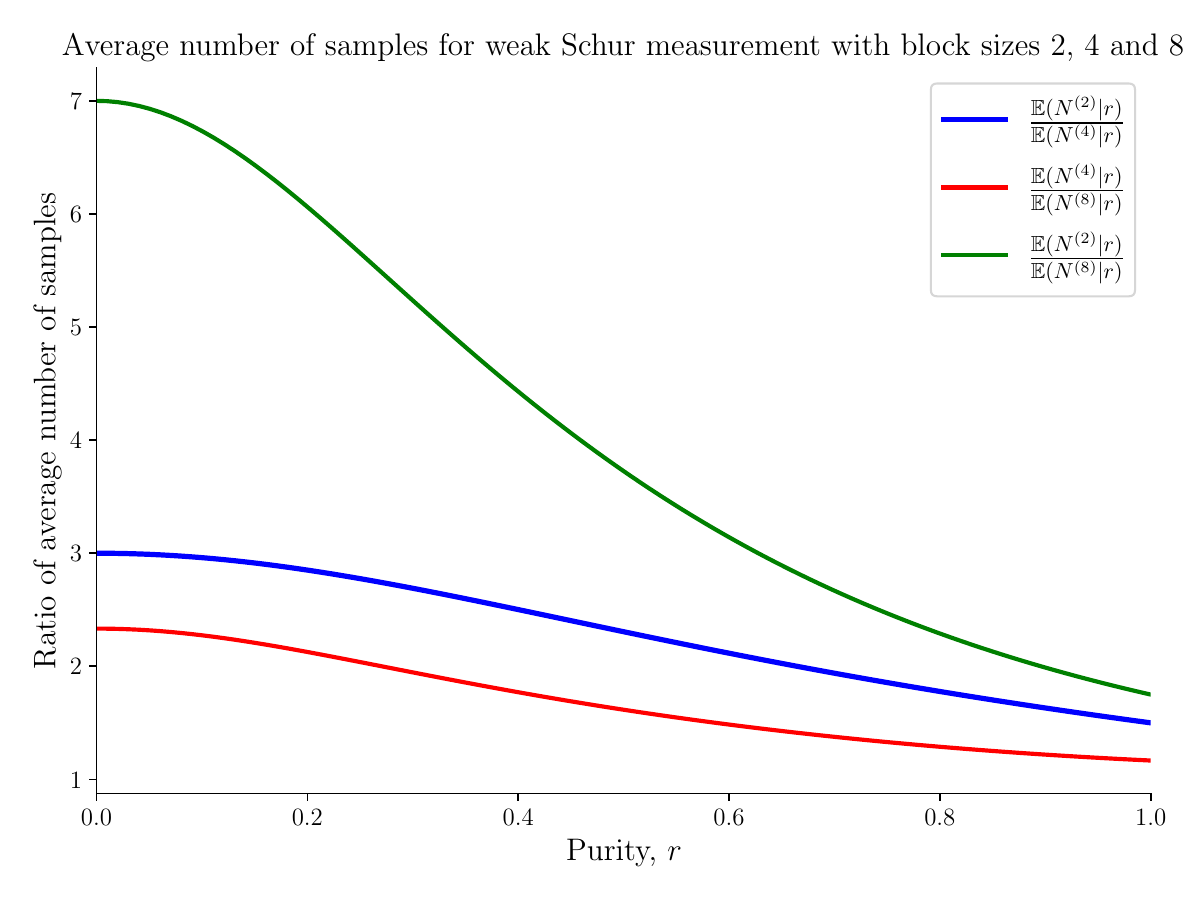}
    \caption{Ratio of the total average number of samples for the \textit{twin-peaks} test using the weak Schur measurement with block sizes $M=2,4$ and 8.}
    \label{fig:ratio_samples}
\end{figure}

An interesting question is how the block size $M$ should be chosen. While increasing $M$ always increases the Fisher-information rate,
larger blocks also require more copies before a measurement outcome is obtained. Consequently, the optimal block size depends on the unknown
purity itself, leading to a trade-off between information gained per measurement and the cost of each measurement.
Near the maximally mixed state 
	\begin{equation}
		\frac{\mathbb{E}(N\mid r)}{\mathbb{E}(N^{(M)}\mid r)}=\frac{2M(M-1)}{3}\, r^2 +\mathcal{O}(r^3)\, .
	\label{app:close_to_mixed}
	\end{equation}
Hence, increasing the block size is beneficial as the expected stopping time grows quadratically with $M$. However, this expansion remains valid only while
$Mr=\mathcal{O}(1)$, meaning that $M_{\rm opt}\sim\frac1r$. Therefore, highly mixed states benefit from increasingly collective measurements. 

On the other hand, for nearly pure states,
	 \begin{equation}
		\frac{\mathbb{E}(N\mid r)}{\mathbb{E}(N^{(M)}\mid r)}=M-3 +2r +\mathcal{O}((r-1)^2)\, ,
	\label{app:close_to_pure}
	\end{equation}
showing that the additional benefit obtained by increasing the block size rapidly diminishes. As the Fisher information rate is already close to the QFI, larger collective measurements yield only marginal reductions in stopping time.

\end{document}